\documentclass{lmcs} %%% last changed 2014-08-20
\pdfoutput=1
\usepackage[utf8]{inputenc}

\usepackage{lastpage}
\lmcsdoi{19}{3}{16}
\lmcsheading{}{\pageref{LastPage}}{}{}%
{Nov.~29,~2021}{Sep.~14,~2023}{}

\keywords{tuple generating dependencies, ontology mediated queries, counting answers, parametrised complexity, approximation schemes}

\usepackage{hyperref}

\usepackage{xspace}
\usepackage{amsfonts}
\usepackage{amsmath, amssymb}
\usepackage{amsthm}
\usepackage{thmtools,thm-restate}

\usepackage{misc}   %custom for logic
\usepackage{color}
\usepackage{graphicx}

\usepackage[shortcuts]{extdash} %for \=/

\usepackage{tikz}
\usepackage{tikz-cd}
\usetikzlibrary{calc}

\usepackage{subcaption}
\makeatletter
\DeclareFontEncoding{LS1}{}{}
\DeclareFontSubstitution{LS1}{stix}{m}{n}
\makeatother

\theoremstyle{plain} %\crefname{satz}{Satz}{S\"atze}
\theoremstyle{plain}\newtheorem{lem*}[thm]{Lemma}
\newcommand{\fun}[3]{\ensuremath{#1\colon #2 \to #3}}

\newcommand{\eqdef}{:=}
\renewcommand{\dom}{\ensuremath{\mn{adom}}}

\newcommand{\maxarity}{\mn{ar}(\Sbf)}

\newcommand{\OMQs}{OMQs\xspace}

\newcommand{\chase}{\mn{ch}}
\newcommand{\core}{\mn{core}}
\newcommand{\marked}{m}
\newcommand{\newschema}{s}
\newcommand{\fpt}{FPT\xspace}
\newcommand{\wOne}{\textnormal{\sc W[1]}\xspace}
\newcommand{\wOneC}{\#\textnormal{\sc W[1]}\xspace}
\newcommand{\wTwoC}{\#\textnormal{\sc W[2]}\xspace}
\newcommand{\aTwoC}{\#\textnormal{\sc A\![2]}\xspace}

\newcommand{\hideMe}[1]{}

\newcommand{\highlightColour}{black}
\newcommand{\highlightColourTwo}{black}

\begin{document}

%opening

\title[Answer Counting under Guarded TGDs]{Answer Counting under Guarded TGDs}

% \titlecomment{{\lsuper*}This article is an extended version of the conference paper
%     \cite{DBLP:conf/icdt/FeierLP21}}

%% etc.

%% required for running head on odd and even pages, use suitable
%% abbreviations in case of long titles and many authors:

%%%%%%%%%%%%%%%%%%%%%%%%%%%%%%%%%%%%%%%%%%%%%%%%%%%%%%%%%%%%%%%%%%%%%%%%%%%

%% the abstract has to PRECEDE the command \maketitle:
%% be sure not to issue the \maketitle command twice!

\begin{abstract}
  We study the complexity of answer counting for ontology-mediated
  queries and for querying under constraints, considering conjunctive
  queries and unions thereof (UCQs) as the query language and guarded
  TGDs as the ontology and constraint language, respectively.  Our
  main result is a classification according to whether answer counting
  is fixed-parameter tractable (FPT), \wOne-equivalent,
  \wOneC-equivalent, \wTwoC-hard, or \aTwoC-equivalent, lifting a
  recent classification for UCQs without ontologies and constraints
  due to Dell et al.~\cite{DBLP:conf/icalp/DellRW19}. The
  classification pertains to various structural measures of queries, namely
  treewidth, contract treewidth, starsize, and linked matching number.
  Our
  results rest
  on the assumption that the arity of relation symbols is bounded by a
  constant and, in the case of~ontology-mediated querying, that all
  symbols from the ontology and query can occur in the data (so-called
  full data schema).
  We also study the meta-problems for the mentioned structural
  measures, that is, to decide
  whether a given ontology-mediated query or constraint-query
  specification is equivalent to one for which the structural
  measure
  is bounded. % If the actual query is a
  % conjunctive query, this enables approximate answer counting in FPT
  % even when the original query is hard to count.
\end{abstract}

\author[C.~Feier]{Cristina Feier}[a]	%required
\address{Department of Computer Science, University of Warsaw, Poland}	%required
%\email{feier@uni-bremen.de}  %optional
%\thanks{thanks 1, optional.}	%optional

\author[C.~Lutz]{Carsten Lutz$^{\ast}$}[b]	%optional
\address{Institute of Computer Science, Leipzig University, Germany}	%optional
%\email{clu@informatik.uni-leipzig.de}  %optional
\thanks{$^\ast$The author was supported by the DFG project LU 1417/3-1 QTEC}	%optional

\author[M.~Przyby\l{}ko]{Marcin Przyby\l{}ko}[\!\ast\; b]	%optional
%\address{Institute of Computer Science, Leipzig University, Germany}	%optional
%\email{przybyl@uni-bremen.de}  %optional
%\thanks{The author was supported by the DFG project LU 1417/3-1 QTEC}%\thanks{thanks 3, optional.}	%optional

\maketitle

%\tableofcontents

%introduction
%\section*{TODO}
%\begin{enumerate}
%    \item move text
%    \item no restatable, no appendix
%    \item no subfigure -- fix some figures
%    \item remember that lmcs uses thm, lem, defi, etc\dots
%    \item does the reduction work for w/o bounded arity?
%\end{enumerate}

\section{Introduction}
\label{sec:intro}

Tuple-generating dependencies (TGDs) are a prominent formalism for
formulating database constraints. A TGD states that if certain facts
are true, then certain other facts must be true as well. This can be
interpreted in different ways. In \emph{ontology-mediated querying},
TGDs give rise to ontology languages and are used to derive new facts
in addition to those that are present in the database. This makes it
possible to obtain additional answers if the data is~incomplete and
also enriches the vocabulary that is available for querying. In a more
classical setup that we refer to as \emph{querying under constraints},
TGDs are used as integrity constraints on the database, that is, a TGD
expresses the promise that if certain facts are present in the
database, then certain other facts are present as well. Integrity
constraints are relevant to query optimization as their presence might
enable the reformulation of a~query into a~`simpler' one. TGDs
generalize a~wide range of other integrity constraints such as
referential integrity constraints (also known as inclusion
dependencies), which was the original motivation for introducing
them~\cite{AbHV95}.

When unrestricted TGDs are used as an ontology language,
ontology-mediated querying is undecidable even for unary queries that
consist of a single atom \cite{CaGK13}. This has led to intense
research on identifying restricted forms of TGDs that regain
decidability, see for instance \cite{BLMS11,CaGK13,CaGP12,LMTV19} and references
therein. In this paper, we consider guardedness as a basic and robust
such restriction: a TGD is guarded if some body atom, the guard,
contains all body variables~\cite{CaGK13}. % Let $\class{G}$ denote
% the class of guarded TGDs.
Guarded TGDs are useful also for formalizing integrity constraints.
For example, inclusion dependencies are a special case of guarded
TGDs. % \cite{AbHV95}.

In what follows, an ontology-mediated query (OMQ) is a triple
$(\Omc,\Sbf,q)$ with \Omc a set of TGDs (the ontology), \Sbf a data
schema, and $Q$ a union of conjunctive queries (UCQ).  Note that \Sbf
contains the relation symbols that can be used in the data while both
the ontology and query can also use additional symbols.  We use
$(\class{G},\class{CQ})$ to denote the language of OMQs in which the
ontology \Omc is a set of guarded TGDs and where $q$ is a conjunctive
query (CQ), and likewise for $(\class{G},\class{UCQ})$.  For querying
under constraints, we consider constraint query specifications (CQSs)
of the form $(\Tmc,\Sbf,q)$ where \Tmc is a set of TGDs (the integrity
constraints) and $q$ is a query, both over schema~\Sbf. Overloading
notation, we use $(\class{G},\class{(U)CQ})$ also to denote the class
of CQSs in which the constraints are guarded TGDs and the queries are
(U)CQs; it will always be clear from the context whether
$(\class{G},\class{(U)CQ})$ denotes an OMQ language or a class of
CQSs.

While being decidable, both ontology-mediated querying and querying
under constraints with guarded TGDs is computationally intractable in
combined complexity. Let us make this precise for query evaluation,
which is the following problem: given a database $D$, a query $Q$, and
a candidate answer $\bar c$, decide whether $\bar c$ is indeed an
answer to $Q$ on $D$.  Evaluating OMQs from $(\class{G},\class{CQ})$
is \TwoExpTime-complete in combined complexity and the same holds for
$(\class{G},\class{UCQ})$~\cite{CaGK13}; in both cases, the complexity
drops to \ExpTime if the arity of relation symbols is bounded by a
constant. Query evaluation for CQSs from~$(\class{G},\class{CQ})$ and
$(\class{G},\class{UCQ})$ is \NPclass-complete.

In this article, we are interested in counting the number of answers to
OMQs and to  queries posed under integrity constraints,
% clu: a CQS is not a query, unlike an OMQ!
% CQSs,
with an
emphasis on the limits of efficiency from the viewpoint of
parameterized complexity theory. Counting the number of answers
is~important to inform the user when there are too many answers to
compute all of them, and it is supported by almost every data
management system. It is also a~fundamental operation in data
analytics and in~decision support where often the count is more
important than the actual answers. Despite its relevance, however, the
problem has received little attention in ontology-mediated
querying and querying under constraints. We refer to
\cite{DBLP:journals/ws/KostylevR15,DBLP:journals/ojsw/KostovK18,Meghyn-IJCAI20,Diego-IJCAI20,DBLP:conf/kr/BienvenuMT22} %underConstraints??
for notable exceptions regarding ontology-mediated querying
that, however, study different counting problems than the present
paper.

We equate efficiency with fixed-parameter tractability (FPT), the
parameter being the size of the OMQ and of the CQS,
respectively. Evaluation is \wOne-hard both for ontology-mediated
querying in $(\class{G},\class{(U)CQ})$ and for querying under
constraints in $(\class{G},\class{(U)CQ})$ \cite{BDFLP-PODS20}. These
lower bounds apply already to Boolean queries where evaluation and
counting coincide, and therefore answer counting is in general not
fixed-parameter tractable in the mentioned cases unless FPT =
\wOne. % \footnote{Evaluation and answer counting is the same problem
  % for Boolean queries.}
The main question that we address is: how can we characterize the
parameterized complexity of answer counting for classes of OMQs or
CQSs $\class{Q} \subseteq (\class{G},\class{(U)CQ})$ and, most
importantly, for which such classes $\class{Q}$ can we count answers
in FPT? The classes $\class{Q}$ will primarily be defined in terms of
structural measures of the (U)CQ, but will also take into account
the interplay between the ontology/constraints and the (U)CQ.  Note
that \PTime combined complexity, a (significant) strengthening of FPT,
cannot be obtained by structural restrictions on the UCQ in
ontology-mediated querying with $(\class{G},\class{(U)CQ})$ because
evaluating Boolean OMQs is \TwoExpTime-complete already for unary
single-atom queries. For querying under constraints, in contrast,
\PTime combined complexity is not excluded up-front and in the case
of query evaluation can in fact sometimes be attained in
$(\class{G},\class{(U)CQ})$ \cite{BaGP16,BFGP20}.

A seminal result due to Grohe states that a recursively enumerable
class $\class{Q}$ of CQs can be evaluated in FPT if and only if there
is a constant that bounds the treewidths of CQs in~$\class{Q}$, modulo
equivalence \cite{Grohe07}. Grohe considers only Boolean CQs, but it
is well-known that the result lifts to the non-Boolean case when the
treewidth of a CQ $q$ is taken to mean the treewidth of the Gaifman
graph of $q$ after dropping all answer variables, see for instance
\cite{BFLP19,BDFLP-PODS20}.  The result rests on the assumptions that
$\text{FPT} \neq \wOne$ and that the arity of relation symbols is
bounded by a constant, which we shall also assume throughout this
article. Grohe's result extends to UCQs in the expected way, that is,
the characterization for UCQs is in terms of the maximum treewidth of
the constituting CQs modulo equivalence, assuming w.l.o.g.\ that there
are no containment relations among the CQs.  An adaptation of Grohe's
proof was used by Dalmau and Jonsson to show that a class $\class{Q}$
of CQs without quantified variables admits answer counting in FPT if
and only if the treewidths of CQs in $\class{Q}$ is bounded by a
constant \cite{countingHomomorphisms}.
%                                required.
In a series of papers by Pichler and Skritek
\cite{DBLP:journals/jcss/PichlerS13}, Durand and Mengel
\cite{DBLP:journals/jcss/DurandM14,countingQueriesSructural}, Chen and
Mengel \cite{countingTrichotomy,countingPositiveQueries}, and Dell et
al.~\cite{DBLP:conf/icalp/DellRW19}, this was extended to a rather
detailed classification of the parameterized complexity of answer
counting for classes of CQs and UCQs that may contain both answer
variables and quantified variables. The characterization is based on
the structural measure of treewidth, which now refers to the entire
Gaifman graph including the answer variables. It also refers to the
additional measures of contract treewidth,
starsize,\footnote{The measure is called dominating starsize in
  \cite{DBLP:conf/icalp/DellRW19} and strict starsize in
  \cite{countingTrichotomy}. We only speak of starsize.  Note that
  this is not identical to the original notion of starsize from
  \cite{DBLP:journals/jcss/DurandM14,countingQueriesSructural}.} and
linked matching number. It links boundedness of these measures by a
constant, modulo equivalence, to the relevant complexities, which turn
out to be FPT, \wOne-equivalence, \wOneC-equivalence, \wTwoC-hardness,
and \aTwoC-equivalence. Here, we speak of `equivalence' rather than of
`completeness' to emphasize that hardness is defined in terms of
(parameterized counting) Turing (fpt-)reductions.

The main results of this article are classifications of the complexity
of answer counting for classes of ontology-mediated queries from
$(\class{G},\class{(U)CQ})$, assuming that the data schema contains
all symbols used in the ontology and query, and for classes of
constraint query specifications from $(\class{G},\class{(U)CQ})$.  Our
classifications parallel the one for the case without TGDs, involve
the same five complexities mentioned above, and link them to the same
structural measures. There is, however, a twist. The ontology interacts
with all of the mentioned structural measures in the sense that for
each measure, there is a class of CQs $\class{Q}$ and an ontology~\Omc
such that the measure is unbounded for $\class{Q}$ modulo
equivalence while there is a constant $k$ such that each OMQ
$(\Omc,\Sbf,q)$, $q \in \class{Q}$, is equivalent to an OMQ
$(\Omc,\Sbf,q')$ with the measure of $q'$ bounded by $k$. % In particular, answer counting is \aTwoC-equivalent for
% $\class{C}$, but it is in FPT for the described class of OMQs.
A similar effect can be observed for querying under constraints. We
can thus not expect to link the complexity of a class $\class{Q}$ of
OMQs to the structural measures of the actual queries in the
OMQs. Instead, we consider a certain class $\class{Q}^\ast$ of CQs
that we obtain from the OMQs in $\class{Q}$ by first rewriting away
the existential quantifiers in TGD heads in the ontology, then taking
the CQs that occur in the resulting OMQs, combining them conjunctively
guided by the inclusion-exclusion principle, next chasing them with
the ontology (which is a finite operation due to the first step), and
then taking the homomorphism core. The structural measures of
$\class{Q}^\ast$ turn out to determine the complexity of answer
counting for the original class of OMQs $\class{Q}$. Interestingly,
the same is also true for classes of constraint query specifications
and thus the characterizations for OMQs and for CQSs coincide. We in
fact establish the latter by mutual reduction between answer counting
for OMQs and answer counting for CQSs.

We also take a brief look at approximate counting. For CQs without
ontologies, significant progress has recently been made by Arenas et
al.\ \cite{counting-cq-approx} who show that a class of CQs
$\class{Q}$ admits a fully polynomial randomized approximation scheme
(FPRAS) if and only if there is a constant bound $k$ on the treewidth
of the queries in $\class{Q}$ where, as for exact counting, treewidth
refers to the Gaifman graph including answer variables. This result is
subject to the assumptions that $\wOne \neq \text{FPT}$,
$\text{P}=\text{BPP}$, the arity of relation symbols is bounded by a
constant, and for every $q \in \class{Q}$, there is a self-join free
$q' \in \class{Q}$ that has the same hypergraph. We observe that it is
not hard to derive from this the existence of a fixed-parameter
tractable randomized approximation scheme (FPTRAS) for classes of OMQs
$\class{Q} \subseteq (\class{G},\class{UCQ})$ that have bounded
treewidth modulo equivalence. We leave a matching lower bound as an
open problem. It is an interesting contrast that the condition for the
existence of an FPTRAS is much less intricate than that for exact
counting, in particular bypassing the class of CQs $\class{Q}^\ast$
mentioned above.

Inspired by our complexity classifications, we then proceed to
  study the meta problems to decide whether a given query is
  equivalent to a query in which some selected structural measures are
  small, and to construct the latter query if it exists.  We do this
  both for ontology-mediated queries and for queries under constraints,
  considering all four measures that are featured in the
  classifications (and sets thereof). We start with querying under
  constraints where we are able to obtain decidability results in all
  relevant cases.
% This leads, in particular, to an approximate
% version of answer counting in FPT: given a CQS $(\Tmc,\Sbf,q) \in
% (\class{G},\class{CQ})$, an \Sbf-database $D$, and a $k \geq 1$, we
% can count in FPT the answers on $D$ to the approximation of $q$ under
% \Tmc in terms of both treewidth $k$ and contract treewith $k$.
%
These results can also be applied to ontology-mediated querying
when (i)~the data schema contains all symbols used in the ontology and
query and (ii)~we require that the ontology used in the OMQ cannot
be replaced with a different one. % Also there,
% we then obtain an approximate version of answer counting in FPT.
For contract treewidth and starsize, we additionally show that it is
never necessary to modify the ontology to attain equivalent OMQs
with small measures, and we provide decidability results without
assumptions~(i) and~(ii). We also observe that treewidth behaves
differently in that modifying the ontology might result in smaller
measures. Deciding the meta problem for the measure of treewidth is
left open as an interesting and non-trivial open problem.

This article is an extended version of the conference paper
\cite{DBLP:conf/icdt/FeierLP21}.  Some proofs are deferred to the
appendix. Since we rely profoundly on (refinements of) results
due to Chen and Mengel
\cite{countingTrichotomy,countingPositiveQueries}, we also provide
in the appendix summaries of the proofs of those results.

% \bigskip
% \bigskip
% \bigskip

% The limits of efficiency of OMQ evaluation in
% $(\class{G},\class{(U)CQ})$ and of query evaluation under constraints
% in $(\class{G},\class{(U)CQ})$ have recently been analyzed in
% \cite{PODS} from the viewpoint of parameterized complexity, the
% parameter being the size of the OMQ and of the CQS, respectively.
% Query evaluation is \wOne-hard in both cases, and thus in general not
% fixed-parameter tractable. A main result of \cite{PODS} is that
% evaluation for a class of OMQs $\class{Q} \subseteq
% (\class{G},\class{(U)CQ})$ is in FPT if and only if there is a
% constant $k \geq 1$ such that every OMQ $Q=(\Omc,\Sbf,q) \in
% \class{Q}$ is equivalent to an OMQ $Q'=(\Omc',\Sbf,q') \in
% (\class{G},\class{(U)CQ})$ with each CQ in $q'$ of treewidth at
% most $k$; likewise, a class of CQSs $\class{C} \subseteq
% (\class{G},\class{(U)CQ})$ is in FPT if and only if there is a
% constant $k \geq 1$ such that every $q \in \class{C}$ is equivalent to
% a UCQ $q'$ in which every CQ has treewidth at most $k$.

\medskip

{\bf Related Work.} The complexity of ontology-mediated querying has
been a subject of intense study from various angles, see for example
\cite{DBLP:conf/rweb/BienvenuO15,DBLP:journals/tods/BienvenuCLW14,DBLP:journals/jods/PoggiLCGLR08}
and references therein.  The parameterized complexity of evaluating
ontology-mediated queries has been studied in
\cite{BFLP19,BDFLP-PODS20, feier-icdt-2022}. While \cite{BFLP19}
consider description logics such as $\mathcal{ELI}$ as the ontology
language, \cite{BDFLP-PODS20, feier-icdt-2022} focus on
$(\class{G},\class{UCQ})$. Query evaluation in FPT coincides with
bounded treewidth modulo equivalence when the arity of relation
symbols is bounded by a constant, unless $\text{FPT} = \wOne$
\cite{BDFLP-PODS20}. When there is no such bound, then it coincides
with bounded submodular width modulo equality unless the exponential
time hypothesis fails\cite{feier-icdt-2022}.
Counting in ontology-mediated querying has been considered in
\cite{DBLP:journals/ws/KostylevR15,DBLP:journals/ojsw/KostovK18,Meghyn-IJCAI20,Diego-IJCAI20,DBLP:conf/ijcai/BienvenuMT21,
  DBLP:conf/dlog/BienvenuMT21}. There, conjunctive queries are
equipped with dedicated counting variables and the focus is to decide,
given an OMQ $Q=(\Omc,\Sbf,q)$, an \Sbf-database $D$ and a $k \geq 0$
whether there is a model of $D$ and \Omc such that the homomorphisms
from $q$ to that model yield at least/at most $k$ bindings of the
counting variables.  The ontology languages studied are versions of
the description logic DL-Lite, with the exception of
\cite{DBLP:conf/dlog/BienvenuMT21} which studies the description logic
$\mathcal{ELI}$. These can all be viewed as guarded TGDs, up to a
certain syntactic normalization in the case of $\mathcal{ELI}$.

Query evaluation under constraints that are guarded TGDs has been
considered in \cite{BaGP16,BFGP20}. A main result is an~FPT upper
bound for CQs that have bounded generalized hypertreewidth modulo
equivalence. These papers also study the meta problem for querying
under constraints that are guarded TGDs and for the measure of
generalized hypertree width. A topic closely related to the evaluation
of queries under constraints is query containment under constraints,
see for example \cite{DBLP:conf/pods/CalvaneseGL98,JoKl84,
  DiegoFunctional16}. We are not aware that answer counting under
integrity constraints has been studied before.

%\cMP{Do we want Diego functional paper? \cite{DiegoFunctional16}}
%%% Local Variables:
%%% mode: latex
%%% TeX-master: "paper_main"
%%% End:

%preliminaries and basic definitions
\section{Preliminaries}
\label{sect:prelim}

% We consider the disjoint countably infinite sets $\Cbf$ and $\Vbf$ of {\em constants} and {\em variables}, respectively.
% %(used in queries and dependencies), respectively.
% We refer to constants and variables as {\em terms}.
For an integer $n
\geq 1$, we use $[n]$ to denote the set $\{1,\ldots,n\}$. To indicate
the cardinality of a set $S$, we may write $\#S$ or~$|S|$.

\medskip
\noindent
%\paragraph{Relational Databases.}
{\bf Relational Databases.}
A {\em schema} \Sbf
is a set of relation symbols $R$ with associated arity $\mn{ar}(R) \geq 0$.
We write $\mn{ar}(\Sbf)$ for
$\max_{R \in \Sbf} \{\mn{ar}(R)\}$.
% and call $\Sbf_f$ the \emph{full
%   schema}.
An {\em \Sbf-fact} is an expression of
the form $R(\bar c)$, where $R \in \Sbf$ and $\bar c$ is an
$\mn{ar}(R)$-tuple of constants. % from \Cbf.
An {\em $\Sbf$-instance} is a (possibly infinite) set of \Sbf-facts
and an {\em \Sbf-database} is a finite $\Sbf$-instance. We write
$\mn{adom}(I)$ for the set of constants in an instance $I$.
For a set $S \subseteq \mn{adom}(I)$, we denote by $I_{|S}$ the
restriction of $I$ to
facts that mention only constants from $S$.
%
%As usual, we write $\size{I}$ for the size of $I$.
%
A {\em homomorphism} from $I$ to an instance $J$ is a function $h :
\mn{adom}(I) \rightarrow \mn{adom}(J)$ such that $R(h(\bar c)) \in J$
for every $R(\bar c) \in I$ where $h(\bar c)$ means the component-wise
application of $h$.
% We write $I \rightarrow J$ for the fact that there is a homomorphism from $I$ to $J$.
A \emph{guarded set} in a database $D$ is a set $S \subseteq
\mn{adom}(D)$ such that all constants in $S$ jointly occur in a fact
in~$D$, possibly together with other constants.  With a maximal
guarded set, we mean a guarded set that is maximal regarding set
inclusion.

We next introduce some operations on instances that are used in the
paper.  An \emph{induced subinstance} of an \Sbf-instance $I$ is any
\Sbf-instance $I'$ obtained from $I$ by choosing a
$\Delta \subseteq \mn{adom}(I)$ and putting
$I'=\{ R(\bar c) \in I \mid \bar c \in \Delta^{\mn{ar}(R)}\}$. If $I$ and
$I'$
are finite, then we speak of an \emph{induced subdatabase}.
The
\emph{disjoint union} of two \Sbf-instances $I_1$ and $I_2$ with
$\mn{adom}(I_1) \cap \mn{adom}(I_2)=\emptyset$ is simply
$I_1 \cup I_2$.  The \emph{direct product} of two \Sbf-instances $I_1$
and $I_2$ is the \Sbf-instance $I$ with domain
$\mn{adom}(I)=\mn{adom}(I_1) \times \mn{adom}(I_2)$ defined as
$$I=\{R((a_1,b_1),\dots,(a_n,b_n)) \mid R(a_1,\dots,a_n) \in I_1 \text{
  and } R(b_1,\dots,b_n) \in I_2 \}.$$
An instance $I'$ is obtained from an instance $I$ by \emph{cloning
  constants} if $I' \supseteq I$ can be constructed by choosing
$c_1,\dots,c_n \in \mn{adom}(I)$ and positive integers
$m_1,\dots,m_n$, reserving fresh constants
$c^{i_1}_1,\dots,c^{i_n}_n$ with $1 \leq i_\ell \leq m_\ell$ for $1 \leq
\ell \leq n$, and adding to $I$ each atom $R(\bar c')$ that can be
obtained from some $R(\bar c) \in I$ by replacing each occurrence of
$c_i$, $1 \leq i \leq n$, with $c^j_i$ for some $j$ with $1 \leq j
\leq m_i$.

\medskip
\noindent
%\paragraph{CQs and UCQs.}
{\bf CQs and UCQs.}
A {\em conjunctive query} (CQ) $q(\bar x)$ over a schema $\Sbf$ is a first-order
formula of the form
$\exists \bar y \, \varphi(\bar x, \bar y)$
%\end{equation}
where $\bar x$ and $\bar y$ are disjoint tuples of variables and
$\varphi$ is a conjunction that may contain \emph{relational atoms}
$R_i(\bar x_i)$ with $R_i \in \Sbf$ and $\bar x_i$ a tuple of
variables of length $\mn{ar}(R_i)$ as well as \emph{equality atoms}
$x_1 = x_2$. The variables used in $\vp$ must be exactly those in
$\bar x$ and $\bar y$, and only variables from $\bar x$ may appear in
equality atoms. {\color{\highlightColour}We assume that $\bar x$ contains no
  repeated variables, which is w.l.o.g.\ due to the presence of
  equality atoms.}  With $\mn{var}(q)$, we denote the set of variables
that occur in $\bar x$ or in $\bar y$.  Whenever convenient, we
identify a conjunction of atoms with a set of atoms. When we are not
interested in order and multiplicity, we treat $\bar x$ as a set of
variables. {\color{\highlightColour}A CQ is \emph{equality-free} if it contains no
  equality atoms.}
%
%If $\bar x$ is empty, then $q$ is a \emph{Boolean CQ}.
%
%We write $\var{q}$ for the set of variables occurring in $q$. By abuse of notation, we may write $\alpha \in q$ to indicate that the atom $\alpha$ occurs in $q$.
%
%The evaluation of CQs is defined in terms of homomorphisms.
Note that we do not admit constants in CQs.\footnote{\color{\highlightColour}We believe that,
  in principle,
  our results can be adapted to the case with constants. This requires
  a suitable revision of the structural measures defined in
  Section~\ref{sect:notgds} as, for example, constants should not
  contribute to the treewidth of a CQ. Also, the results for CQs
  without ontologies that we build upon would first have to be
  extended to include constants.}  We write $\class{CQ}$ for the class
of all CQs.

Every CQ $q(\bar x)$ can be seen as a database $D_q$ in a natural way,
namely by dropping the existential quantifier prefix and the equality
atoms, and viewing variables as constants.
% $D_q$, known as the {\em
%   canonical database} of~$q$, obtained by dropping the existential
% quantifier prefix and the equality atoms, and viewing variables as
% constants.
%
%We may simply write $q$ instead of $D[q]$.
%
  A \emph{homomorphism}
$h$  from a CQ $q$ to an instance $I$ is a homomorphism from $D_q$ to
 $I$ such that $x=y \in q$ implies $h(x)=h(y)$.
 A tuple $\bar c \in \mn{adom}(I)^{|\bar x|}$ is an {\em answer} to
 $q$ on $I$ if there is a homomorphism $h$ from $q$ to $I$ with
 $h(\bar x) = \bar c$.

 % {\color{cyan} The presence of equality atoms brings some subtleties.
 %   In particular, unless $q(\bar x)$ is equality-free it is not
 %   guaranteed that there is a homomorphism from $q$ to $D_q$ that is
 %   the identity on $\bar x$.  We use $\tilde q(\bar y)$ to denote the
 %   CQ obtained from $q$ by removing all equality atoms and identifying
 %   any two variables $x_1,x_2$ with $x_1 = x_2 \in q$.  To achieve
 %   that the result is unique, we assume a fixed order $\prec$ on the
 %   variables and that the identification results in $x_1$ when
 %   $x_1 \prec x_2$.  With $D^\sim_q$, we mean the database
 %   $D_{\tilde q}$. There is a homomorphism from $q$ to $D^\sim_q$ that
 %   is the identity on~$\bar y$. We shall thus use $D^\sim_q$ in place
 %   of $D_q$ whenever this property is important to us.  We remark that
 %   $q$ and $\tilde q$ have the same number of answers on any database.}

 A {\em union of conjunctive queries} (UCQ) over a schema $\Sbf$ is a first-order formula of the form
$
q(\bar x) := q_1(\bar x) \vee \cdots \vee q_n(\bar x),
$
where $n \geq 1$, and $q_1(\bar x),\dots, q_n(\bar x)$ are CQs over
$\Sbf$.
We refer to the variables in $\bar x$ as the \emph{answer variables}
of $q$ and the {\em arity} of $q$ is defined as the number of its answer
variables. An example for a UCQ with two answer variables $x_1,x_2$ is $x_1=x_2 \vee \exists
y \, R(x_1,y) \wedge R(x_2, y)$.
 A tuple $\bar c \in \mn{adom}(I)^{|\bar x|}$ is an {\em answer} to
 $q$ on instance $I$ if it is an answer to $q_i$ on $i$, for some $i$
 with $1 \leq i \leq n$.
The \emph{evaluation} of $q$ on an instance $I$, denoted $q(I)$, is
the set of
all answers to $q$ on $I$.
A (U)CQ of arity zero is called {\em Boolean}. The only possible
answer to a Boolean query is the empty tuple. For a Boolean (U)CQ $q$,
we may write $I \models q$ if $q(I) = \{()\}$ and $I \not\models q$
otherwise. Note that all notions defined for UCQs also apply to CQs,
which are simply UCQs with a single disjunct.
We write $\class{UCQ}$ for the class of all UCQs.

Let $q_1(\bar x)$ and $q_2(\bar x)$ be two UCQs over the same schema
\Sbf. We say that $q_1$ is \emph{contained} in
 $q_2$, written $q_1 \subseteq_\Sbf  q_2$, if
 $q_1(D) \subseteq q_2(D)$ for every \Sbf-database $D$. Moreover, $q_1$
 and $q_2$ are \emph{equivalent}, written $q_1 \equiv_\Sbf q_2$, if
 $q_1 \subseteq_\Sbf q_2$ and $q_2 \subseteq_\Sbf q_1$.

% \smallskip
% \noindent
% \paragraph{Cores.}
 {\color{\highlightColour}
     We next define the important notion of a homomorphism
   core of a CQ $q(\bar x)$. The potential presence of equality atoms
   in $q$ brings some subtleties.  In particular, it is not guaranteed
   that there is a homomorphism from $q$ to $D_q$ that is the identity
   on $\bar x$. To address this issue, we resort to the database
   $D^\sim_q$ obtained from $D_q$ by identifying any
   constants/variables $x_1,x_2$ such that $x_1 = x_2 \in q$. For $V$
   the set of all variables from $\bar x$ that occur as constants in
   $D^\sim_q$, it is easy to see that there is a homomorphism from $q$
   to $D^\sim_q$ that is the identity on all variables in $V$. We say
   that $q$ is a \emph{core} if every homomorphism $h$ from $q$ to
   $D^\sim_q$ that is the identity on $V$ is surjective.  } Every CQ
 $q(\bar x)$ is equivalent to a CQ $p(\bar x)$ that is a core and can
 be obtained from $q$ by dropping atoms.  In~fact, $p$ is unique up to
 isomorphism and we call it \emph{the core of} $q$. For a UCQ $q$, we
 use $\mn{core}(q)$ to denote the disjunction whose disjuncts are the
 cores of the CQs in~$q$.

For a UCQ $q$, but also for any other syntactic object $q$, we use
$||q||$ to denote the number of symbols needed to write $q$ as
a word over a suitable alphabet.

Our main interest is in the complexity of counting the number of
answers. Every choice of a query language $\class{Q}$, such as
$\class{CQ}$ and $\class{UCQ}$, and a class of databases $\class{D}$
gives rise to the following answer counting problem:

\begin{center}
	\fbox{\begin{tabular}{ll}
			PROBLEM : & {\sf AnswerCount}$(\class{Q},\class{D})$
			\\INPUT : & A query $q
                                              \in
                                             \class{Q}$ over some schema
                                             \Sbf and
			an $\Sbf$-database $D \in \class{D}$
			\\
			OUTPUT : &  $\#q(D)$
	\end{tabular}}
\end{center}

\medskip
\noindent
Our main interest is in the parameterized version of the above problem
where we generally assume that the parameter is the size of the input
query, see below for more details. When $\class{D}$ is the class of
all databases, we simply write {\sf AnswerCount}$(\class{Q})$.

\medskip
\noindent
%\paragraph{TGDs, Guardedness.}
{\bf TGDs, Guardedness, Fullness}
A {\em
  tuple-generating dependency} (TGD) $T$ over $\Sbf$ is a first-order
sentence of the form
$ \forall \bar x \forall \bar y \, \big(\phi(\bar x,\bar y)
\rightarrow \exists \bar z \, \psi(\bar x,\bar z)\big) $ such that
$\exists \bar y \, \phi(\bar x,\bar y)$ and
$\exists \bar z \, \psi(\bar x,\bar z)$ are CQs without equality
atoms. As a special case, we also allow $\phi(\bar x,\bar y)$ to be
the empty conjunction, i.e.~logical truth, denoted by \mn{true}.
For simplicity, we write $T$ as
$\phi(\bar x,\bar y) \rightarrow \exists \bar z \, \psi(\bar x,\bar
z)$. We call $\phi$ and $\psi$ the {\em body} and {\em head} of $T$,
denoted $\mn{body}(T)$ and $\mn{head}(T)$, respectively.  An instance
$I$ over $\Sbf$ \emph{satisfies} $T$, denoted $I \models T$, if
$q_\phi(I) \subseteq q_\psi(I)$.  It {\em satisfies} a set of TGDs
$S$, denoted $I \models S$, if $I \models T$ for each $T \in S$. We
then also say that $I$ is a \emph{model} of $S$.
%We work with finite sets of TGDs.
We write $\class{TGD}$ to denote the class of all TGDs.

A TGD $T$ is {\em guarded} if $\mn{body}(T)$ is \mn{true} or there
exists an atom $\alpha$ in its body that contains all variables that
occur in $\mn{body}(T)$~\cite{CaGK13}. Such an atom $\alpha$ is a {\em
  guard} of $T$. While there may be multiple guard atoms in the body
of a TGD, we generally assume that one of them is chosen as the actual
guard and may thus speak of `the' guard atom.
We write $\class{G}$ for the class of guarded TGDs. A~TGD $T$ is
\emph{full} if the tuple $\bar z$ of variables is empty, that is,
it uses no existential quantification in the head. We use
$\class{FULL}$ to denote the class of full TGDs and shall often refer
to $\class{G} \cap \class{FULL}$, the class of TGDs that are both
guarded and full. Note that this class is essentially the class of
Datalog programs with guarded rule bodies.

\smallskip
\noindent
{\bf Ontology-Mediated Queries.}
An \emph{ontology} \Omc is a finite set of TGDs.  An \emph{ontology
  mediated query (OMQ)} takes the form $Q= (\Omc,\Sbf,q)$ where $\Omc$
is an ontology, $\Sbf$ is a finite schema called the \emph{data
  schema}, and $q$ is a UCQ.  Both \Omc and $q$ can use symbols from
\Sbf, but also additional symbols, and in particular \Omc can
`introduce' additional symbols to enrich the vocabulary available for
querying.  We assume w.l.o.g.\ that all relation symbols in $q$ that
are not from \Sbf occur also in~\Omc.
% This can always be achieved by introducing dummy TGDs $R(\bar x)
% \rightarrow R(\bar x)$.
In fact, any OMQ violating this condition is trivial in that it never
returns any answers.
When \Omc and $q$ only use symbols from \Sbf, then we say that the
data schema of $Q$ is \emph{full}.  The {\em arity} of $Q$ is defined
as the arity of $q$. We write $Q(\bar x)$ to emphasize that the answer
variables of $q$ are $\bar x$ and for brevity often refer to the
data schema simply as the schema.

A tuple $\bar c \in \mn{adom}(D)^{|\bar x|}$ is an {\em answer} to $Q$
on \Sbf-database $D$ if $\bar c \in q(I)$ for each model $I$ of \Omc with
$I \supseteq D$.
The {\em evaluation of $Q(\bar x)$} on $D$, denoted $Q(D)$, is the set of all answers to $Q$ on~$D$.
\begin{exa}
  Consider the OMQ $(\Omc,\Sbf,q)$ where \Omc consist of the following
  TGDs:
  $$
  \begin{array}{r@{\;}c@{\;}lcr@{\;}c@{\;}l}
    \mn{Book}(x) &\rightarrow& \mn{Publication}(x) & &
    \mn{Article}(x) &\rightarrow& \mn{Publication}(x) \\[1mm]
    \mn{Publication}(x) & \rightarrow & \exists y \, \mn{hasPublisher}(x,y) &&
    \mn{Publication}(x) & \rightarrow & \exists y \, \mn{hasAuthor}(x,y) \\[1mm]
    % \exists y \, \mn{hasPublisher}(x,y) & \rightarrow &     \mn{Publication}(x) \\[1mm]
    \mn{hasPublisher}(x,y) & \rightarrow &
                                           \mn{Publisher}(y) &&
    \mn{hasAuthor}(x,y) & \rightarrow &
                                        \mn{Author}(y) \\[1mm]
\multicolumn{7}{c}{\mn{Publication}(x) \wedge \mn{hasAuthor}(x,y)
    \wedge \mn{hasPublisher}(x,y) \rightarrow
    \mn{SelfPublication}(x),}
  \end{array}
  $$
  \Sbf is the set of all relation symbols in \Omc, and
  $$
  q(x)= \exists y \,
    \mn{Author}(x) \wedge \mn{hasAuthor}(y,x) %\wedge \mn{Book}(y)
    \wedge \mn{SelfPublication}(y).
  $$
  The conjunctive query $q$ asks to return all authors that have
  self-published and the ontology \Omc adds knowledge about the domain
  of publications. Now consider the \Sbf-database $D$ that consists of
  the following facts:
  $$
  \begin{array}{lll}
    \mn{Book}(\mn{alice}) & \mn{hasAuthor}(\mn{alice},\mn{carroll}) &
    \mn{hasPublisher}(\mn{alice},\mn{macmillan}) \\[1mm]
    \mn{Book}(\mn{finn}) & \mn{hasAuthor}(\mn{finn},\mn{twain})  &
                                                                   \mn{hasPublisher}(\mn{finn},\mn{twain})
    \\[1mm]
    \mn{Book}(\mn{beowulf}) & \mn{SelfPublication}(\mn{beowulf}).
  \end{array}
  $$
  A straightforward semantic analysis shows that
  $\mn{twain} \in Q(D)$, despite the fact that the database $D$
  does not explicitly state the fact that \mn{finn} is a
  self-publication. While $\mn{beowulf}$ is a self-publication and we
  know from the ontology that it has an author, this author is not
  returned as an answer because their identity is unknown.
  In fact, $Q(D)=\{ \mn{twain} \}$.
\end{exa}
% In fact, use of ontologies
% is to enrich the schema that is available for formulating queries.

%
%
An \emph{OMQ language} is a class of OMQs. For a class of TGDs
$\class{C}$ and a class of UCQs $\class{Q}$, we write
$(\class{C},\class{Q})$ to denote the OMQ language that consists of
all OMQs $(\Omc,\Sbf,q)$ where \Omc is a set of TGDs from $\class{C}$
and $q \in \class{Q}$. For example, we may write
$(\class{G} \cap \class{FULL},\class{UCQ})$. %  and
% $(\class{ELH}^{\mn{dr}},\class{CQ})$
We say that an OMQ language
$(\class{C},\class{Q})$ has \emph{full data schema} if every
OMQ in it has.

\medskip
\noindent
{\bf The Chase.}
We next introduce the well-known chase procedure for making explicit
the consequences of a set of TGDs \cite{MaMS79,JoKl84, FKMP05,
  CaGK13}.  We first define a single chase step. Let $I$ be an
instance over a schema~$\Sbf$ and
$T= \phi(\bar x,\bar y) \rightarrow \exists \bar z \, \psi(\bar x,\bar
z)$ a TGD over~$\Sbf$. We say that $T$ is \emph{applicable} to a tuple
$(\bar c,\bar c')$ of constants in $I$ if
$\phi(\bar c,\bar c') \subseteq I$. In this case, {\em the result of
  applying $T$ in $I$ at $(\bar c,\bar c')$} is the instance
$J = I \cup \psi(\bar c,\bar c'')$, where $\bar c''$ is the tuple
obtained from $\bar z$ by simultaneously replacing each variable $z$
with a fresh distinct constant that does not occur in $I$. We describe
such a single chase step by writing
$I \xrightarrow{T, \,(\bar c,\bar c')} J$.
Let $I$ be an instance and $S$ a finite set of TGDs. A {\em chase
  sequence for $I$ with $S$} is a sequence of chase steps
$$ I_0
\xrightarrow{T_0,\,(\bar c_0,\bar c'_0)} I_1
\xrightarrow{T_1,\,(\bar c_1,\bar c'_1)} I_2 \dots$$
such that
(1) $I_0 = I$, (2) $T_i \in S$ for each $i \geq 0$, and (3)
$J \models S$ with $J = \bigcup_{i \geq 0} I_i$.  The instance
$J$ is the (potentially infinite) {\em result} of this chase sequence,
which always exists. The chase sequence is \emph{fair} if whenever
a TGD $T \in S$ is applicable to a tuple $(\bar
c,\bar c')$ in some $I_i$, then $I_j \xrightarrow{T, \,(\bar c,\bar c')}
I_{j+1}$ is part of the sequence for some $j \geq i$.
Note that our chase is oblivious, that is, a TGD is triggered whenever
its body is satisfied, even if also its head is already satisfied. As
a consequence, every fair chase sequence for $I$ with $S$ leads to
the same result, up to isomorphism. Thus, we can refer to {\em the}
result of chasing $I$ with $S$, denoted $\mn{ch}_S(I)$. The following
lemma gives the well-known main properties of the chase.

\begin{lem}\label{pro:chase}
  ~\\[-4mm]
  \begin{enumerate}

  \item
  Let $S$ be a finite set of TGDs and $I$ an instance. Then
  for every model $J$ of $S$
  with $I \subseteq J$, there is a homomorphism $h$ from
  $\mn{ch}_S(I)$ to $J$ that is the identity on $\mn{adom}(I)$.

\item $Q(D) = q(\mn{ch}_\Omc(D))$ for every OMQ
  $Q = (\Omc,\Sbf,q) \in (\class{TGD},\class{UCQ})$ and \Sbf-database
  $D$.

  \end{enumerate}
\end{lem}
Point~1 can be proved by constructing $h$ step by step, starting from
the identity on $\mn{adom}(I)$ and following chase rules.  Point~2 is
an easy consequence of Point~1 and the semantics of OMQs.
% We consider the standard (oblivious) chase procedure for guarded
% TGDs, defined in detail in the appendix. The result of chasing
% instance \Dmc with set of TGDs $S$ is denoted $\mn{ch}_S(I)$.

\smallskip

We shall often chase with sets $S$ of guarded full TGDs, that is, TGDs
from $\class{G} \cap \class{FULL}$. In contrast to the case of guarded
TGDs, the chase is then clearly finite. Moreover, it can be
constructed within the following time bounds.
\begin{lem}
  \label{lem:chasefpt}
  Given a database $D$ and finite set $S$ of TGDs from
  $\class{G} \cap \class{FULL}$, $\mn{ch}_S(D)$ can be
  constructed in time $f(||S||) \cdot O(||D||^3)$ for some
  computable function $f$. % and constant $c$.
\end{lem}
The time bound stated in Lemma~\ref{lem:chasefpt} can be achieved in a
straightforward way. To find a homomorphism from a TGD
$\phi(\bar x, \bar y) \rightarrow \psi(\bar x)$ in $S$ with guard
$R(\bar x, \bar y)$ to $D$, we can scan $D$ linearly to find all facts
that $R(\bar x, \bar y)$ can be mapped to and then verify by
additional scans that the remaining atoms in $\phi$ are also
satisfied. This takes time $||D||^2 \cdot n$, where $n$ is the number
of atoms in $\phi$.  {\color{\highlightColourTwo}Because all TGDs are guarded, it
  is easy to prove by induction on the number of chase rule
  applications that for every added fact $R(\bar{b})$, all constants
  in $\bar b$ must co-occur in some fact $T(\bar c)$ in $D$ where $T$
  occurs in $S$. Consequently, the chase can add at most
  $||D||\cdot k^k \cdot \ell$ fresh facts where $k$ is the maximum
  arity of relation symbols in $S$ and $\ell$ is the number of
  relation symbols that occur on the right-hand side of a TGD in
  $\Omc$. Note that $k^k$ is the maximum number of ways to choose a
  $k$-tuple of constants from a fact $T(\bar c)$ in $D$ where $T$
  occurs in $S$. }

For sets $S$ of TGDs from $\class{G} \cap \class{FULL}$, we may also
chase a CQ $q(\bar x)$ with $S$, denoting the result with
$\mn{ch}_S(q)$. What we mean is the (finite) result of chasing
database $D_q$ with $S$, viewing the result as a CQ with answer
variables $\bar x$, and adding back the equality atoms of $q$ (that
are dropped in the construction of $D_q$). We then have the
following.
\begin{lem}
  \label{lem:chaseoklem}
  $q(\mn{ch}_S(D)) = \mn{ch}_S(q)(\mn{ch}_S(D))$ for all databases
  $D$, CQs $q$, and finite sets of TGDs $S$ from
  $\class{G} \cap \class{FULL}$.
\end{lem}
It is clear that $\mn{ch}_S(q_i)(D) \subseteq q_i(D)$ for every
database $D$ because any homomorphism from $\mn{ch}_\Omc(q_i)$ to $D'$
is also a homomorphism from $q_i$ to~$D'$. The converse containment
also holds as every homomorphism from $q_i$ to $\mn{ch}_S(D)$ is also
a homomorphism from $\mn{ch}_\Omc(q_i)$ to~$\mn{ch}_S(D)$. This can be
shown by induction, considering all CQs
$q_i=p_1,\dots,p_\ell=\mn{ch}_\Omc(q_i)$ that arise when chasing $q_i$
with~\Omc.

\medskip
\noindent
%\paragraph{Treewidth.}
{\bf Treewidth.}
Treewidth is a widely used notion that measures the degree of tree-likeness of a graph.
Let $G = (V,E)$ be an undirected graph. A  \emph{tree decomposition} of $G$ is a pair $\delta = (T_\delta, \chi)$, where $T_\delta = (V_\delta,E_\delta)$ is a tree, and $\chi$ is a labeling function $V_\delta \rightarrow 2^{V}$, i.e., $\chi$  assigns a subset of $V$ to each node of $T_\delta$, such that:
\begin{enumerate}
	\item $\bigcup_{t \in V_\delta} \chi(t) = V$,

	\item if $\{u,v\} \in E$, then $u,v \in \chi(t)$ for some $t \in V_\delta$,

	\item for each $v \in V$, the set of nodes $\{t \in V_\delta \mid v \in \chi(t)\}$ induces a connected subtree of~$T_\delta$.
\end{enumerate}
The {\em width} of $\delta$ is the number $\max_{t \in V_\delta}
\{|\chi(t)|\} - 1$. If the edge set $E$ of $G$ is non-empty, then the
{\em treewidth} of $G$ is the minimum width over all its tree
decompositions; otherwise, it is defined to be one. Note that trees
have treewidth~1.
Each instance $I$ is associated with an undirected graph (without self
loops) $G_I = (V,E)$, called the {\em Gaifman graph} of $I$, defined
as follows: $V = \mn{adom}(I)$, and $\{a,b\} \in E$ iff there is a
fact $R(\bar c) \in I$ that mentions both $a$ and~$b$.
%, i.e., $a$ and $b$ coexist in $R(\bar t)$.
%
The \emph{treewidth} of $I$ is the treewidth of~$G_I$.

\medskip
\noindent
%\paragraph{Parameterized Complexity.}
{\bf Parameterized Complexity.}
A \emph{counting problem} over a finite alphabet $\Lambda$ is a
function $P:\Lambda^* \rightarrow \mathbb{N}$ and a
\emph{parameterized counting problem} over $\Lambda$ is a pair
$(P,\kappa)$, with $P$ a counting problem over $\Lambda$ and $\kappa$
the \emph{parameterization} of $P$, a function
$\kappa: \Lambda^* \rightarrow \mathbb{N}$ that is computable in
\PTime. An example of a parameterized counting problem is $\#${\sf pClique} in which $P$ maps (a suitable encoding of) each pair $(G,k)$
with $G$ an undirected graph and $k \geq 0$ a clique size to the
number of $k$-cliques in $G$, and where
$\kappa(G,k)=k$. Another example is $\#${\sf pDomSet} where
$P$ maps each pair $(G,k)$ to the number of dominating sets of
size $k$, and where again $\kappa(G,k)=k$.

A counting problem $P$ is a \emph{decision problem} if the range of
$P$ is $\{0,1\}$, and a \emph{parameterized decision problem} is
defined accordingly. An example of a parameterized decision problem is
{\sf pClique} in which $P$ maps each pair $(G,k)$ to~1 if the
undirected graph $G$ contains a $k$-clique and to 0 otherwise,
and where $\kappa(G,k)=k$.

A parameterized problem $(P,\kappa)$ is \emph{fixed-parameter
  tractable} (fpt) if there is a computable function
$f: \mathbb{N} \rightarrow \mathbb{N}$ such that $P(x)$ can be
computed in time $|x|^{O(1)}{\cdot}f(\kappa(x))$ for all inputs~$x$.
We use \FPT\xspace  to denote the class of all
parameterized counting problems that are
fixed-parameter tractable.

A \emph{Turing fpt-reduction} from a parameterized counting problem
$(P_1,\kappa_1)$ %over $\Lambda_1$
to a parameterized counting problem $(P_2,\kappa_2)$ %over $\Lambda_2$
is an algorithm that computes $P_1$ with %an~
oracle access to $P_2$, runs
within the time bounds of fixed parameter tractability for
$(P_1,\kappa_1)$, and when started on input $x$ only makes oracle
calls with argument $y$ such that $\kappa_2(y) \leq f(\kappa_1(x))$,
for some computable function~$f$.  The reduction is called a
\emph{parsimonious fpt-reduction} if only a single oracle call is made
at the end of the computation and its output is then returned as the
output of the algorithm without any further modification.

A parameterized counting problem $(P,\kappa)$ is \emph{$\wOneC$-easy}
if it can be reduced to $\#${\sf pClique} and it is
\emph{$\wOneC$-hard} if $\#${\sf pClique} reduces to $(P,\kappa)$,
both in terms of Turing fpt-reductions. $\wOne$-easiness and
-hardness are defined analogously, but using {\sf pClique} in place of
$\#${\sf pClique}, and likewise for $\wTwoC$ and $\#${\sf pDomSet},
and for $\aTwoC$ and the parameterized problem of counting the answers
to CQs, the parameter being the size of the CQ. For $C \in \{ \wOne,
\wOneC, \allowbreak \wTwoC, \aTwoC\}$, $(P,\kappa)$ is \emph{$C$-equivalent} if it
is $C$-easy and $C$-hard.  Note that we follow
\cite{countingTrichotomy,DBLP:conf/icalp/DellRW19} in defining both
easiness and hardness in terms of Turing fpt-reductions; stronger
notions would rely on parsimonious fpt-reductions
\cite{DBLP:journals/siamcomp/FlumG04}.

\section{The Classification Without TGDs}
\label{sect:notgds}

In the series of papers
\cite{DBLP:journals/jcss/DurandM14,countingQueriesSructural,countingTrichotomy,countingPositiveQueries,DBLP:conf/icalp/DellRW19},
the parameterized complexity of answer counting is studied for classes of
CQs and UCQs, resulting in a rather detailed classification. We present it in
this section as a reference point and as a basis for establishing our
own classifications later on. We start with introducing the various
structural measures that play a role in the classification.

\begin{figure}
    \begin{tikzpicture}[yscale=1, xscale=1.1]
        \scriptsize
        \tikzstyle{quantified}=[circle, draw, scale=.6]
        \tikzstyle{answer} = [circle, fill, scale=.6]
        
\node (off) at (-5,0) {};
        
\node (cap) at (-2.5,-1) {(a) query {$q$}};
\node[answer] (x1) at ($(3,4) + (off)$) {};
\node[answer] (x2) at ($((4,3.5) + (off)$) {};
\node[answer] (x3) at ($(3,3) + (off)$) {}; 
\node[answer] (x4) at ($(3,2) + (off)$) {}; 
\node[answer] (x5) at ($(2,1) + (off)$) {}; 
\node[answer] (x6) at ($(2,3.5)  + (off)$) {}; 

\node (x1n) at ($(2.6,4) + (off)$) {$x_1$};
\node (x2n) at ($(4.4,3.5) + (off)$) {$x_2$};
\node (x3n) at ($(3.3,2.8) + (off)$) {$x_3$}; 
\node (x4n) at ($(3.4,2) + (off)$) {$x_4$}; 
\node (x5n) at ($(2.25,1.1) + (off)$) {$x_5$}; 
\node (x6n) at ($(2.3,3.5) + (off)$) {$x_6$}; 

%    \node[circle, fill] (x5) at (4,0) {};%

\node[quantified] (y1) at ($(2,3) + (off)$) {};
\node (y1n) at ($(2.4,2.8) + (off)$) {$y_1$};
%        \node[quantified] (y2) at (2,1) {};
%\node[quantified] (y3) at (1,0) {}; 
%\node[quantified] (y4) at (0,1) {}; 
\node[quantified] (y5) at ($(1,2) + (off)$) {}; 
\node[quantified] (y6) at ($(1,3) + (off)$) {};
\node (y5n) at ($(.65,2.2) + (off)$) {$y_2$}; 
\node (y6n) at ($(.65,3.2) + (off)$) {$y_3$};  
%        \node[quantified] (y7) at (0,2) {}; 

\node[quantified] (c1)  at ($(3.8,4.2) + (off)$) {};
\node (c1n)  at ($(4.2,4.2) + (off)$) {$y_4$};
%\node[quantified] (c2)  at (5.5,.5) {};
%    \node[circle, draw] (y5) at (4,0) {};

%triangles

\draw [rounded corners=4mm] ($(y1) +(.4,.2)$)--($(y5) +(-.2,-.4)$)--($(y6) +(-.2,+.2)$)--cycle;
\draw [rounded corners=2mm] ($(c1) +(.3,.2)$)--($(x1) +(-.2,.2)$)--($(x3) +(-.1,-.4)$)--cycle;
\draw [rounded corners=2mm] ($(x2) +(.2,-.1)$)--($(x3) +(-.3,-.3)$)--($(c1) +(0.05,+.4)$)--cycle;

%\draw [brown]  ($(1.5,2.5) + (off)$) circle [radius=.85];
%\draw [rotate around={-45:((-1.8,4)}, green]  (-2.8,4) ellipse [x radius=.5, y radius=1];
%\draw[rotate around={-30:(-1.4,3.5)}, brown] (-1.4,3.5) ellipse (17pt and 30pt);
%\draw[rotate around={-30:(-1.7,3.8)}, brown] (-1.7,3.8) ellipse (17pt and 30pt);

%\draw 	 
%(x1) -- 
%(x2)--(x3) -- (x1) 
%(y1) -- (y6)
%(y5) -- (y6)
%(y1) -- (y5);

\draw%[green]
%(x1) -- (c1)
%(x3) -- (c1)
%        (x3) -- (y1)
%(x4) -- (y2)
%        (x4) -- (y5)
(x5) -- (y5);

\draw[] (x6) -- (y1);
\draw[] (x4) -- (y5);

%\draw
%(x2) -- (c1);
%        (x4) -- (c2)
%        (x5) -- (c2);
%
%\draw [blue, dashed]  (1.2,2.7) circle [radius=1];
%\draw [blue, dashed]  (3.7,4.5) circle [radius=.5];
%\draw [blue, dashed]  (5,1.5) circle [radius=1.2];
        
%%%%%%%%%%%%%%%%%%%%%%%%%%%%%%%%%%%%%%%%%%%%%%%%%%%%%%%%%%%%%%%%%%%%
        \node (cap) at (2,-1) {(b) Gaifman graph {$G_q$} of {$q$}};
        \node[answer] (x1) at (3,4) {};
        \node[answer] (x2) at (4,3.5) {};
        \node[answer] (x3) at (3,3) {}; 
        \node[answer] (x4) at (3,2) {}; 
        \node[answer] (x5) at (2,1) {}; 
        %    \node[answer] (x6) at (-2,2) {}; 
        
        \node (x1n) at (2.6,4) {$x_1$};
        \node (x2n) at (4.4,3.5) {$x_2$};
        \node (x3n) at (3.8,2.9) {$x_3{=}x_6$}; 
        \node (x4n) at (3.4,2) {$x_4$}; 
        \node (x5n) at (2.2,1.1) {$x_5$}; 
        
        %    \node[circle, fill] (x5) at (4,0) {};%
        
        \node[quantified] (y1) at (2,3) {};
        \node (y1n) at (2.3,2.8) {$y_1$};
%        \node[quantified] (y5) at (2,1) {};
        %\node[quantified] (y6) at (1,0) {}; 
        %\node[quantified] (c1) at (0,1) {}; 
        \node[quantified] (y2) at (1,2) {}; 
        \node[quantified] (y3) at (1,3) {};
        \node (y2n) at (.65,2.2) {$y_2$}; 
        \node (y3n) at (.65,3.2) {$y_3$};  
%        \node[quantified] (y7) at (0,2) {}; 
        
        \node[quantified] (y4)  at (3.8,4.2) {};
        \node (y4n)  at (3.55,4.29) {$y_4$};
        %\node[quantified] (c2)  at (5.5,.5) {};
        %    \node[circle, draw] (y5) at (4,0) {};

        \draw 	 
        %(x1) -- 
        (x2)--(x3) -- (x1) 
        (y1) -- (y3)
        (y2) -- (y3)
        (y1) -- (y2);

        \draw%[green]
        (x1) -- (y4)
        (x3) -- (y4)
%        (x3) -- (y1)
        %(x4) -- (y2)
%        (x4) -- (y5)
        (x5) -- (y2);
        
        \draw[purple] (x3) -- (y1);
        \draw[purple] (x4) -- (y2);
        
        \draw
        (x2) -- (y4);
%        (x4) -- (c2)
%        (x5) -- (c2);
        
%        \draw [blue, dashed]  (1.2,2.7) circle [radius=1];
%        \draw [blue, dashed]  (3.7,4.2) circle [radius=.5];
        \draw [rounded corners=4mm, blue, dashed] ($(x3) +(.4,.2)$)--($(y3) +(-.2,.3)$)--($(y2) +(-.2,-.2)$)--($(x5) +(0,-.3)$)--($(x4) +(+.2,-.1)$)--cycle;
        \draw [rounded corners=4mm, blue, dashed] ($(y4) +(.1,.35)$)--($(x1) +(-.2,.2)$)--($(x3) +(-.2,-.3)$)--($(x2) +(.3,0)$)--cycle;
        %\draw [blue, dashed]  (5,1.5) circle [radius=1.2];

%%%%%%%%%%%%%%%%%%%%%%%%%%%%%%%%%%%%%%%%%%%%%%%%%%%%%%%%%%%%%%%%%%
        \node (cap) at (7,-1) {(c) contract of {$G_q$}};
        \node[answer] (z1) at (7,4) {};
        \node[answer] (z2) at (8,3.5) {};
        \node[answer] (z3) at (7,3) {}; 
        \node[answer] (z4) at (7,2) {}; 
        \node[answer] (z5) at (6,1) {}; 
        %    \node[answer] (x6) at (-2,2) {}; 
        
        \node (z1n) at (7.6,4) {$x_1$};
        \node (z2n) at (8.4,3.5) {$x_2$};
        \node (z3n) at (7.65,2.9) {$x_3{=}x_6$}; 
        \node (z4n) at (7.4,2) {$x_4$}; 
        \node (z5n) at (6.3,1) {$x_5$}; 
        \draw 	 
        (z3) -- (z2)
        %(z2) -- (z1) 
        (z3) -- (z1);

        \draw[red] (z4) -- (z5);
        \draw[red] (z3) -- (z4);
        \draw[red] (z1) -- (z2);
        \draw[red, out=150,in=30] (z5) --  (z3);
        
        %vertical
        \draw[dashed] (0, -.5) -- (0, 4.5);
        \draw[dashed] (5, -.5) -- (5, 4.5);
      \end{tikzpicture}
      \caption{An example for Gaifman graphs and their contracts.}
    % \caption{\scriptsize The Gaifman graph (b) and the contract (c) of the CQ (a)
    %  The brown shapes in (a) are the tertiary relations $T$ and the edges depict the binary relations $R$. 
    %  The red edges in (c) indicate the edges added by the connected components of quantified variables.
    %  Those components are indicated in (b) by the blue ellipses.
    %  The purple edges in (b) indicate the maximal matching for the linked matching number.\\
    %  Here we have: $\text{TW}=2$, $\text{CTW}=2$, $\text{SS}=3$, and $\text{LMN}=2$.}
\label{fig:exquery}
 \end{figure}
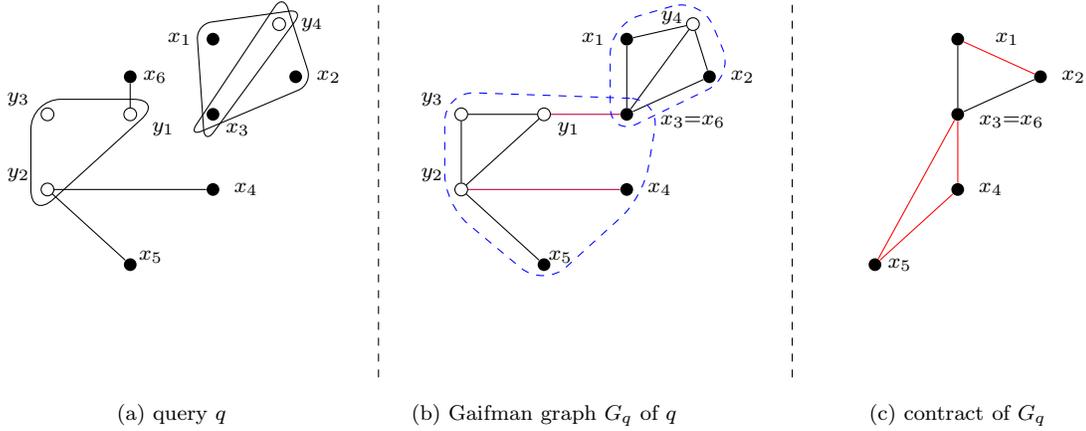

Let $q(\bar x) = \exists \bar y \, \varphi(\bar x, \bar y)$ be a CQ.
The \emph{Gaifman graph} of $q$, denoted $G_q$, is defined as
$G_{D^\sim_q}$.
The \emph{treewidth (TW)} of $q(\bar x)$ is the treewidth of $G_q$.

An \emph{$\bar x$-component} of $G_q$ is any undirected graph that
can be obtained
as follows: (1)~take the subgraph of $G_q$ induced by vertex set
$\bar y$,
%$\mn{var}(q) \setminus \bar x$,
(2)~choose a maximal connected
component $(V_c,E_c)$, and (3)~re-add all edges from $G_q$
that contain at least one vertex from~$V_c$.
Note that the last step may re-add answer variables as
vertices, but no quantified variables.
The
\emph{contract} of $G_q$, denoted $\mn{contract}(G_q)$, is the
restriction of $G_q$ to the answer variables, extended with every
edge $\{x_1,x_2\} \subseteq \bar x$ such that $x_1,x_2$ co-occur in
some $\bar x$-component of~$G_q$. We shall often be interested in the
treewidth of the contract of a CQ %$G_q$
$q$, which we refer to as the
\emph{contract treewidth (CTW)} of $q$.
An example is given in Figure~\ref{fig:exquery}. Part~(a) shows CQ
$$
\begin{array}{rcl}
q(x_1, x_2, x_3, x_4, x_5, x_6) &=& {\color{\highlightColourTwo} \exists y_1  \exists y_2 \exists y_3 \exists y_4\,} T(y_1,y_2,y_3) \wedge T(x_2,x_3,y_4),
                                    T(x_1,y_4, x_3) \; \wedge \\[1mm]
  && \hspace*{2.2cm} R(x_6, y_1) \wedge R(x_4,y_2) \wedge R(y_2, x_5) \wedge  x_3 = x_6
  \end{array}
$$
where filled nodes indicate answer variables and hollow nodes
quantified variables, the triangles represent the ternary relation
$T$, and the edges the binary relation $R$. Part (b) shows the
Gaifman graph $G_q$ of $q$, where $x_3$ and $x_6$ have been
identified. The dashed blue boxes show the $\bar x$-components
and the contract of $G_q$ is shown in Part~(c) with edges that
have been added due to the $\bar x$-components shown in red.
Both the treewidth and contract treewidth of $q$ are two.

The \emph{starsize (SS)} of $q$ is the maximum number of answer
variables in any $\bar x$-component of $G_q$. Note that the same
notion is called strict starsize in \cite{countingTrichotomy} and
dominating starsize in~\cite{DBLP:conf/icalp/DellRW19}. It is
different from the original notion of starsize from
\cite{DBLP:journals/jcss/DurandM14,countingQueriesSructural}. The
starsize of the CQ in Figure~\ref{fig:exquery} is three.

% Let $Y = \mn{var}(q) \setminus \bar x$.
A set of quantified variables $S$ in $q$ is \emph{node-well-linked} if
for every two disjoint sets $S_1,S_2 \subseteq S$ of the same
cardinality, there are $|S_1|$ vertex disjoint paths in $G_q$ that
connect the vertices in $S_1$ with the vertices in $S_2$. For example,
$S$ is node-well-linked if $G_q|_S$ takes the form of a grid or of a
clique.  A matching $M$ from the answer variables $\bar x$ to the
quantified variables $\bar y$ in the graph $G_q$ (in the standard
sense of graph theory) is \emph{linked} if the set $S$ of quantified
variables that occur in $M$ is node-well-linked. The \emph{linked
  matching number (LMN)} of $q$ is the size of the largest linked
matching from $\bar x$ to $\bar y$ in $G_q$. One should think of the
linked matching number as a strengthening of starsize. We do not only
demand that many answer variables are interlinked by the same
$\bar x$-component, but additionally require that this component is
sufficiently large and highly connected (`linked').  In Part~(b) of
Figure~\ref{fig:exquery}, the purple edges in~(b) indicate the maximal
matching. The LMN of the CQ in that figure is two.

Figure~\ref{fig:examples} contains some example CQs with associated
measures. For a class of CQs $\class{C}$, the contract treewidths of
CQs in $\class{C}$ being bounded by a constant implies that the same
is true for starsizes, and bounded starsizes in turn imply bounded
linked matching numbers. In fact, the starsize of a CQ
  $q$ is bounded by the contract treewidth of $q$ plus one and its
  linked matching number is bounded by its starsize.  There are no
implications between treewidth and contract treewidth. In
Figure~\ref{fig:examples}, Example~(a) generalizes to any treewidth
while always having contract treewidth~1 and Example~(c), which has
contract treewidth~3, generalizes to any contract treewidth (and
starsize) while always having treewidth~1. We refer to
\cite{countingTrichotomy,DBLP:conf/icalp/DellRW19} for additional
examples.
\begin{figure}
  \tikzstyle{quantified}=[circle, draw, scale=.6]
  \tikzstyle{answer} = [circle, fill, scale=.6]
  \newcommand{\noPic}{4}
  
  \centering

  \begin{subfigure}[b]{0.24\textwidth}
     \centering
     \resizebox{\linewidth}{!}{
         \begin{tikzpicture}%[scale=1.]
          
%         \node (a) at (-3 , -2) {};
%         \node (b) at (4,4) {};
          
          \node[quantified] (y1) at (-2,-2) {};
          \node[quantified] (y2) at (-2,0) {};
          \node[quantified] (y3) at (-2,2) {}; 
          %    \node[circle, fill] (x5) at (4,0) {};%
          
          \node[quantified] (x1) at (2,-2) {};
          \node[quantified] (x2) at (2,0) {};
          \node[quantified] (x3) at (2,2) {}; 
          %    \node[circle, draw] (y5) at (4,1) {};

          \draw 	 
          (x1) -> (y1)
          (x1) -> (y2)
          (x1) -> (y3)
          
          (x2) -> (y1)
          (x2) -> (y2)
          (x2) -> (y3)

          (x3) -> (y1)
          (x3) -> (y2)
          (x3) -> (y3);

          \end{tikzpicture}
      }
      \caption{TW = 3, CTW = 1}
      \label{fig:example-TW}
  \end{subfigure}
%%%%%%%%%%%%%%%%%%%%%%%%%%%%%%%%%%%%%%%%%%%%%%%%%%%%%%%%%%%%%%
  \begin{subfigure}[b]{0.24\textwidth}
     \centering
     \resizebox{\linewidth}{!}{
         \begin{tikzpicture}%[yscale=1.25]
         
%         \node (a) at (-3 , -2) {};
%         \node (b) at (4,4) {};
         
         \node[answer] (x1) at (-1,0) {};
         \node[answer] (x2) at (.5,2) {};
         \node[answer] (x3) at (2,0) {}; 
         \node[answer] (x4) at (.5,-2) {};
         %    \node[circle, fill] (x5) at (4,0) {};%
         
         \node[quantified] (y) at (.5,0) {};
         \node[quantified] (y2) at (-2,0) {};
         \node[quantified] (y3) at (-.25,-1) {}; 
         %   \node[quantified] (y4) at (0,-1) {};
         %    \node[circle, draw] (y5) at (4,1) {};

         \draw
         (x1) -> (y)
         %    (x2) -> (y)
         (x3) -> (y)
         %    (x4) -> (y)
         (x2) -- (y2)
         (y2) -- (x4)
         
         (x1) -- (y3)
         (y3) -- (x4)
         
         (x1) -- (x2)
         (x2) -- (x3)
         (x3) -- (x4);

         \end{tikzpicture}
     }
     \caption{ CTW {=} 3, SS {=} 2}
     \label{fig:example-CTW}
  \end{subfigure}
%%%%%%%%%%%%%%%%%%%%%%%%%%%%%%%%%%%%%%%%%%%%%%%%%%%%%%%%%%%%%%
  \begin{subfigure}[b]{0.24\textwidth}
      \centering
      \resizebox{\linewidth}{!}{
          \begin{tikzpicture}%[scale=1]
          
%          \node (a) at (-3 , -2) {};
%          \node (b) at (4,4) {};
          
          \node[answer] (x1) at (-2,0) {};
          \node[answer] (x2) at (0,2) {};
          \node[answer] (x3) at (2,0) {}; 
          \node[answer] (x4) at (0,-2) {};
          %    \node[circle, fill] (x5) at (4,0) {};%
          
          \node[quantified] (y) at (0,0) {};
          %   \node[quantified] (y2) at (0,1) {};
          %   \node[quantified] (y3) at (1,0) {}; 
          %   \node[quantified] (y4) at (0,-1) {};
          %    \node[circle, draw] (y5) at (4,1) {};

          \draw
          (x1) -> (y)
          (x2) -> (y)
          (x3) -> (y)
          (x4) -> (y);

          \end{tikzpicture}
      }
      \caption{SS = 4, LMN = 1}
      \label{fig:example-SS}
  \end{subfigure}
%%%%%%%%%%%%%%%%%%%%%%%%%%%%%%%%%%%%%%%%%%%%%%%%%%%%%%%%%%%%%%
  \begin{subfigure}[b]{0.24\textwidth}
    \centering
    \resizebox{\linewidth}{!}{
        \begin{tikzpicture}[scale=1]

%        \node (a) at (-3 , -2) {};
%        \node (b) at (4,4) {};
        
        \node[answer] (x1) at (-2,0) {};
        \node[answer] (x2) at (0,2) {};
        \node[answer] (x3) at (2,0) {}; 
        \node[answer] (x4) at (0,-2) {};
        %    \node[circle, fill] (x5) at (4,0) {};%
        
        \node[quantified] (y1) at (-1,0) {};
        \node[quantified] (y2) at (0,1) {};
        \node[quantified] (y3) at (1,0) {}; 
        \node[quantified] (y4) at (0,-1) {};
        %    \node[circle, draw] (y5) at (4,1) {};

        \draw
        (y1) -> (y2)
        (y1) -> (y3)
        (y1) -> (y4)
        
        (y2) -> (y1)
        (y2) -> (y3)
        (y2) -> (y4)
        
        (y3) -> (y1)
        (y3) -> (y2)
        (y3) -> (y4)
        
        (y4) -> (y1)
        (y4) -> (y2)
        (y4) -> (y3)
        (y4) -> (y4)
        
        (x1) -> (y1)
        (x2) -> (y2)
        (x3) -> (y3)
        (x4) -> (y4);

        \end{tikzpicture}
    }
    \caption{LMN = 4}
    \label{fig:example-LMN}
  \end{subfigure}
%%%%%%%%%%%%%%%%%%%%%%%%%%%%%%%%%%%%%%%%%%%%%%%%%%%%%%%%%%%%%%

\caption{Examples for structural measures:
  Example~(\subref{fig:example-TW}) is the (3,3)\=/complete bipartite
  graph, the contract of Example~(\subref{fig:example-CTW}) is the
  4\=/clique, Example~(\subref{fig:example-SS}) is a 4\=/star, and
  Example~(\subref{fig:example-LMN}) is a 4\=/star with the~4\=/clique in
  the~centre.}
    \label{fig:examples}
\end{figure}
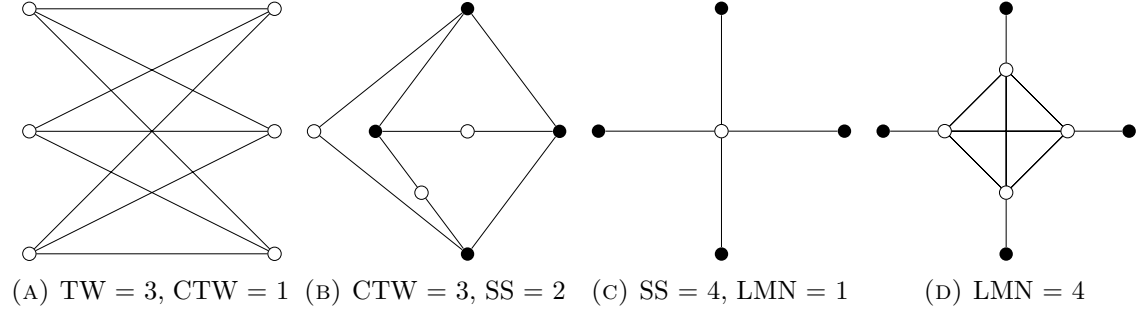

It is a fundamental observation that cores of CQs are guaranteed to
have minimum measures among all equivalent CQs, as stated by the
following lemma \cite{countingTrichotomy,DBLP:conf/icalp/DellRW19}.
\begin{lem}
\label{lem:coresaregreat}
  If a CQ $q$ is equivalent to a CQ of treewidth $k$, then
  $\mn{core}(q)$
  has treewidth at most $k$. The same is true for contract treewidth,
  starsize,
  and linked matching number.
\end{lem}
An additional ingredient needed to formulate the classification for
UCQs emerges from \cite{countingPositiveQueries}. There, Chen and
Mengel associate with every UCQ $q$ a set of CQs
$\mn{cl}_{\mn{CM}}(q)$ such that counting the number of answers to $q$
is closely tied to counting the number of answers to the CQs in
$\mn{cl}_{\mn{CM}}(q)$. We now introduce this set, which we refer to
as the \emph{Chen-Mengel closure}, in detail.

Two CQs $q_1(\bar x_1)$ and $q_2(\bar x_2)$ over the same schema \Sbf
are \emph{counting equivalent} if $\#q_1(D)=\#q_2(D)$ for all
\Sbf-databases $D$.  Let $q(\bar x)=p_1 \vee \cdots \vee p_n$.  The
starting point for defining $\mn{cl}_{\mn{CM}}(q)$ is the observation
that, by the inclusion-exclusion principle,
every database $D$ satisfies
$$
\begin{array}{rcl}
  \#q(D) &=& \displaystyle \sum_{I \subseteq [n]} (-1)^{|I|+1} \cdot \#\big (\bigwedge_{i
             \in I}p_i(D) \big ).
\end{array}
$$
We can manipulate this sum as follows: if there are two summands
$c_1 \cdot \# \big (\bigwedge_{i \in I_1}p_i(D) \big )$ and
$c_2 \cdot \# \big (\bigwedge_{i \in I_2}p_i(D) \big )$ such that
$\bigwedge_{i \in I_1}p_i$ and $\bigwedge_{i \in I_2}p_i$ are
counting equivalent, then delete both summands and add
$(c_1+c_2) \cdot \#(\bigwedge_{i \in I_1}p_i(D) \big )$ to the sum.
After doing this exhaustively, delete all summands with coefficient zero. The elements
of $\mn{cl}_{\mn{CM}}(q)$ are all CQs $\bigwedge_{i \in I}p_i$
in the original sum that are counting equivalent to some CQ
%CQ
$\bigwedge_{i \in J}p_i$ which remains in the sum.\footnote{This
  definition slightly deviates from that of Chen and Mengel, who
  include no two CQs that are counting equivalent. For all
   relevant purposes, however, the two definitions are
   interchangable.
 } Note that the number of CQs in $\mn{cl}_{\mn{CM}}(q)$ might be
 exponentially larger than the number of CQs in $q$ and that
 $\mn{cl}_{\mn{CM}}(q)$ does not need to contain all CQs from the
 original UCQ $q$. For a class $\class{Q}$ of UCQs, we use
 $\mn{cl}_{\mn{CM}}(\class{Q})$ to denote
 $\bigcup_{q \in \class{Q}} \mn{cl}_{\mn{CM}}(q)$.

\theoremstyle{defC}
\newtheorem{exaC}[thm]{Example}

\begin{exaC}[\cite{countingPositiveQueries}]
  Consider schema $\Sbf = \{A,R\}$ with $A$ unary and $R$ binary.
  Let
  $$
  \begin{array}{rcl}
%    p'(x,y,z,t) &=& A(x) \wedge A(y) \wedge A(z) \wedge A(t) \\[1mm]
q(x,y,z,t) &=& p_1(x,y,z,t) \lor p_2(x,y,z,t) \lor p_3(x,y,z,t) \text{ where}
    \\[1mm]
    p_1(x,y,z,t) &=& R(x,y) \wedge  R(y,z) \wedge A(x) \wedge A(y) \wedge A(z) \wedge A(t)\\[1mm]
    p_2(x,y,z,t) &=& R(z,t) \wedge  R(t,x) \wedge A(x) \wedge A(y) \wedge A(z) \wedge A(t)\\[1mm]
    p_3(x,y,z,t) &=& R(y,z) \wedge R(z,t) \wedge A(x) \wedge A(y) \wedge A(z) \wedge A(t)
  \end{array}
    $$
    and for $I \subseteq \{1,2,3\}$, let $p_{I}$ be the CQ $\bigwedge_{i\in I} p_i$.
    By inclusion-exclusion,
    for every $\Sbf$\=/database $D$ we have
    $$
    \begin{array}{rcl}
    \#q(D) &=& \#p_{\{1\}}(D) + \#p_{\{2\}}(D) + \#p_{\{3\}}(D)
      \\[1mm]
      &&- \, \#p_{\{1,2\}}(D) - \#p_{\{1,3\}}(D) - \#p_{\{2,3\}}(D)
      \\[1mm]
      &&+ \, \#p_{\{1,2,3\}}(D)
    \end{array}
    $$

    It is not hard to see that $p_{\{1\}}$, $p_{\{2\}}$, and
    $p_{\{3\}}$ are pairwise counting equivalent, and so are
    $p_{\{1,3\}}$ and $p_{\{2,3\}}$. Moreover, $p_{\{1,2\}}$ and
    $p_{\{1,2,3\}}$ %, the two queries of treewidth $2$,
    are equivalent and thus counting equivalent.  Applying the
    manipulation described above, we thus arrive at the sum
    \[
    \#q(D) = 3 \cdot \#p_{\{1\}}(D) -2 \cdot \#p_{\{1,3\}}(D).
    \]

    It follows that $\mn{cl}_{\mn{CM}}(q) = \{\#p_{\{1\}}, \#p_{\{2\}},\#p_{\{3\}}, \#p_{\{1,3\}},\#p_{\{2,3\}}\}$.
    Note that the CQs $p_{\{1,2\}}$ and $p_{\{1,2,3\}}$ cancelled each other out.
\end{exaC}
{\color{\highlightColourTwo}
Note that $\mn{cl}_{\mn{CM}}(q)$ is defined so that for every
\Sbf-database $D$, $\#q(D)$ can be computed in polynomial time from
the counts $\#q'(D)$, $q' \in \mn{cl}_{\mn{CM}}(q)$. This, in fact, is
the raison d'etre of the Chen-Mengel closure.
}

We are now ready to state the characterization.
\begin{thmC}[\cite{countingTrichotomy,countingPositiveQueries,DBLP:conf/icalp/DellRW19}]
  \label{thm:chenmengeldell}
  Let $\class{Q} \subseteq \class{UCQ}$ be recursively enumerable and
  have relation symbols of bounded arity, and let
  $\class{Q}^\star=\{ \mn{core}(q) \mid q \in \mn{cl}_{\mn{CM}}(\class{Q})\}$. Then the following holds:
	\begin{enumerate} %\itemsep=0pt

        \item If the treewidths %of $\class{Q}$
          and the contract treewidths of CQs in $\class{Q}^\star$
          are bounded, then \linebreak[4]{\sf AnswerCount}$(\class{Q})$ is in~\fpt;
          it is even in \PTime when
          $\class{Q} \subseteq \class{CQ}$.

        \item If the treewidths of CQs in $\class{Q}^\star$ are unbounded
          and the contract treewidths of CQs in $\class{Q}^\star$ are bounded, then
{\sf
            AnswerCount}$(\class{Q})$ is \wOne-equivalent.
        \item If the contract treewidths of CQs in $\class{Q}^\star$ are
          unbounded and the starsizes of CQs in $\class{Q}^\star$
          are bounded, then {\sf
            AnswerCount}$(\class{Q})$ is \wOneC-equivalent.
        \item If the starsizes of CQs in
          $\class{Q}^\star$ are unbounded, % and the linked
          % matching numbers of CQs in $\class{Q}^\star$
          % are bounded,
          then {\sf AnswerCount}$(\class{Q})$ is
          \wTwoC-hard.
        \item If the linked matching numbers of CQs in
          $\class{Q}^\star$ are unbounded, then {\sf
            AnswerCount}$(\class{Q})$ is~\aTwoC-equivalent.
	\end{enumerate}
\end{thmC}
We remark that $\mn{cl}_{\mn{CM}}(q)=\{q\}$ when $q$ is a CQ,
and thus $\class{Q}^\star=\{ \mn{core}(q) \mid
q \in \class{Q}\}$ when $\class{Q} \subseteq \class{CQ}$ in
Theorem~\ref{thm:chenmengeldell}. The assumption that relation
symbols have bounded arity is needed only for the lower bounds,
but not for the upper bounds.

\medskip

Note that the classification given by Theorem~\ref{thm:chenmengeldell}
is not complete. It leaves open the possibility that there is a class
of (U)CQs $\class{Q}$ such that {\sf AnswerCount}$(\class{Q})$ is
\wTwoC-hard, but neither \wTwoC-equivalent nor \aTwoC-equivalent.
It is conjectured in \cite{DBLP:conf/icalp/DellRW19} that such a class
$\class{Q}$ indeed exists and in particular that there might be
classes $\class{Q}$ such that {\sf AnswerCount}$(\class{Q})$  is \#\textnormal{\sc W$_\mn{func}$[2]}-equivalent. The classification
also leaves open whether having unbounded linked matching numbers is a
necessary condition for \aTwoC-hardness. While a complete
classification
is certainly desirable we note that, from our perspective, the most
relevant aspect is the delineation of the FPT cases from the hard
cases, achieved by Points~(1)-(3) of the theorem.

%\section{Problems Studied and Results on Exact Counting}
\section{Problems Studied and Main Results}

We introduce the problems studied and state the main results of this
paper. We start with ontology-mediated querying and then proceed
to querying under constraints.
Every OMQ language $\class{Q}$ gives rise to an answer
counting problem, defined exactly as in Section~\ref{sect:prelim}:

\begin{center}
	\fbox{\begin{tabular}{ll}
			PROBLEM : & {\sf AnswerCount}$(\class{Q})$
			\\INPUT : & A query $q
                                              \in
                                             \class{Q}$ over schema
                                             \Sbf and
			an $\Sbf$-database $D$
			\\
			OUTPUT : &  $\#q(D)$
	\end{tabular}}
\end{center}

Our first main result is a counterpart of
Theorem~\ref{thm:chenmengeldell} for the OMQ language
$(\class{G},\class{UCQ})$, restricted to OMQs based on the full
schema.  To illustrate the effect on the complexity of counting of
adding an ontology, we first show that the ontology interacts with
all of the measures in Theorem~\ref{thm:chenmengeldell}.

\begin{exa}
  \label{ex:lowering-measures}
  Let $\Omc = \{ R(x,y) \rightarrow S(x,y) \}$ and $\Sbf = \{R,S\}$. For all $n \geq 0$, let
  \[
  \begin{array}{r c l @{\ } l}
  q_n(x_1,\dots,x_n, z_1,\dots,z_n) &=& \exists_{1 < i+j < n+2}\
                                          y_{i,j} &  \bigwedge_{1 \leq i \leq n} R(x_i,z_1) \wedge
  \bigwedge_{1 \leq i < n } R(z_{i} , z_{i+1}) \, \wedge \\
  & & &\bigwedge_{i+j = n+ 1} S(x_{i}, y_{i,j} ) \, \wedge\\
  & & &\bigwedge_{2 < i+j < n + 2} S(y_{i+1,j}, y_{i,j}) \wedge S(y_{i,j+1}, y_{i,j}).
  \end{array}
  \]
  Then $q_n$ is a core of   treewidth $\lfloor \frac{n}{2} \rfloor$,
contract treewidth $n$, starsize $n$, and
  linked matching number $n$. But the OMQ
  $(\Omc,\Sbf,q_n)$ is equivalent
  to $(\Omc,\Sbf,p_n)$ with $p_n$ obtained from $q_n$ by
  dropping all $S$-atoms. Since $p_n$ is tree\=/shaped and has no
  quantified variables, all measures are at most 1.
  Figure~\ref{fig:lowering-measures-example} depicts query $q_4$.
\end{exa}

\begin{figure}
    %\begin{wrapfigure}{r}{0.45\textwidth}
    \centering
    \begin{tikzpicture}[xscale=1.5]
    \tikzstyle{quantified}=[circle, draw, scale=.6]
    \tikzstyle{answer} = [circle, fill, scale=.6]

    \node[quantified] (y10) at (1,0) {};
    \node[quantified] (y20) at (1,1) {};
    \node[quantified] (y30)at (1,2) {}; 
    \node[quantified] (y40) at (1,3) {};
    
    \node[quantified] (y11) at (1.5,0.5) {};
    \node[quantified] (y21) at (1.5,1.5) {};
    \node[quantified] (y31)at (1.5,2.5) {};

    \node[quantified] (y12) at (2,1) {};
    \node[quantified] (y22) at (2,2) {};

    \node[quantified] (y13) at (2.5,1.5) {};

    \node[answer] (x1) at (0,0) {};
    \node[answer] (x2) at (0,1) {};
    \node[answer] (x3) at (0,2) {}; 
    \node[answer] (x4) at (0,3) {};

    \node[answer] (z1) at (-1,1.5) {};
    \node[answer] (z2) at (-1.5,1.5) {};
    \node[answer] (z3) at (-2,1.5) {}; 
    \node[answer] (z4) at (-2.5,1.5) {};
    
    \draw[->] (x1) -> (y10) node[midway,above]{$S$};
    \draw[->] (x2) -> (y20) node[midway,above]{$S$};
    \draw[->] (x3) -> (y30) node[midway,above]{$S$};
    \draw[->] (x4) -> (y40) node[midway,above]{$S$};
    
    \draw[->] (y10) -> (y11) node[midway,above]{$S$};
    \draw[->] (y20) -> (y11) node[midway,above]{$S$};
    \draw[->] (y20) -> (y21) node[midway,above]{$S$};
    \draw[->] (y30) -> (y21) node[midway,above]{$S$};
    \draw[->] (y30) -> (y31) node[midway,above]{$S$};
    \draw[->] (y40) -> (y31) node[midway,above]{$S$};
    
    \draw[->] (y11) -> (y12) node[midway,above]{$S$};
    \draw[->] (y21) -> (y12) node[midway,above]{$S$};
    \draw[->] (y21) -> (y22) node[midway,above]{$S$};
    \draw[->] (y31) -> (y22) node[midway,above]{$S$};
    
    \draw[->] (y12) -> (y13) node[midway,above]{$S$};
    \draw[->] (y22) -> (y13) node[midway,above]{$S$};

    \draw[->] (x1) -> (z1) node[midway,above]{$R$};
    \draw[->] (x2) -> (z1) node[midway,above]{$R$};
    \draw[->] (x3) -> (z1) node[midway,above]{$R$};
    \draw[->] (x4) -> (z1) node[midway,above]{$R$};
    
    \draw[->] (z1) -> (z2) node[midway,above]{$R$};
    \draw[->] (z2) -> (z3) node[midway,above]{$R$};
    \draw[->] (z3) -> (z4) node[midway,above]{$R$};
    \end{tikzpicture}
    \caption{CQ $q_4$ from Example~\ref{ex:lowering-measures}.
    Filled circles indicate answer variables.}
    \label{fig:lowering-measures-example}
\end{figure}
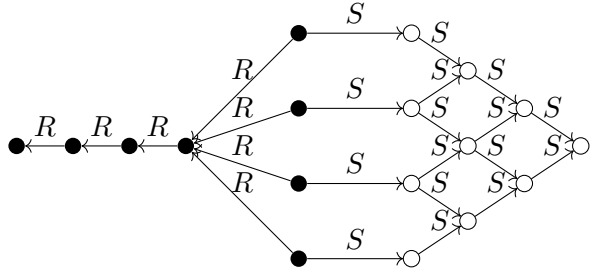

\medskip

Before we state our characterization, we observe as a
preliminary that OMQs from $(\class{G},\class{UCQ})$ can be rewritten
into equivalent ones from $(\class{G} \cap \class{FULL},\class{UCQ})$,
that is, existential quantifiers can be removed from rule heads when
the actual query is adjusted in a suitable way.  This has already been
observed in the literature.
\begin{thmC}[\cite{BDFLP-PODS20}]
  \label{thm:removeEx}
  For every OMQ $Q \in (\class{G},\class{UCQ})$, there
  is an equivalent OMQ from $(\class{G} \cap
  \class{FULL},\class{UCQ})$ that can be effectively computed.
  % such that the semantic treewidth, the
  % contract treewidth, and the starsize of $Q^\exists$ are identical to
  % that of~$Q$.
\end{thmC}
%
%
% the statement about semantic treewidth made in
% Theorem~\ref{thm:removeEx} is actually trivial since $Q$ and
% $Q^\exists$ are equivalent.
% A proof of Theorem~\ref{thm:removeEx} is
% given in the appendix.
The proof of Theorem~\ref{thm:removeEx} is constructive, that is, it
provides an explicit way of computing, given an OMQ
$Q=(\Omc,\Sbf,q) \in (\class{G},\class{UCQ})$, an equivalent OMQ from
$(\class{G} \cap \class{FULL},\class{UCQ})$. We denote this OMQ with
$Q^\exists
=(\Omc^\exists,\Sbf,q^\exists)$
  % such that the semantic treewidth, the
  % contract treewidth, and the starsize of $Q^\exists$ are identical to
  % that of~$Q$.
and call it the \emph{$\exists$-rewriting} of~$Q$.  It is worth noting
that even if $q$ contains no equality atoms, such atoms might be
introduced during the construction of $q^\exists$. What is more,
different CQs in the produced UCQ can comprise different equalities on
answer variables, and thus repeated answer variables cannot be
used. This is actually the main reason for admitting equality atoms in
(U)CQs in this paper.

For OMQs $Q \in (\class{G},\class{UCQ})$, we define a set
$\mn{cl}_{\mn{CM}}(Q)$ of OMQs from $(\class{G},\class{CQ})$
in exact analogy with the definition of
$\mn{cl}_{\mn{CM}}(q)$ for UCQs $q$, that is, for
$Q=(\Omc,\Sbf, p_1 \vee \cdots \vee p_n)$, we use the OMQs
$(\Omc,\Sbf,p_i)$ in place of the CQs $p_i$ from the UCQ $q$ in the
definition of $\mn{cl}_{\mn{CM}}(Q)$. This requires the use of
counting equivalence for OMQs, which is defined in the expected
way. %$\mn{cl}_{\mn{CM}}(Q)$
For a class $\class{Q}$ of OMQs, we use $\mn{cl}_{\mn{CM}}(\class{Q})$
to denote $\bigcup_{Q \in \class{Q}} \mn{cl}_{\mn{CM}}(Q)$.
% The following is~immediate from the definition of $\mn{cl}_{\mn{CM}}(Q)$.
% %
% \begin{lem}
%   \label{lem:cmOMQ}
%   For each $Q=(\Omc,\Sbf,q) \in (\class{G},\class{UCQ})$ and
%   \Sbf-database $D$, $\#Q(D)$ can be computed in polynomial
%   time from the counts $\#Q'(D)$, $Q' \in \mn{cl}_{\mn{CM}}(Q)$.
% \end{lem}
%

For a class
$\class{Q} \subseteq (\class{G},\class{UCQ})$, we now identify a class
$\class{Q}^\star$ of CQs by setting
%
% $$
% \class{Q}^\star = \{ \mn{core}(\mn{ch}_{\Omc^\exists}(p))\mid \exists Q \in \class{C}: Q^\exists
% = (\Omc^\exists,\Sbf,q^\exists) \text{ and } p \in
% \mn{cl}_{\mn{CM}}(q^\exists)\}.
% $$
%
\[
\class{Q}^\star = \{ \mn{core}(\mn{ch}_{\Omc^\exists}(p))\mid \exists
Q \in \class{Q}:
%Q^\exists
%= (\Omc^\exists,\Sbf,q^\exists) \text{ and }
(\Omc^\exists,\Sbf,p) \in
\mn{cl}_{\mn{CM}}(Q^\exists)\}.
\]
% Note that the set $\mn{cl}_{\mn{CM}}(Q^\exists)$  is built by~exploiting counting equivalence with respect to~the set of tuple generating dependencies~$\Omc^\exists$.
%
% \smallskip
%
% For technical reasons, we will sometimes also consider answer
%
% As a convention, if the class $\class{D}$ of considered databases is
% the class of all $\Sbf$\=/databases, then we may write
% $\answerCount(\class{O})$ instead of
% $\answerCount(\class{O}, \class{D})$.  Similarly, if
% $\class{O} = \{Q\}$ is a singleton, we may simply write
% $\answerCount(Q,\class{D})$.
%
%\smallskip
%\noindent
%
% We next give an initial complexity classification for the OMQ language
% $(\class{G} \cap \class{FULL},\class{CQ})$. Ultimately, we are of
% interested in $(\class{G},\class{UCQ})$. However, it turns out that
% our results are easier to state and to prove for the former, so this
% is our starting point.
% For a
% class $\class{C}$ of OMQs, we say that \emph{the treewidths of
%   $\class{C}$ are bounded} if there is a $k \geq 0$ such that every
% OMQ in $\class{C}$ is of semantic treewidth at most~$k$, and likewise
% for semantic treewidth, contract treewidth, and starsize.
%
In other words, the CQs in $\class{Q}^\star$ are obtained by choosing
an OMQ from $\class{Q}$, replacing it with~$Q^\exists$, then choosing
an OMQ $(\Omc^\exists,\Sbf,p)$ from $\mn{cl}_{\mn{CM}}(Q)$, chasing
$p$ with $\Omc^\exists$, and finally taking the core.
Our first main result is as follows.

\begin{thm}
  \label{thm:main-results-OMQs}
  Let $\class{Q} \subseteq (\class{G},\class{UCQ})$ be a recursively
  enumerable class of OMQs with full data schema and relation symbols of
  bounded arity. Then the  following hold:
	\begin{enumerate}

        \item If the treewidths %of $\class{Q}$
          and contract treewidths of CQs in $\class{Q}^\star$
          are bounded, then {\sf AnswerCount}$(\class{Q})$ is in FPT.
          % it is even in \PTime when
          % $\class{Q} \subseteq (\class{ELH}^{\mn{dr}},\class{CQ})$.

        \item If the treewidths of CQs in $\class{Q}^\star$ are unbounded
          and the contract treewidths of CQs in $\class{Q}^\star$ are bounded, then
{\sf
            AnswerCount}$(\class{Q})$ is \wOne-equivalent.
        \item If the contract treewidths of CQs in $\class{Q}^\star$ are
          unbounded and the starsizes of CQs in $\class{Q}^\star$
          are bounded, then {\sf
            AnswerCount}$(\class{Q})$ is \wOneC-equivalent.
        \item If the starsizes of CQs in
          $\class{Q}^\star$ are unbounded, % and the linked
          % matching numbers of CQs in $\class{Q}^\star$
          % are bounded,
          then {\sf AnswerCount}$(\class{Q})$ is
          \wTwoC-hard.
        \item If the linked matching numbers of CQs in
          $\class{Q}^\star$ are unbounded, then {\sf
            AnswerCount}$(\class{Q})$ is \aTwoC-equivalent.

	\end{enumerate}
        The upper bounds also hold when the arity of relation symbols
        is unbounded.
\end{thm}

\smallskip
Points~(1) to~(5) of Theorem~\ref{thm:main-results-OMQs} parallel exactly
those of Theorem~\ref{thm:chenmengeldell}, but of course the
definition of $\class{Q}^\star$ is a different one. It is through this definition
that we capture the potential interaction between the ontology and the
structural measures. Note, for example, that the class of OMQs
$(\Omc,\Sbf,q_n)$, $n \geq 1$,
from Example~\ref{ex:lowering-measures} would be classified
as \aTwoC-equivalent if $\mn{core}(\mn{ch}_{\Omc^\exists}(p))$
was replaced with $p$ in the definition of $\class{Q}^\star$ while
it is in fact in FPT.
Also note that the \PTime statement in Point~(1) of
Theorem~\ref{thm:chenmengeldell} is absent in
Theorem~\ref{thm:main-results-OMQs}. In fact, evaluating Boolean OMQs
from $(\class{G},\class{UCQ})$ is \TwoExpTime-complete (\ExpTime-hard
when
the arity of relation symbols is bounded by a constant) \cite{CaGK13} and
since for Boolean OMQs evaluation coincides with answer counting, \PTime cannot be
attained.

\smallskip Our second main result concerns querying under integrity
constraints that take the form of guarded TGDs. In contrast to OMQs,
the constraints are thus not used for deductive reasoning, but instead
give rise to a promise regarding the shape of the input database.
Following \cite{BDFLP-PODS20}, we define a {\em constraint-query
  specification} (CQS) to be a triple $S = (\Tmc,\Sbf,q)$ where $\Tmc$ is
a finite set of TGDs over finite schema $\Sbf$ and $q$ a UCQ over~$\Sbf$.  We
call $\Tmc$ the set of {\em integrity constraints}.
%We may write $\sch{\mi{Sp}}$ for $\ins{S}$.
%
Overloading notation, we write $(\class{C},\class{Q})$ for the class
of CQSs in which the set of integrity constraints is formulated in the
class of TGDs~$\class{C}$, and the query is coming from the class of
queries $\class{Q}$. It will be clear from the context whether
$(\class{C},\class{Q})$ is an OMQ language or a class of CQSs.
Every class $\class{C}$ of CQSs gives rise to the following answer counting
problem.
\begin{center}
	\fbox{\begin{tabular}{ll}
			PROBLEM : & {\sf AnswerCount}$(\class{C})$
			\\INPUT : & A set of TGDs $\Tmc$, a query $q$, and
                        an \Sbf-database $D$ that satisfies $\Tmc$ \\
                        & such that $(\Tmc,\Sbf,q) \in \class{C}$.
			\\
			OUTPUT : &  $\#q(D)$
	\end{tabular}}
\end{center}

Our second main result parallels
Theorems~\ref{thm:chenmengeldell} and~\ref{thm:main-results-OMQs}.
We refrain from explicitly listing all cases again.
\begin{thm}
  \label{thm:closedmain}
  Let $\class{Q} \subseteq (\class{G},\class{UCQ})$ be a recursively
  enumerable class of CQSs with relation symbols of bounded
  arity. Then Statements~1-5 of Theorem~\ref{thm:main-results-OMQs}
  hold.
\end{thm}
Note that the delineation of the considered complexities is
identical for ontology-mediated querying and for querying under
constraints. In particular, Theorem~\ref{thm:closedmain} (implicitly)
uses exactly the same class of CQs $\class{Q}^\star$ and the same associated
measures.

It would be interesting to know whether $\mn{AnswerCount}(\class{Q})$
being in FPT coincides with $\mn{AnswerCount}(\class{Q})$ being in
\PTime for classes of CQSs
$\class{Q} \subseteq (\class{G},\class{CQ})$. Note that this is the
case for evaluation in the presence of constraints that are guarded
TGDs \cite{BaGP16,BFGP20} and also for answer counting without constraints
\cite{countingTrichotomy}. The proofs of these results, however,
break in our setting.

\section{Querying Under Integrity Constraints}

We derive Theorem~\ref{thm:closedmain} from
Theorem~\ref{thm:main-results-OMQs} by means of reduction, so that in the
rest of the paper we may concentrate on the case of ontology-mediated
querying. In fact, Theorem~\ref{thm:closedmain} is a consequence of
Theorem~\ref{thm:main-results-OMQs} and the following result.%
\begin{thm}
%\begin{restatable}{thm}{closedmain}
    \label{thm:reduction-open-world-closed-world}
    Let $\class{C} \subseteq (\class{G}, \class{UCQ})$ be a~recursively
    enumerable class of CQSs and let $\class{C}'$ be $\class{C}$ viewed
    as a class of OMQs based on the full schema.%
    \footnote{Syntactically, a CQS $(\Tmc,\Sbf,q)$ and an OMQ $(\Tmc,\Sbf,q)$ are actually the same
        thing except that the definition of CQSs is more strict
        regarding the schema \Sbf; as a consequence when
        viewing a CQS as an OMQ,  the latter is based on the full schema.}
    Then there is a
    Turing fpt\=/reduction from $\mn{AnswerCount}(\class{C}')$ to
    $\mn{AnswerCount}(\class{C})$ and there is a parsimonious
    polynomial time reduction from $\mn{AnswerCount}(\class{C})$ to
    $\mn{AnswerCount}(\class{C}')$.
%\end{restatable}
\end{thm}
The reduction from $\mn{AnswerCount}(\class{C})$ to
$\mn{AnswerCount}(\class{C}')$ is immediate: given a set of guarded
TGDs $\Tmc$, a CQ $q$, and an \Sbf-database $D$ that satisfies $\Tmc$,
we can view $(\Tmc,\Sbf,q)$ as an OMQ $Q$ based on the full schema and
return $\#Q(D)$ as $\#q(D)$. It is easy to see that this is correct.

For the converse reduction, we are given a $Q=(\Omc,\Sbf,q)$ that is a
CQS from $\class{C}$ viewed as an OMQ and an \Sbf-database~$D$. It
seems a natural idea to simply view $Q$ as a CQS, which it originally
was, and replace $D$ with $\mn{ch}_\Omc(D)$ so that the promise is
satisfied, and to then return $\#q(\mn{ch}_\Omc(D))$ as
$\#Q(D)$. However, there are two obstacles. First, $\mn{ch}_\Omc(D)$
need not be finite; and second, chasing adds fresh constants which
changes the answer count.  We solve the first problem by replacing the
infinite chase with a (finite!)  database $D^\star$ that extends $D$
and satisfies \Omc. This is based on the following result from
\cite{BDFLP-PODS20} which is essentially a consequence of $\class{G}$ being
finitely controllable \cite{DBLP:conf/lics/BaranyGO10}.
\begin{thmC}[\cite{BDFLP-PODS20}]
\label{thm:finmod}
Given an ontology $\Omc \subseteq \class{G}$, an \Sbf-database $D$,
and an $n \geq 1$, one can effectively construct a finite database
$D^*$ % \supseteq D$
that satisfies the following conditions:
  \begin{enumerate}

  \item $D^* \models \Omc$ and $D \subseteq D^*$;

  \item $\bar a \in q(D^*)$ iff $\bar a \in Q(D)$ for all OMQs
    $(\Omc,\Sbf,q)$ where $q$ has
    at most $n$ variables and for all tuples $\bar
    a$ that use only constants in $\mn{adom}(D)$.

  \end{enumerate}
  The construction of $D^*$ takes time
  $f(||\Omc||+n) \cdot ||D||^{O(1)}$ with $f$ a computable function.
\end{thmC}
To address the second problem, we correct the count. Note that this
cannot be done by introducing fresh unary relation symbols as markers
to distinguish the original constants from those introduced by the
chase as this would require us to change the query, potentially
leaving the class of queries that we are working with. We instead use
an approach inspired by \cite{countingTrichotomy}.
The idea is to compute $\#q(D')$ on a set of databases $D'$ obtained
from $D^\star$ by cloning % (that is, duplicating)
constants in
$\mn{adom}(D) \subseteq \mn{adom}(D^\star)$.
%\setminus \mn{adom}(D^*)$ in a systematic way, arranging
The results can be arranged in a system of equations whose
coefficients form a Vandermonde matrix.  Finally, the system can be
solved to obtain
$\#q(D)$.  This is formalized by the following lemma where we use
$\mn{clones}(D)$ to denote the class of~all
$\Sbf$\=/databases that can be obtained from \Sbf-database
$D$ by cloning constants. %A proof is in the appendix of the long version.
\begin{lem}
%\begin{restatable}{lem}{lempartdomans}
  \label{lemma:partial-domain-answers}
%  Fix a constant $r$.
  There is an algorithm that, given a UCQ $q(\bar{x})$ over schema
  \Sbf, an~$\Sbf$\=/database~$D$, and a set $F \subseteq \dom(D)$,
  computes $\#(q(D) \cap F^{|\bar{x}|})$ in time $||q||^{O(\mn{ar}(S))}
  \cdot ||D||$ % , where $f$ is a computatble function and $p$ is a
  % polynomial,
  using an oracle for $\mn{AnswerCount}(\{q\},
  \mn{clones}(D))$.
    %The class $\textit{clones}(D)$
    %is the class of~all $\Sbf$\=/databases obtained from $D$ by cloning constants.
    %The oracle is invoked only on $\Sbf$\=/databases $D'$ such that $D' \models \Omc$.
%\end{restatable}
\end{lem}
\begin{proof} %[Proof of Lemma~\ref{lemma:partial-domain-answers}]
    We first give a brief  overview of the algorithm.  Assume that the input
    is a UCQ $q(\bar x)$, a database $D$, and a set $F$.  The algorithm first constructs
    databases $D_1,\dots,D_{|\bar x|+1}$ by starting with $D$
    and cloning constants from $F$. Then, it computes $\#q(D_j)$ for $1 \leq j \leq
    |\bar x|+1$ and, finally, constructs and solves a system of linear
    equations for which one of the unknowns is the desired value $\#\big( q(D) \cap F^{|\bar{x}|} \big)$.
    We now make this precise.

    For $1 \leq j < |\bar x|+1$, database $D_{j}$ is constructed from
    ${D}$ by cloning each element from ${F}$ exactly $j-1$ times.  In
    particular, $D_{1} = {D}$.  Observe that
    $$
    |\dom(D_{j})| \leq %|\dom({D})| + j|{F}| \leq
    j|\mn{adom}(D)|,\ \
%     \ \text{
%        and } \
    |D_{j}| \leq
    j^{\mn{ar}(\Sbf)} {\cdot} |D| \leq ||q||^{\mn{ar}(\Sbf)} {\cdot} |D|,
%|\Sbf|(j|\mn{adom}(D)|)^r.
    $$
    and each $D_j$ can be constructed in time $||q||^{O(\mn{ar}(S))}
  \cdot ||D||$.

 %   Thus, each database $D_j$ is of size polynomial in that of $D$.

    Now, for $0 \leq i \leq |\bar{x}|$ and  $1 \leq j < |\bar x|+1$, let
    $q^i(D_{j})$ denote the subset of answers $\bar{a} \in q(D_{j})$ such
    that exactly $i$ positions in $\bar a$ have constants that are
    in~${F}$ or have been obtained from such constants by cloning. We
    claim that $\#q^i(D_{j}) = j^i \cdot \#q^i(D)$, that is, having
    $j$ such clones multiplies each answer $\bar a \in q(D)$ having $i$
    positions of the described kind exactly $j^i$ times.
    By the semantics, this is immediate if $q$ is a CQ. So assume that
    $q(\bar x)=p_1(\bar x) \vee \cdots \vee p_k(\bar x)$ where each
    $p_i$ is a  CQ.  By
    the inclusion exclusion principle,
    $$\#q^i(D_{j}) = \sum_{S \subseteq \{1,\dots,k\}}
    (-1)^{|S|-1}\#(\bigwedge_{\ell \in S} p_\ell)^i(D_j)$$
    and likewise for $D$ in place of $D_j$. Since
    $\#p^i(D_{j}) = j^i \cdot \#p^i(D)$ for each CQ
    $p = \bigwedge_{\ell \in S} p_\ell$ that occurs in the sum,
    we obtain $\#q^i(D_{j}) = j^i \cdot \#q^i(D)$, as claimed.

    Let $1 \leq j \leq |\bar x|+1$. Since the sets $q^i(D_{j})$ partition
    the set $q(D_{j})$, we have that $\#q(D_{j}) =
    \sum_{i=0}^{|\bar{x}|} \#q^i(D_{j})$. Moreover, since  we have shown that $\#q^i(D_{j}) = j^i \cdot \#q^i(D)$ we can infer that
    $$
    \#q(D_{j}) =  \sum_{i=0}^{|\bar{x}|} {\color{\highlightColourTwo} j^i}\cdot \#q^i(D).
    $$
    In the above equation, there are $|\bar{x}|+1$ unknown values
    $\#q^0(D),\dots, \#q^{|\bar x|}(D)$ and one value,
    i.e.~$\#q(D_{j})$, that can be computed by the oracle.

    Taking this equation for $j=0, \dots, |\bar{x}| + 1$ generates a
    system of $|\bar{x}|+1$ linear equations with $|\bar{x}|+1$
    variables.  The coefficients of the system form a~Vandermonde
    matrix, which implies that the equations are independent and that the
    system has a~unique solution. Thus, we can solve the system in
    polynomial time, e.g.\ by Gaussian elimination, to compute the
    values $\#q^0(D),\dots, \#q^{|\bar x|}(D)$.

    Clearly $q^{|\bar{x}|}(D)=q({D}) \cap {F}^{|\bar{x}|}$, so
    returning $\#q^{|\bar{x}|}(D_{1})$ yields $\#({q}({D}) \cap
    F^{|\bar{x}|})$,
    as desired.

    It can be verified that, overall, the algorithm runs in the
    time stated in Lemma~\ref{lemma:partial-domain-answers}.
    %
    % The algorithm works as follows.
    % \begin{enumerate} %\itemsep=0pt
    %     \item Construct databases $D^{j}$.
    %     \item Invoke the oracle $|\bar{x}|+1$ times to compute the values $|{q}(D^{j})|$.
    %     \item Solve the system of equations $[|{q}(D^{j})|] = [i^j] [|\Hmc_{i}^{1}|]$.
    %     \item Return $|\Hmc^{1}_{|\bar{x}|}|$.
    %
    % \end{enumerate}
\end{proof}
Now for the reduction from $\mn{AnswerCount}(\class{C}')$ to
$\mn{AnswerCount}(\class{C})$ claimed in
Theorem~\ref{thm:reduction-open-world-closed-world}.  Let $Q(\bar{x})
= (\Omc, S, q)$ be a CQS from $\class{C}$ viewed as an OMQ, and let
$D$ be an $\Sbf$\=/database. We first construct the database $D^*$ as
per Theorem~\ref{thm:finmod} with $n$ being the number of variables in $q$. We then apply the algorithm asserted by
Lemma~\ref{lemma:partial-domain-answers} with $D^*$ in place of $D$
and with $F \eqdef \dom(D)$.  Cloning preserves guarded TGDs and thus
we can use the oracle (which can compute $\#q(D')$ for any
\Sbf-database $D'$ that satisfies \Omc) for computing
$\mn{AnswerCount}(\{q\}, \mn{clones}(D))$ as required by
Lemma~\ref{lemma:partial-domain-answers}.

\section{Counting Equivalence}
\label{sect:countequiv}

For the proofs of both the upper and lower bounds stated in
Theorem~\ref{thm:main-results-OMQs}, we need a good grasp of counting
equivalence. For the lower bounds, the same is true for the related
notion of semi-counting equivalence. In this section, we make some
fundamental observations regarding these notions.

In the lower bound proofs, we shall often be concerned with classes of
databases
$\class{D}^\Sbf_{\Omc} = \{ \mn{ch}_{\Omc}(D) \mid D \text{ an
  \Sbf-database}\}$ for some ontology \Omc from
$\class{G} \cap \class{FULL}$. Note that since \Omc is from
$\class{FULL}$, each $\mn{ch}_{\Omc}(D)$ is finite and thus indeed a
database. We observe some important properties of the class
$\class{D}^\Sbf_{\Omc}$ that are folklore and easy to see.  For a
schema \Sbf, we define the \Sbf-database $D^\top_\Sbf$ by fixing a
constant $c$ and setting
$D^\top_\Sbf = \{ R(c,\dots,c) \mid R \in \Sbf \}$.
\begin{lem}
  \label{lem:closure}
  For every ontology $\Omc \subseteq \class{TGD}$ and schema \Sbf that
  contains all symbols in~\Omc, the class of instances
  $\class{D}^\Sbf_{\Omc}$
  %$\{ \mn{ch}_{\Omc}(D) \mid D \text{ an \Sbf-database}\}$
  is closed
  under direct product and contains~$D^\top_\Sbf$. If
  $\Omc \subseteq \class{G}$, then it~is closed under disjoint union
  and cloning of elements.  If
   $\Omc \subseteq \class{G} \cap \class{FULL}$, then it is closed
   under induced subdatabases.
\end{lem}
{\color{\highlightColourTwo}
For closure under direct products, it suffices to observe that there
is a homomorphism from the direct product $I$ of instances $I_1$ and
$I_2$ to each of the components $I_1$ and $I_2$. Thus, applicability
of a TGD in the product implies applicability in both components.
Moreover, the result of the applications in the components is then
clearly also found in the product, see e.g.~\cite{faginSTOC80} for
more details. The arguments for the other closure properties are
similar, but simpler.
}

We now make a central observation regarding the relationship between
(semi-)counting equivalence over classes of databases
$\class{D}^\Sbf_{\Omc}$ and (semi-)counting equivalence over the class
of all databases. But let us first introduce the notion of
semi-counting equivalence. Two CQs $q_1(\bar x_1)$ and $q_2(\bar x_2)$
over the same schema \Sbf are \emph{semi-counting equivalent} if they
are counting equivalent over all \Sbf-databases~$D$ such that $\#q_1(D) >0$
and $\#q_2(D) >0$. %The observation is as follows.
For a CQ~$q$, we use $\hat q$ to denote the CQ obtained from $q$ by
dropping all maximal connected subqueries that contain no answer
variable.\footnote{Note that if $q$ is Boolean, then $\hat q$ is the
  empty CQ. It evaluates to true on every database.}
\newcommand{\lemcountingEquivClassD}{
    Let $q_1(\bar x_1)$ and $q_2(\bar x_2)$ be
    {\color{\highlightColour}equality-free} CQs over schema \Sbf and let
    $\class{D}$ be a class of \Sbf-databases that contains
    $D_{q_i}$ and $D_{\hat q_i}$ for
    $i \in \{1,2\}$ and is closed under cloning. Then
    \begin{enumerate}
        \item $q_1$ and $q_2$  are \emph{counting equivalent} over
        $\class{D}$ iff $q_1$ and $q_2$  are \emph{counting
            equivalent} over the class of all \Sbf-databases;
        \item if $\class{D}$ is closed under disjoint union %, induced
        %                   subdatabases,
        and
        contains $D^{\top}_{\Sbf}$, then $q_1$ and $q_2$
        are  \emph{semi-counting equivalent} over $\class{D}$ iff
        $q_1$ and $q_2$  are  \emph{semi-counting equivalent} over the class of all \Sbf-databases.
    \end{enumerate}
}
\begin{lem}
    \label{lem:countingEquivClassD}
    \lemcountingEquivClassD
\end{lem}
The `if' directions of Points~(1) and~(2) of
Lemma~\ref{lem:countingEquivClassD} are trivial.  The `only if'
directions are a consequence of results
on % the decision procedures for
counting equivalence and semi-counting equivalence obtained
in~\cite{countingPositiveQueries}.  We give more
details in the appendix.

We next observe that counting equivalence and semi-counting
equivalence are decidable over classes of databases
$\class{D}^\Sbf_{\Omc}$. For the class of all databases, this has been
shown in \cite{countingPositiveQueries}. In fact, it is shown there
that CQs $q_1$ and $q_2$
are
counting equivalent iff there is a way to rename their answer
variables to make them equivalent in the standard sense, and that they
are semi-counting equivalent iff $\hat q_1$ and $\hat q_2$
are counting
equivalent.  Consequently, both problems are in {\sc NP}.
{\color{\highlightColour}For a CQ $q$, let $\tilde q$ denote the CQ obtained from $q$
by removing all equality atoms and identifying any two variables
$x_1,x_2$ with $x_1 = x_2 \in q$.
}

\begin{prop}
  \label{prop:counequivdecclass}
  Let $\Omc \subseteq \class{G} \cap \class{FULL}$ and let \Sbf be a
  schema that contains all symbols from~\Omc. Given CQs $q_1(x_1)$ and
  $q_2(x_2)$ over \Sbf, it is decidable whether $q_1$ and $q_2$ are counting
  equivalent over $\class{D}^\Sbf_{\Omc}$.  The same holds for
  semi-counting equivalence.
\end{prop}
\begin{proof}
  {\color{\highlightColour} Let $q_1(\bar x_1)$ and $q_2(\bar x_2)$ be given as the
    input. Let $i \in \{1,2\}$. It is easy to see that $q_i(\bar x_i)$ is
    (semi-)counting equivalent to $\tilde q_i(\bar y_i)$ over the class of
    all databases, and consequently also over
    $\class{D}^\Sbf_{\Omc}$. We may thus assume that $q_1$ and
    $q_2$ are equality-free as otherwise we can replace them with
    $\tilde q_1$ and $\tilde q_2$.  We then construct
    $\mn{ch}_\Omc(q_1)$ and $\mn{ch}_\Omc(q_2)$, check whether they
    are (semi-)counting equivalent over the class of all databases
    using the decision procedure from~\cite{countingPositiveQueries},
    and return the result.

  We have to argue that this is correct. By
  Lemma~\ref{lem:chaseoklem}, % We first observe that $q_1$
  % and $q_2$ are \mbox{(semi-)counting} equivalent over
  % $\class{D}^\Sbf_{\Omc}$ if $\mn{ch}_\Omc(q_1)$ and
  % $\mn{ch}_\Omc(q_2)$ are (semi-)counting equivalent over
  % $\class{D}^\Sbf_{\Omc}$. abcde We may thus
  it suffices to
  decide whether
  $\mn{ch}_\Omc(q_1)$ and $\mn{ch}_\Omc(q_2)$ are (semi-)counting
  equivalent over $\class{D}^\Sbf_{\Omc}$, which by
  Lemma~\ref{lem:countingEquivClassD} is identical to their
  (semi-)counting equivalence over the class of all databases.  Note
  that the preconditions of Lemma~\ref{lem:countingEquivClassD} are
  satisfied. In particular, $q'_i = \mn{ch}_\Omc(q_i)$ is equality-free and
  both $D_{q'_i}$ and $D_{\hat q'_i}$ are in $\class{D}^\Sbf_{\Omc}$.
  } % colour bracket
\end{proof}
In the upper bound, it shall be necessary to compute the Chen-Mengel
closure of an OMQ $Q \in
(\class{G}\cap\class{FULL},\class{UCQ})$. This is possible by simply
following the definition of $\mn{cl}_{\mn{CM}}(Q)$, but requires us to
decide counting equivalence of OMQs.  We show that this is possible.
\begin{cor}
  \label{prop:OMQCEdec}
  Given OMQs
  $Q_1(\bar x_1),Q_2(\bar x_2) \in
  (\class{G}\cap\class{FULL},\class{CQ})$ over the full schema, where
  $Q_i=(\Omc,\Sbf,q_i)$ for $i \in \{1,2\}$, it is decidable whether
  $Q_1$ and $Q_2$ are counting equivalent.
\end{cor}
Corollary~\ref{prop:OMQCEdec} is a direct consequence of
Proposition~\ref{prop:counequivdecclass}. In fact, it follows from
Point~2 of Lemma~\ref{pro:chase} and the definition of
$\class{D}^\Sbf_{\Omc}$ yields that $Q_1$ and $Q_2$ are
(semi-)counting equivalent if and only if $q_1$ and $q_2$ are
(semi-)counting equivalent over~$\class{D}^\Sbf_{\Omc}$. The latter
can be decided using Proposition~\ref{prop:counequivdecclass}.

% \bigskip

% To prove Proposition~\ref{prop:OMQCEdec}, we observe that $Q_1$ and
% $Q_2$ are counting equivalent if and only if $\mn{ch}_\Omc(q_1)$ and
% $\mn{ch}_\Omc(q_2)$ are, and thus we can use the decision
% procedure for counting equivalence of CQs presented in
% \cite{countingPositiveQueries}.  We next make this precise, in
% parallel also considering the weaker notion of semi-counting
% equivalence. The observation is as follows.

% \smallskip

% In the proof of Proposition~\ref{prop:OMQCEdec}, we will work with
% databases of the form
% $\class{D}^\Sbf_{\Omc} = \{ \mn{ch}_{\Omc}(D) \mid D \text{ an
%   \Sbf-database}\}$ where \Omc is an ontology from
% $\class{G}\cap\class{FULL}$ and \Sbf a schema.  It is folklore and
% easy to see that such classes satisfy all relevant properties, see
% e.g.~\cite{faginSTOC80} for direct products.
% %
% \begin{lem}
%     \label{lem:closure}
%     For every ontology $\Omc \subseteq \class{TGD}$ and schema \Sbf such
%     that all symbols in \Omc are from \Sbf, $\class{D}^\Sbf_\Omc$ is
%     closed under direct product and contains
%     $D^\top_\Sbf$. If $\Omc \subseteq \class{G}$, then
%     $\class{D}^\Sbf_\Omc$ is additionally closed under disjoint
%     union and cloning of elements.

\section{Proof of Theorem~\ref{thm:main-results-OMQs}}

We prove the upper bounds in Theorem~\ref{thm:main-results-OMQs} by
Turing fpt-reductions to the corresponding upper bounds in
Theorem~\ref{thm:chenmengeldell}, and the lower bounds by Turing
fpt-reduction from the corresponding lower bounds in
Theorem~\ref{thm:chenmengeldell}. In both cases, the assumption that
the arity of relation symbols is bounded is only required for
Theorem~\ref{thm:chenmengeldell}, but not for the Turing
fpt-reductions that we give. Consequently, any future classifications
of $\mn{AnswerCount}(\class{C})$ for classes of CQs $\class{C}$ that
does not rely on this assumption also lift to classes of OMQs through
our reductions.

\subsection{Upper Bounds}

We first establish the upper bounds presented in
Theorem~\ref{thm:main-results-OMQs}. All these bounds are proved in
a uniform way, by providing a Turing FPT reduction from
$\mn{AnswerCount}(\class{Q})$, for any class
$\class{Q} \subseteq (\class{G},\class{UCQ})$ of OMQs, to
$\mn{AnswerCount}(\class{Q}^\star)$. It then remains to use
the corresponding upper bounds for classes of CQs from
Theorem~\ref{thm:chenmengeldell}.  For the reduction,
it is not necessary to assume that the
arity of relation symbols is bounded by a constant.

Let $\class{Q} \subseteq (\class{G},\class{UCQ})$ be a class of OMQs
with the full schema.  We need to exhibit an fpt algorithm for
$\mn{AnswerCount}(\class{Q})$ that has access to an oracle for
$\mn{AnswerCount}(\class{Q}^\star)$.
% such that the treewidths and the contract treewidths of CQs in
% $\class{Q}^\star$ % \mn{cl}^\exists_{\mn{CM}}(\class{C})$
% are bounded by a constant
% $k$.
%Given an OMQ $Q \in \class{C}$ and an
%\Sbf-database $D$, we first replace $Q$ by its $\exists$-rewriting
%$Q^\exists=(\Omc^\exists,\Sbf,q^\exists)$. Since $Q$ is equivalent to
%$Q^\exists$ and by the universality of the
%chase, %and by Lemma~\ref{pro:chase},
%$\#Q(D)=\#q^\exists(\mn{ch}_{\Omc^\exists}(D))$, thus it suffices
%to compute the latter count. Since $\Omc^\exists$ is from $\class{G}
%\cap \class{FULL}$, $\mn{ch}_{\Omc^\exists}(D)$ is finite and can be
%computed within the time requirements of FPT. By Lemma~\ref{lem:cm},
%we can compute $\#q^\exists(\mn{ch}_{\Omc^\exists}(D))$ within the
%time requirements of FPT once we have computed
%$\#p(\mn{ch}_{\Omc^\exists}(D))$ for all $p$ such that $(\Omc^\exists,\Sbf,p) \in \mn{cl}_{\mn{CM}}(q^\exists)$
%or, equivalently,
%$\#\mn{core}(\mn{ch}_{\Omc^\exists}(p))(\mn{ch}_{\Omc^\exists}(D))$ for
%all $p$ such that $(\Omc^\exists,\Sbf,p) \in \mn{cl}_{\mn{CM}}(q^\exists)$.
%
Let an OMQ $Q \in \class{Q}$ and an \Sbf-database
$D$ be given. The algorithm first replaces $Q$ by its
$\exists$-rewriting
$Q^\exists=(\Omc^\exists,\Sbf,q^\exists)$ as per
Theorem~\ref{thm:removeEx}. Equivalence of $Q$ and
$Q^\exists$ implies
$\#Q(D)=\#Q^\exists(D)$, and thus it suffices to compute the latter
count.

To compute
$\#Q^\exists(D)$ within the time requirements of FPT, we first
compute the set $\mn{cl}_{\mn{CM}}(Q^\exists)$, then for every
$Q' \in \mn{cl}_{\mn{CM}}(Q^\exists)$ we determine
$\#Q'(D)$ within the time requirements of FPT, and finally we combine the results
to
$\#Q(D)$ as per {\color{\highlightColourTwo}the following lemma, which is
  an immediate consequence of the definition of $\mn{cl}_{\mn{CM}}(Q)$.
\begin{lem}
  \label{lem:cmOMQ}
  For each $Q=(\Omc,\Sbf,q) \in (\class{G},\class{UCQ})$ and
  \Sbf-database $D$, $\#Q(D)$ can be computed in polynomial
  time from the counts $\#Q'(D)$, $Q' \in \mn{cl}_{\mn{CM}}(Q)$.
\end{lem}}
Note that we need to effectively compute
$\mn{cl}_{\mn{CM}}(Q^\exists)$, % can be effectively computed, %  and is
% independent of $D$,
which is possible by Corollary~\ref{prop:OMQCEdec} in the case that the schema is full.

Let $Q'=(\Omc^\exists,\Sbf,p) \in \mn{cl}_{\mn{CM}}(Q^\exists)$.
Since $\Omc^\exists$ is from $\class{G} \cap
 \class{FULL}$,
 $\mn{ch}_{\Omc^\exists}(D)$ can be computed within the
 time requirements of FPT by Lemma~\ref{lem:chasefpt}.
 % Moreover,
% $\#(\Omc^\exists,\Sbf,p)(D) =
% \#p(\mn{ch}_{\Omc^\exists}(D))$.
 To compute $\#Q'(D)$, we may thus
construct $\mn{ch}_{\Omc^\exists}(D)$ and then compute
$\#p(\mn{ch}_{\Omc^\exists}(D))$. Equivalently, we can compute
and use $\mn{core}(\mn{ch}_{\Omc^\exists}(p))$ in place of $p$.

It remains to note that the CQs $\mn{core}(\mn{ch}_{\Omc^\exists}(p))$, for
$(\Omc^\exists,\Sbf,p) \in \mn{cl}_{\mn{CM}}(Q^\exists)$, are exactly
the CQs from $\class{Q}^\star$.

% and thus their treewidths and contract
% treewidths are bounded by $k$. Consequently, we can apply the fpt
% algorithm from Point~1 of Theorem~\ref{thm:chenmengeldell} as a black
% box and overall attain FPT running time.

% The remaining upper bounds from Theorem~\ref{thm:main-results-OMQs}
% can be proved analogously, exploiting that easiness for \wOne and
% \wOneC is defined in terms of Turing fpt-reductions.

%\section{Lower Bounds: From UCQs to CQs}

\subsection{Lower Bounds: Getting Started}

We next turn towards lower bounds in
Theorem~\ref{thm:main-results-OMQs}, which we all consider in
parallel. Let $\class{Q} \subseteq (\class{G},\class{UCQ})$ be a class
of OMQs. We provide a Turing fpt-reduction from
$\mn{AnswerCount}(\class{C})$ to $\mn{AnswerCount}(\class{Q})$ for a
class of CQs $\class{C}$ such that if $\class{Q}$ satisfies the
preconditions in one of the four lower bounds stated in
Theorem~\ref{thm:main-results-OMQs} (in Points~(2) to~(5), respectively),
then $\class{C}$ satisfies the preconditions from the corresponding
point of Theorem~\ref{thm:chenmengeldell}. While the constructed class
of CQs $\class{C}$ is closely related to $\class{Q}^\star$, it is not
identical.

We in fact obtain the desired Turing fpt-reduction by composing three
Turing fpt-reductions.  The first reduction consists in
transitioning to the $\exists$-rewritings of the OMQs in the original
class. The second reduction enables us to consider OMQs that use CQs
rather than UCQs.\footnote{It is interesting to note in this context
  that the construction of $Q^\exists$ may produce a UCQ even if the
  original OMQ $Q$ uses a~CQ.} And in the third reduction, we remove
ontologies altogether, that is, we reduce classes of CQs to classes of
OMQs. % All these Turing fpt-reductions essentially preserve the semantic
% measures, so ultimately we get the desired lower bounds in
% Theorem~\ref{thm:main-results-OMQs} from the corresponding lower
% bounds in Theorem~\ref{thm:chenmengeldell}.
%
We start with the first reduction, which is essentially an immediate
consequence of
Theorem~\ref{thm:removeEx}. %This part does not rely on using the full data schema.
\begin{thm}
  \label{thm:existsfptred}
  Let $\class{Q} \subseteq (\class{G},\class{UCQ})$ be recursively
  enumerable and let
  $\class{Q}' \subseteq (\class{G} \cap \class{FULL},\class{UCQ})$ be
  the class of $\exists$-rewritings of OMQs from $\class{Q}$.
  There is a parsimonious fpt-reduction from
  $\mn{AnswerCount}(\class{Q}')$ to $\mn{AnswerCount}(\class{Q})$.
\end{thm}
\begin{proof}
%The proof of Theorem~\ref{thm:existsfptred} is straightforward:
  Given a $Q=(\Omc,\Sbf,q) \in \class{Q}'$ and an \Sbf-database $D$,
  find some $Q'=(\Omc',\Sbf,q') \in \class{Q}$ such that $Q$ is an
  $\exists$-rewriting of $Q'$ by recursively enumerating $\class{Q}$
  and exploiting that $\exists$-rewritings can be effectively computed
  as per Theorem~\ref{thm:removeEx} and OMQ equivalence is decidable
  in $(\class{G},\class{UCQ})$ \cite{BaBP18}. Then compute and return
  $\#Q'(D)$.
\end{proof}

\smallskip

% We remark that the semantic measures of $\class{C}$ in
% Theorem~\ref{thm:existsfptred} are exactly those of $\class{C'}$ as
% every $\exists$-rewriting is equivalent to the original OMQ it is
% obtained from.

% \smallskip

\subsection{Lower Bounds: From UCQs to CQs}
\label{sect:fromucqtocq}

The second reduction is given by the following theorem. Recall that
for any $\class{Q} \subseteq (\class{G},\class{UCQ})$,
the class $\mn{cl}_{\mn{CM}}(\class{Q})$ consists of OMQs from
$(\class{G},\class{CQ})$, that is, it only uses CQs
but no UCQs.
%
% It does not rely on
% assuming the full schema.
%
% An OMQ
% $Q=(\Omc,\Sbf,q) \in (\class{G} \cap \class{FULL},\class{UCQ})$ is
% \emph{reduced} if there are no distinct CQs $p_1,p_2$ in~$q$ such that
% the OMQs $(\Omc,\Sbf,p_1)$ and $(\Omc,\Sbf,p_2)$ are equivalent.
% We can generally assume that we work with classes of reduced OMQs
% only since every OMQ can effectively be converted into an equivalent
% reduced one given that OMQ containment in $(\class{G},\class{UCQ})$
% is decidable \cite{Andreas}.
%
\begin{thm}
%\begin{restatable}{thm}{thmUCQtoCQ}
\label{thm:from-UCQ-to-CQ-open-world}
Let $\class{Q} \subseteq (\class{G} \cap \class{FULL},\class{UCQ})$ be
a recursively enumerable class of %reduced
OMQs with full schema. % and relation symbols of bounded arity.
% , and
% let $\class{C}' \subseteq (\class{G} \cap \class{FULL},\class{CQ})$ be
% the class of OMQs
% $\{Q' \mid \exists Q\in \class{C} : Q' \in
% \mn{cl}_{\mn{CM}}(Q)\}$.
Then there is
a Turing fpt-reduction from $\mn{AnswerCount}(\mn{cl}_{\mn{CM}}(\class{Q}))$ to
$\mn{AnswerCount}(\class{Q})$.
%\end{restatable}
\end{thm}
%
% Note that contract treewidths and starsizes % and linked
% % matching numbers
% of $\class{C}$ are bounded if and only if they are bounded in
% $\class{C}'$, with the same bounds. Moreover, if the semantic
% treewidths of $\class{C}$ are unbounded, then are those of
% $\class{C'}$, as stated by the following lemma. We remark that
% semantic treewidth of OMQs in $\class{C}$ and $\class{C}'$ is still
% defined w.r.t.\ equivalent OMQs from $(\class{G},\class{UCQ})$, but by
% Proposition~\ref{prop:uniform} we could as well define it w.r.t.\
% equivalent OMQs from $(\class{G} \cap \class{FULL},\class{UCQ})$.
% %
% \begin{lem}
%   Let $Q=(\Omc,\Sbf,q) \in (\class{G} \cap \class{FULL},\class{UCQ})$
%   have semantic treewidth at least $k$. Then there is a CQ $p$ in $q$
%   such that the OMQ $(\Omc,\Sbf,p)$ has semantic treewidth at least
%   $k$. {\color{blue} if we go for CM-closure, then this has to be
%     replaced by something more sophisticated.}
% \end{lem}
% %
% \begin{proof}
%   Let $q=p_1 \vee \cdots \vee p_n$.
%   Assume to the contrary that every OMQ $Q_i := (\Omc,\Sbf,p_i)$, $ 1
%   \leq i \leq n$, has
%   semantic treewidth at most $\ell<k$. Then every $Q_i$ is equivalent
%   to $Q'_i := (\Omc,\Sbf,(p_i)^a_k)$, where $(p_i)^a_\ell$ is the
%   TW$_\ell$-approximation of $Q_i$. Moreover, the treewidth of each
%   $Q'_i$ is at most $\ell$. It is easy to see that $Q$ is equivalent
%   to $(\Omc,\Sbf, (p_1)^a_k \vee \cdots \vee (p_n)^a_k)$ and thus
%   $Q$ has semantic treewidth at most $\ell$, a contradiction.
% \end{proof}
% %
In \cite{countingPositiveQueries}, Chen and Mengel establish
Theorem~\ref{thm:from-UCQ-to-CQ-open-world} in the special case where
ontologies are empty. A careful analysis of their proof reveals that
it actually establishes something stronger, namely a Turing
fpt-reduction from
$\mn{AnswerCount}(\mn{cl}_{\mn{CM}}(\class{Q}),\class{D})$ to
$\mn{AnswerCount}(\class{Q},\class{D})$ for all classes of UCQs
$\class{Q}$ and all classes of databases $\class{D}$ that satisfy
certain natural properties.
% ; recall that $\mn{cl}_{\mn{CM}}(\class{Q})$
% is a class of CQs rather than UCQs.
This is important for us because
it turns out that the class of databases obtained by chasing
with an ontology from $\class{G} \cap \class{FULL}$ satisfies all
the relevant properties, and thus
Theorem~\ref{thm:from-UCQ-to-CQ-open-world} is a consequence
of Chen and Mengel's constructions. We now make this more precise.

%
% The \emph{well of positivity} for a schema \Sbf, denoted
% $D^\top_\Sbf$, is the \Sbf-database with a single
% constant $c$ defined as $D^\top_\Sbf = \{ R(c,\dots,c)
% \mid R \in \Sbf \}$.
%
% A UCQ $q$ over schema \Sbf is \emph{reduced} if all CQs in it are
% cores and there are no containments among the CQs in $q$, that is, $q$
% does not contain distinct CQs $p_1,p_2$ such that
% $p_1 \subseteq_\Sbf p_2$.
For a class of databases $\class{D}$ and a
CQ $q$, we use $\mn{cl}^{\class{D}}_{\mn{CM}}(q)$ to denote the
version of the Chen-Mengel closure that is defined exactly as
$\mn{cl}_{\mn{CM}}(q)$, except that all tests of counting equivalence
are over the class of databases $\class{D}$ rather than over the
class of all databases.
\newcommand{\thmchenMengelStronger}{
    Let $\class{D}$ be a class of databases over some schema \Sbf such
    that $\class{D}$ is closed under disjoint union, direct product,
    and contains $D^\top_\Sbf$.
    Then there is an algorithm that
    \begin{enumerate}

        \item takes as input a %reduced
        UCQ $q$, a CQ
        $p \in \mn{cl}^{\class{D}}_{\mn{CM}}(q)$, and a database
        $D \in \class{D}$, \\
        {\color{\highlightColour}subject to the promise that for all
            $p' \in \mn{cl}^{\class{D}}_{\mn{CM}}(q)$,
            there is an
            equality-free CQ $p''$ such that $D_{p''} \in \class{D}$,
            $D_{\hat p''} \in \class{D}$,
            and
            $p'$ and $p''$ are counting equivalent over $\class{D}$,}

        \item has access to an oracle for
            $\mn{AnswerCount}(\{q\},\class{D})$, to a procedure for
            enumerating~$\class{D}$, and to procedures for deciding
             counting equivalence and semi\=/counting equivalence
             between CQs over~$\class{D}$,

        \item  runs in time $f(||q||) \cdot p(||D||)$ with $f$ a
            computable function and $p$ a polynomial,

        \item outputs $\#p(D)$.
    \end{enumerate}
}
%\end{restatable}
\begin{thmC}[{\cite{countingPositiveQueries}}]
    \label{thm:chenMengelStronger}
    \thmchenMengelStronger
\end{thmC}
The difference between access to an oracle and access to procedures in
Point~(2) of Theorem~\ref{thm:chenMengelStronger} is that the running
time of the oracle does not contribute to~the running time of the
overall algorithm while the running time of the procedures~does.  When
used with the class $\class{D}$ of all databases,
Lemma~\ref{thm:chenMengelStronger} is simply the special case of
Theorem~\ref{thm:from-UCQ-to-CQ-open-world} where ontologies are
empty.  % We refrain from repeating the details of the proof of
% Theorem~\ref{thm:chenMengelStronger} given in
% \cite{countingPositiveQueries}. However,
In the appendix, we summarize the proof of
Theorem~\ref{thm:chenMengelStronger} given in
\cite{countingPositiveQueries}, showing that it works not only for the
class of all databases as considered in
\cite{countingPositiveQueries}, but also for all stated classes of
databases $\class{D}$.

Before we prove that Theorem~\ref{thm:chenMengelStronger} implies
Theorem~\ref{thm:from-UCQ-to-CQ-open-world}, we make the
following observation on Chen-Mengel closures.
\begin{lem}
  \label{lem:cmcloscoin}
  Let $Q=(\Omc,\Sbf,q) \in (\class{G} \cap \class{FULL},\class{UCQ})$
  and $\class{D}=\class{D}^\Sbf_{\Omc}$. Then
  $\mn{cl}^{\class{D}}_{\mn{CM}}(q)=\{ q' \mid (\Omc,\Sbf,q') \in
  \mn{cl}_{\mn{CM}}(Q)\}$.
\end{lem}
\begin{proof}
  The definitions of $\mn{cl}^{\class{D}}_{\mn{CM}}(q)$ and
  $\mn{cl}_{\mn{CM}}(Q)$ exactly parallel each other. In both cases,
  we build an equation based on the inclusion-exclusion principle,
  then manipulate it based on certain
  % please do NOT replace with emptiness
  counting equivalence tests, and then read off
  $\mn{cl}^{\class{D}}_{\mn{CM}}(q)$ resp.\ $\mn{cl}_{\mn{CM}}(Q)$
  from the result. The only difference is that the construction of
  $\mn{cl}_{\mn{CM}}(Q)$ uses OMQ $(\Omc,\Sbf,p)$ whenever the
  construction of $\mn{cl}^{\class{D}}_{\mn{CM}}(q)$ uses CQ $p$.  In
  particular, a counting equivalence test between two OMQs
  $Q_1=(\Omc,\Sbf,q_1)$ and $Q_2=(\Omc,\Sbf,q_2)$ % (on the class of all
  % databases)
  in the former case correspond to a counting equivalence test between
  $q_1$ and $q_2$ over the class of databases
  $\class{D}=\class{D}^\Sbf_{\Omc}$ in the latter case.  To prove
  Lemma~\ref{lem:cmcloscoin}, it thus suffices to show that these
  tests yield the same result. But this follows from the fact that
  $Q_i(D)=q_i(\mn{ch}_\Omc(D))$ for all \Sbf-databases $D$.
\end{proof}
We now argue that Theorem~\ref{thm:chenMengelStronger} implies
Theorem~\ref{thm:from-UCQ-to-CQ-open-world}. Thus let
$\class{Q} \subseteq (\class{G} \cap \class{FULL},\class{UCQ})$ be a
recursively enumerable class of %reduced
OMQs with full schema.  % and relation symbols of bounded arity.
We need
to give an fpt algorithm with an oracle for
$\mn{AnswerCount}(\class{Q})$ that, given an OMQ
$Q'=(\Omc,\Sbf,q') \in \mn{cl}_{\mn{CM}}(\class{Q})$ and an
\Sbf-database $D$, computes $\#Q'(D)$.  By~enumeration, we can find an
OMQ $Q=(\Omc,\Sbf,q) \in \class{Q}$ such that
$Q' \in \mn{cl}_{\mn{CM}}(Q)$. %   We may assume w.l.o.g.\ that $q$ is a
% reduced UCQ.
Lemma~\ref{lem:cmcloscoin} yields
$q' \in \mn{cl}^{\class{D}}_{\mn{CM}}(q)$ for
$\class{D}=\class{D}^\Sbf_{\Omc}$. By Lemma~\ref{lem:closure}, we may
thus invoke the algorithm from Theorem~\ref{thm:chenMengelStronger}
with $\class{D}=\class{D}^\Sbf_{\Omc}$, the UCQ $q$, CQ
$q' \in \mn{cl}^{\class{D}}_{\mn{CM}}(q)$, and the database
$\mn{ch}_{\Omc}(D)$. The algorithm returns
$\#q(\mn{ch}_{\Omc}(D))=\#Q'(D)$, as desired.
Note that $\mn{ch}_{\Omc}(D)$ is finite because
$\Omc \subseteq \class{FULL}$ and can be produced within the time
requirements of fixed-parameter tractability by
Lemma~\ref{lem:chasefpt}. {\color{\highlightColour}Also note that for every
  $p' \in \mn{cl}^{\class{D}}_{\mn{CM}}(q)$, we may use
  $\mn{ch}_\Omc(\tilde p')$ as the equality-free CQ $p''$ required
  by Point~(1) of Theorem~\ref{thm:chenMengelStronger}.
  In fact, $\tilde p'$ is counting equivalent to $p'$, even over the
  class of all databases, and $\mn{ch}_\Omc(\tilde p')$ is
  equivalent to $\tilde p'$ over $\class{D}=\class{D}^\Sbf_{\Omc}$
  by Lemma~\ref{lem:chaseoklem}.
}

We still need to argue that the oracle and procedures from Point~(2) of
Theorem~\ref{thm:chenMengelStronger} are indeed available.  As the
oracle for $\mn{AnswerCount}(\{q\},\class{D}^\Sbf_\Omc)$, we can use
an oracle for $\mn{AnswerCount}(\{Q\})$: by definition of
$\class{D}^\Sbf_\Omc$, any $D \in \class{D}^\Sbf_\Omc$ satisfies
$q(D)=Q(D)$. And as an oracle for $\mn{AnswerCount}(\{Q\})$, in turn,
we can clearly use the strictly more general oracle for
$\mn{AnswerCount}(\class{Q})$ that we have at our disposal in the
Turing fpt-reduction that we are building.  The procedure for
enumerating $\class{D}^\Sbf_\Omc$ required by Point~(2) is also easy to
provide. % The existence of the procedure for enumerating
% $\mn{cl}^{\class{D}}_{\mn{CM}}(q)$ is implied by
% Lemma~\ref{lem:cmcloscoin}
% and Proposition~\ref{prop:OMQCEdec} (and thus depends on the
% schema being full).
We can just enumerate all \Sbf-databases, chase with \Omc, and filter
out duplicates.  Finally, the procedures for deciding counting
equivalence and semi-counting equivalence of CQs over
$\class{D}^\Sbf_\Omc$ are provided by
Proposition~\ref{prop:counequivdecclass}.

\subsection{Lower Bounds: Removing Ontologies}
\label{subsect:removeonto}
% \label{sec:full-OMQ-to-CQ}
%
We next give the reduction that
% The aim of this section is to complete the proof of the lower bounds
% in Theorem~\ref{thm:main-results-OMQs} by providing the Turing
% fpt-reduction that
removes ontologies. % The lower bounds in
% Theorem~\ref{thm:main-results-OMQs} are a consequence of those
% in Theorem~\ref{thm:chenmengeldell} and the reductions in
% Theorems~\ref{thm:existsfptred},~\ref{thm:from-UCQ-to-CQ-open-world},
% and~\ref{thm:OMQ-to-CQ}.
%
% , as made precise by
% the following theorem.
%
\begin{thm}
    \label{thm:OMQ-to-CQ}
    Let $\class{Q} \subseteq (\class{G} \cap \class{FULL},\class{CQ})$ be
    a recursively enumerable class of %reduced
    OMQs with full schema. %  and
    % relation symbols of bounded arity.
    There is~a~class $\class{C} \subseteq \class{CQ}$ that only
    contains  cores and such
    that: % the following conditions are satisfied:
    \begin{enumerate} %\itemsep=0pt

    % \item the semantic treewidths, semantic contract treewidths, and
    %   semantic starsizes of CQs in $\class{C}'$ coincide with that of
    %   the
    %   OMQs in $\class{C}$; {\color{blue}I find this formulation quite
    %     vague; do we mean in the sense of sets here? I think we should
    %     make precise or just speak about boundedness}

    \item there is a Turing fpt-reduction from $\mn{AnswerCount}(\class{C})$ to $\mn{AnswerCount}(\class{Q})$;

    % \item for every OMQ $Q \in \class{C}$, we find a CQ
    %   $q \in \class{C}'$ that has the same semantic treewidth,
    %   semantic contract treewidth,and semantic starsize, and vice versa.

    \item
      % $\{ G_q \mid q \in \class{C}' \} = \{ G_{\mn{core}(\mn{ch}_\Omc(q))} \mid
      % (\Omc,\Sbf,q) \in \class{C} \}$.
      for every OMQ $Q= (\Omc, \Sbf, q) \in \class{Q}$, we find a CQ
      $p \in \class{C}$ such that $p$ and $\core(\chase_{\Omc}(q))$
      have the same Gaifman graph. % \footnote{They even have the same
        % hypergraph.}

%
      % the treewidth of $p$ %$\core(p)$
      % is~equal to~that of~$\core(\chase_{\Omc}(q))$, and likewise for
      % contract treewidth, starsize, and linked matching number.

    \end{enumerate}
\end{thm}
Before we prove Theorem~\ref{thm:OMQ-to-CQ}, we first show how we can
make use of the three Turing fpt-reductions stated as
Theorems~\ref{thm:existsfptred},~\ref{thm:from-UCQ-to-CQ-open-world},
and \ref{thm:OMQ-to-CQ}, to obtain the lower bounds in
Theorem~\ref{thm:main-results-OMQs} from those in
Theorem~\ref{thm:chenmengeldell}.  Let us consider, for example, the
\wOne lower bound from Point~(2) of Theorem~\ref{thm:main-results-OMQs}.
Take a class $\class{Q}_0 \subseteq (\class{G},\class{UCQ})$ of OMQs
such that the treewidths of CQs in
$$
\class{Q}^\star_0 = \{ \mn{core}(\mn{ch}_{\Omc^\exists}(p))\mid \exists Q \in \class{Q}_0: % Q^\exists
% = (\Omc^\exists,\Sbf,q^\exists) \text{ and }
(\Omc^\exists,\Sbf,p) \in
\mn{cl}_{\mn{CM}}(Q^\exists)\}
$$
are unbounded. % and the semantic
% contract treewidths of OMQs in $\mn{cl}^\exists_{\mn{CM}}(\class{C})$
% are bounded
Theorems~\ref{thm:existsfptred}
and~\ref{thm:from-UCQ-to-CQ-open-world} give a Turing fpt-reduction from $\mn{AnswerCount}(\class{Q})$
to $\mn{AnswerCount}(\class{Q}_0)$ where
$$
\class{Q} = \{ Q' \mid \exists Q \in \class{C}_0: Q' \in
\mn{cl}_{\mn{CM}}(Q^\exists)\}.
$$
By assumption, the treewidths of the CQs $\mn{core}(\mn{ch}_\Omc(q))$,
$(\Omc,\Sbf,q) \in \class{Q}$, are unbounded. Let $\class{C}$ be the
class of CQs whose existence is asserted by
Theorem~\ref{thm:OMQ-to-CQ}.  By Point~(2) of that theorem, the
treewidths of the CQs in $\class{C}$ are unbounded and thus
$\mn{AnswerCount}(\class{C})$ is \wOne-hard by Point~(2) of
Theorem~\ref{thm:chenmengeldell}. Composing the Turing fpt-reduction
from $\mn{AnswerCount}(\class{C})$ to $\mn{AnswerCount}(\class{Q})$
given by Point~(1) of Theorem~\ref{thm:OMQ-to-CQ} with the reduction
from $\mn{AnswerCount}(\class{Q})$ to $\mn{AnswerCount}(\class{Q}_0)$,
we obtain a Turing fpt-reduction from $\mn{AnswerCount}(\class{C})$ to
$\mn{AnswerCount}(\class{Q}_0)$ and thus the latter is \wOne-hard. The
other lower bounds can be proved analogously.

\medskip We now turn to the proof of Theorem~\ref{thm:OMQ-to-CQ} which
in turn uses three consecutive
fpt-reductions. % that are both inspired by
% constructions in~\cite{countingTrichotomy}.
The first reduction is easy and ensures that all involved CQs (inside
OMQs) are equality-free.  The second reduction allows us, informally
spoken, to mark every variable in a CQ (inside an OMQ) by a unary
relation symbol that uniquely identifies it. In the third reduction,
we make use of these markings to remove the ontology.
% We start with a
% preliminary. We shall sometimes replace a CQ $q$ in a UCQ
% $(\Omc,\Sbf,q) \in (\class{G} \cap \class{FULL},\class{CQ})$ with
% $\mn{core}(\mn{ch}_\Omc(q))$. Note that $\mn{ch}_\Omc(q)$ is finite
% since $\Omc \in \class{FULL}$. The following lemma explains our
% interest in ths replacement. Note the connection to
% Lemmas~\ref{lem:coresaregreat}
% and~\ref{not_yet_clear_approximations_in_upper_bound}.
% %
% \begin{lem}
% \label{lem:corechasecanonical}
%   Let $Q=(\Omc,\Sbf,q) \in (\class{G} \cap \class{FULL},\class{CQ})$
%   have full data schema.  Then the semantic treewidth of $Q$ agrees with
%   the syntactic treewidth of $\mn{core}(\mn{ch}_\Omc(q))$, and
%   likewise for contract treewidth and starsize.
% \end{lem}
% %
% \begin{proof}
%   We only consider starsize, the arguments for the other two measures
%   are identical.  Let $Q=(\Omc,\Sbf,q) \in (\class{G} \cap
%   \class{FULL},\class{CQ})$ have full data schema.  Since $Q$ is equivalent
%   to $(\Omc,\Sbf,\mn{core}(\mn{ch}_\Omc(q))$, the semantic starsize of
%   $Q$ cannot be larger than the starsize of
%   $\mn{core}(\mn{ch}_\Omc(q)$. It thus remains to argue that it also
%   cannot be larger. Assume to the contrary that $Q \equiv Q'$ for
%   some $Q'=(\Omc',\Sbf,q') \in (\class{G},\class{UCQ})$.
%   {\color{blue}Needs approximation stuff to derive $\Omc'=\Omc$; clean
%   up that part first.}
% \end{proof}
% %
% Now for the first step of the proof of Theorem~\ref{thm:OMQ-to-CQ}.
{\color{\highlightColour}
For the first reduction, recall that CQ $\tilde q$ is obtained from CQ $q$
by removing all equality atoms and identifying any two variables
$x_1,x_2$ with $x_1 = x_2 \in q$.
}
\begin{lem}
  \label{lem:removeequality}
  Let $\class{Q} \subseteq (\class{G} \cap \class{FULL},\class{CQ})$
  be a recursively enumerable class of %reduced
  \OMQs with full schema and let $\class{Q}^\sim=\{ (\Omc,\Sbf,\tilde q)
  \mid (\Omc,\Sbf,q) \in \class{Q}$. % and relation symbols of bounded arity.
    % , and let
    % $\class{C}^\marked \subseteq (\class{G} \cap
    % \class{FULL},\class{CQ})$ be the class of OMQs
    % $\{Q^\marked \mid
    % Q \in \class{C}\}$.
    Then there is a Turing
    fpt-reduction from $\mn{AnswerCount}(\class{Q}^\sim)$ to
    $\mn{AnswerCount}(\class{Q})$.
%\end{restatable}
\end{lem}
\begin{proof}
  Given a $Q=(\Omc,\Sbf,q) \in \class{Q}^\sim$ and an \Sbf-database $D$,
  find a $Q'=(\Omc,\Sbf,p) \in \class{Q}$ such that $q=\tilde p$ by
  recursively enumerating $\class{Q}$. Then compute and return
  $\#Q'(D)$. By construction of $q=\tilde p$, it is clear that $q$ and
  $p$ are counting equivalent. Consequently,
  $\#Q'(D) = \# p(\mn{ch}_\Omc(D)) = \#q(\mn{ch}_\Omc(D)) = \#Q(D)$.
\end{proof}
We next give the second reduction.  The \emph{marking} of a CQ $q$
over schema \Sbf is the CQ $q^\marked$ obtained from $q$ by adding an
atom $R_x(x)$ for each $x \in \mn{var}(q)$ where $R_x$ is a fresh unary
relation symbol. Note that $q^\marked$ is over schema $\Sbf^\marked$
obtained from \Sbf by adding all the fresh unary
symbols. %, and that it is a core.
The \emph{core-chased marking} of an OMQ
$Q=(\Omc,\Sbf,q) \in (\class{G} \cap \class{FULL}, \class{CQ})$ is the
OMQ
$Q^\marked =
(\Omc,\Sbf^{\marked},\mn{core}(\mn{ch}_\Omc(q))^{\marked}) \in
(\class{G}\cap \class{FULL}, \class{CQ})$. This can be lifted to
classes of OMQs $\class{Q}$ as expected, that is,
$\class{Q}^{\marked}=\{Q^\marked \mid Q \in \class{Q}\}$.
% We say that an~\OMQ $Q^\marked = (O,\Sbf^{\marked},q^{\marked})$ from
% $(\class{G}, \class{CQ})$ is~\emph{marked} if~for every variable~$x$
% in~$q$, there is~an~atom $R_{x}(x)$ and the symbol~$R_{x}$ is not used
% in \Omc or otherwise in $q$. We refer to the added atoms as
% \emph{markings}.  We add the marked symbol $^\marked$ to emphasise
% that a query is marked or a schema is enhanced by markings.  In
% particular, if the query is clear from the context by $\Sbf^\marked$
% we understand the schema $\Sbf$ with the appropriate markings added.
%
\begin{lem}
%\begin{restatable}{lem}{lemCQtomarked}
    \label{lemma:CQ-to-marked-CQ}
    Let $\class{Q} \subseteq (\class{G} \cap \class{FULL},\class{CQ})$
    be a recursively enumerable
    class of {\color{\highlightColour}equality-free} %reduced
    \OMQs with full schema. % and relation symbols of bounded arity.
    % , and let
    % $\class{C}^\marked \subseteq (\class{G} \cap
    % \class{FULL},\class{CQ})$ be the class of OMQs
    % $\{Q^\marked \mid
    % Q \in \class{C}\}$.
    Then there is a Turing
    fpt-reduction from $\mn{AnswerCount}(\class{Q}^\marked)$ to
    $\mn{AnswerCount}(\class{Q})$.
%\end{restatable}
\end{lem}
To prove Lemma~\ref{lemma:CQ-to-marked-CQ}, we again adapt a reduction
by Chen and Mengel that addresses the case of CQs without
ontologies, but that can be lifted to relevant
classes of databases similarly to
Theorem~\ref{thm:chenMengelStronger}.

\newcommand{\lemCQtoMarked}{
    Let $\class{D}$ be a class of databases over schema $\Sbf^\marked$ that is closed under
    direct products, cloning, and induced subdatabases. %disjoint union,
    Then there is~an~algorithm that
    \begin{itemize}
        \item takes as input an equality-free CQ $q$ such that
            $q^\marked$ is over schema $\Sbf^\marked$ and a database $D \in
            \class{D}$, %{\color{blue}why not give it $q^\marked$??}
            subject to the promise that $q$ is a core and $D_{q^\marked}
            \in\class{D}$, %{\color{cyan}no $\sim$ here since $q$ is equality-free}
        \item has access to an oracle for
            $\mn{AnswerCount}(\{q\},\class{D})$,\footnote{Note that since
            $\Sbf \subseteq \Sbf^\marked$, $q$ may be viewed as a CQ over
            schema $\Sbf^\marked$.}
        \item runs in time $f(||q||)
            \cdot p(||D||)$, $f$ a computable function and $p$ a
        polynomial, and
        \item outputs $\#q^\marked(D)$.
    \end{itemize}
}

\begin{thmC}[\cite{countingTrichotomy}]
    \label{lemma:CQ-to-marked-CQ-byChenMengel}
    \lemCQtoMarked
\end{thmC}
When used with the class $\class{D}$ of all databases,
Lemma~\ref{lemma:CQ-to-marked-CQ-byChenMengel} is simply the special
case of Lemma~\ref{lemma:CQ-to-marked-CQ} where ontologies are empty.
In the appendix, we give an overview of the proof of
  Lemma~\ref{lemma:CQ-to-marked-CQ-byChenMengel} in
  \cite{countingTrichotomy}, also showing that it extends to classes
  of databases $\class{D}$ that satisfy the stated properties.

  We now use Lemma~\ref{lemma:CQ-to-marked-CQ-byChenMengel} to prove
  Lemma~\ref{lemma:CQ-to-marked-CQ}.  Let
  $\class{Q} \subseteq (\class{G} \cap \class{FULL},\class{CQ})$ be a
  recursively enumerable class of equality-free %reduced
  \OMQs with full schema. % and relation symbols of bounded arity.
  We give an fpt algorithm that uses $\mn{AnswerCount}(\class{Q})$ as
  an oracle and, given an OMQ
  $Q^\marked(\bar
  x)=(\Omc,\Sbf^\marked,\mn{core}(\mn{ch}_\Omc(q))^\marked) \in
  \class{Q}^\marked$ and an~$\Sbf^\marked$\=/database~$D$, computes
  $\#Q^\marked(D)$.

First, the algorithm enumerates $\class{Q}$ to find an OMQ $Q(\bar x)$
such that $Q^\marked$ is the core\=/chased marking of $Q$, that is,
$Q=(\Omc,\Sbf,q)$. % with $q=\mn{core}(\mn{ch}_\Omc(p))^\marked$.
It then
starts the algorithm from
Lemma~\ref{lemma:CQ-to-marked-CQ-byChenMengel} for the class of
databases
$$
    \class{D}^{\Sbf^\marked}_\Omc = \{ \chase_\Omc(D') \mid \text{$D'$ an $\Sbf^\marked$-database} \},
$$
and with the CQ $\mn{core}(\mn{ch}_\Omc(q))$ and the database
$\chase_{\Omc}(D)$ as the input.  The algorithm
outputs
$\#\mn{core}(\mn{ch}_\Omc(q))^\marked(\chase_\Omc(D))=\#Q^\marked(D)$,
as required.

We should argue that the preconditions of
Lemma~\ref{lemma:CQ-to-marked-CQ-byChenMengel} are satisfied.  By
Lemma~\ref{lem:closure}, $\class{D}^{\Sbf^\marked}_\Omc$ is closed
under direct products % , disjoint union,
and cloning. Since the ontologies in $\class{Q}$ are from
$\class{FULL}$, $\class{D}^{\Sbf^\marked}_\Omc$~is also closed under
induced subdatabases. Moreover, class $\class{D}^{\Sbf^\marked}_\Omc$
contains $D_{\mn{core}(\mn{ch}_\Omc(q))^\marked}$ since
$\mn{core}(\mn{ch}_\Omc(q))^\marked=
\mn{ch}_\Omc(\mn{core}(\mn{ch}_\Omc(q))^\marked)$ and the schema is
full.  As the oracle for
$\mn{AnswerCount}(\{\mn{core}(\mn{ch}_\Omc(q))\},
\class{D}^{\Sbf^\marked}_\Omc)$ needed by the algorithm, we can use
the oracle for $\mn{AnswerCount}(\class{Q})$ that we have available,
as follows.

Given a database $D \in \class{D}^{\Sbf^\marked}_\Omc$, we first
construct database $D|_\Sbf$ by dropping all atoms that use a symbol
from $\Sbf^\marked \setminus \Sbf$ and then ask the oracle for
$\mn{AnswerCount}(\class{Q})$ to return $\#Q(D|_\Sbf)$. We argue that
this is the same as the required $\#\mn{core}(\mn{ch}_\Omc(q))(D)$. In
fact,
$$\#Q(D|_\Sbf)=\#q(\mn{ch}_\Omc(D|_\Sbf))=\#q(D|_\Sbf)=\#q(D)=\#\mn{core}(\mn{ch}_\Omc(q))(D).$$
The first equality is due to the universality of the chase. For the
second equality, recall that $D \in \class{D}^{\Sbf^\marked}_\Omc$ and
is thus of the form $D=\mn{ch}_\Omc(D')$ with $D'$ an
$\Sbf^\marked$-database. Since \Omc does not use the symbols from
$\Sbf^\marked \setminus \Sbf$, this implies the second equality.  The
third equality holds because $q$ does not use the symbols from
$\Sbf^\marked \setminus \Sbf$. And the final equality holds because
$D=\mn{ch}_\Omc(D')$ and thus any homomorphism from $q$ to $D$ is also
a homomorphism from $\mn{ch}_\Omc(q)$ to $D$.  Moreover, taking the
core produces an equivalent CQ.  %
    % to compute
    %   $\#p^{\marked}(\chase_{O}(D))$.
    %   This can be done because the class $\class{D} = \{ \chase_{O}(D) \mid D \text{ is an $\Sbf^\marked$\=/database}\}$,
    %   a class of chased $\Sbf^\marked$\=/databases, is closed under product, cloning, disjoint unions
    %   and we have that $\chase_{O}(D_{q}) \in \class{D}$.
    %   Finally, the algorithm returns the value $\#p^{\marked}(\chase_{O}(D))$.
    %
    %   The correctness of the algorithm stems from the fact that we can compute the chase in fpt and from the following equations.
    %   \[
    %   \#p^{\marked}(\chase_{O}(D)) = \#(\Omc,\Sbf^\marked,p^\marked)(\chase_{O}(D)) = \#(\Omc,\Sbf^\marked,p^\marked)(D)
    %   \]

\medskip

Now for the third fpt-reduction that we use in the proof of
Theorem~\ref{subsect:removeonto}. It facilitates that with the presence
of~markings it is possible to remove ontologies, in the following
sense.
\begin{lem}
%\begin{restatable}{lem}{lemrmont}
    \label{lemma:marked-CQ-to-plainCQ-redux}
    Let $\class{Q} \subseteq (\class{G} \cap \class{FULL},\class{CQ})$
    be a recursively enumerable class of {\color{\highlightColour}equality\=/free} %reduced
    OMQs with full schema and $\class{Q}^\marked$ their core-chased
    markings. %and relation symbols of bounded arity
    % that are core-chased markings.
    There exists a class
    $\class{C}\subseteq \class{CQ}$ of cores with the arities
    of
    relation symbols identical to those in $\class{Q}$ such that:
\begin{enumerate}

\item there is a Turing fpt\=/reduction from
  $\mn{AnswerCount}(\class{C})$ to
  $\mn{AnswerCount}(\class{Q}^\marked)$;

\item $\class{C}$ is based on the same Gaifman graphs as $\class{Q}^\marked$:
$\{ G_q \mid q \in \class{C} \} = \{ G_q \mid
(\Omc,\Sbf,q) \in \class{Q}^\marked \}$.

\end{enumerate}
%\end{restatable}
\end{lem}
We provide a proof of Lemma~\ref{lemma:marked-CQ-to-plainCQ-redux}
below. Before, however, we show how
Theorem~\ref{thm:OMQ-to-CQ} follows from
Lemmas~\ref{lemma:CQ-to-marked-CQ} and~\ref{lemma:marked-CQ-to-plainCQ-redux}.
\begin{proof}[Proof of Theorem~\ref{thm:OMQ-to-CQ}]
  Let $\class{Q} \subseteq (\class{G} \cap \class{FULL},\class{CQ})$
  be a recursively enumerable class of %reduced
  OMQs with full schema. % and relation symbols of bounded  arity.
  % , and let $\class{C}^\marked$ be the class of core-chased markings of
  % OMQs from $\class{C}$, that is, $\class{C}^\marked=\{ Q^\marked\mid Q \in
  % \class{C}\}.$% We observe the following.
%  %
%  \\[2mm]
%  {\bf Claim.} For every $Q=(\Omc,\Sbf,q) \in \class{C}$, the
%  syntactic treewidth of the CQ $\mn{core}(\mn{ch}_\Omc(q))^\marked$ in
%  $Q^\marked$ is no smaller than the semantic treewidth of $Q$, and likewise
%  for contract treewidth and starsize.
%  \\[2mm]
%  %
%  In fact, assume that the syntactic treewidth of
%  $\mn{core}(\mn{ch}_\Omc(q))^\marked$ is $k$.%  Then its syntactic treewidth
%  % is $k$ since $\mn{core}(\mn{ch}_\Omc(q))^\marked$ is a core.
%  Marking does
%  not affect the Gaifman graph and thus does not change treewidth.
%  Consequently, the treewidth of $\mn{core}(\mn{ch}_\Omc(q))$ is $k$.
%  Since $Q$ is equivalent to $(\Omc,\Sbf,\mn{core}(\mn{ch}_\Omc(q)))$,
%  the semantic treewidth of $Q$ is at most $k$. We can argue in the
%  same way for contract treewidth and starsize.
%% , see
%%   Lemma~\ref{lem:coresaregreat}.
%
  From Lemma~\ref{lemma:marked-CQ-to-plainCQ-redux}, we obtain a class
  $\class{C}$ of CQs that are cores and are based on the same Gaifman
  graphs as $\class{Q}^\marked$. This is the class whose existence is
  postulated by Theorem~\ref{thm:OMQ-to-CQ}. We argue that Points~(1)
  and~(2) of that theorem are satisfied.  The Turing fpt-reduction
  required by Point~(1) is the composition of~the reductions asserted
  by~Lemmas~\ref{lemma:marked-CQ-to-plainCQ-redux},
  \ref{lemma:CQ-to-marked-CQ}, and~\ref{lem:removeequality}.
  Point~(2) is a consequence of the facts that $\class{C}$ is based on
  the same Gaifman graphs as $\class{Q}^\marked$ and neither does
  marking a CQ affect its Gaifman graph nor does the transition from a
  CQ $q$ to $\tilde q$. To see the latter, recall that the same
  variable identifications that take place when constructing
  $\tilde q$ from $q$ are also part of the definition of the Gaifman
  graph $D_q$ of $q$.
\end{proof}

%\section{Proof of Lemma~\ref{lemma:marked-CQ-to-plainCQ-redux}}

Now for the announced proof of
Lemma~\ref{lemma:marked-CQ-to-plainCQ-redux}, a key
ingredient to the proof of Theorem~\ref{thm:main-results-OMQs}.
\begin{proof}[Proof of Lemma~\ref{lemma:marked-CQ-to-plainCQ-redux}]
  To prove the lemma, we define the required class of CQs $\class{C}$ and describe
an~fpt algorithm that
%    \begin{itemize} %\itemsep=0pt
%    \item
takes as an input a~query $q \in \class{C}$
          over schema $\Sbf$ and an~$\Sbf$\=/database~$D$,
          % \item
          has access to an oracle for
          $\mn{AnswerCount}(\class{Q}^\marked)$, and
          % \item
          outputs $\#q(D)$.
    %\end{itemize}
    %
    Every $Q=(\Omc,\Sbf,q) \in \class{Q}$ gives rise to a CQ
    $q^\newschema$ in $\class{C}$ that is formulated in a schema different
    from \Sbf (whence the superscript `$\newschema$').  To define
    $q^\newschema$, fix a total order on $\mn{var}(q)$. For
    every guarded set $S$ in $D_q$, let $\overline{S}$ be the tuple
    that contains the variables in $S$ in the fixed order.  Now
    $q^\newschema$ contains, for every maximal guarded set $S$ in
    $D_q$, the atom $R_S(\overline{S})$ where $R_S$ is a fresh
    relation symbol of arity $|S|$.
    Note that $q^\newschema$ is self-join free, that is, it contains no
    two distinct atoms that use the same relation symbol. It is thus
    a core. Moreover, the Gaifman graph of $q^\newschema$ is
    identical to that of $q^\marked$ since the maximal guarded sets of
    $D_{q^\marked}$ are exactly those of
    $D_{q^\newschema}$.
    An example of a transformation from~$q$ to  $q^\newschema$ can be found in Figure~\ref{fig:transf-q-to-q-sharp}.% Let $\Sbf^\sharp$ denote the
    % schema if $q^\sharp$.
    % The schema of
    % $q^\sharp$ is \Sbf, that is, $\Sbf^\marked$ without the unary symbols
    % that have been introduced when producing the core-chased marking
    % $Q$ from an original OMQ.
    %
    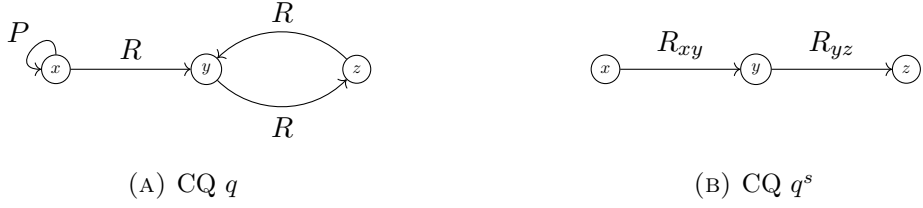
\begin{figure}
    \tikzstyle{quantified}=[circle, draw, scale=.6]
    \tikzstyle{answer} = [circle, fill, scale=.6]
    \centering
    \begin{subfigure}[b]{0.49\textwidth}
        \centering
        %        \resizebox{\linewidth}{!}{
        \begin{tikzpicture}%[scale=1.]
        
        \node (rt) at (2,1) {};
        \node (lb) at (-2,-1) {};
        
        \node[quantified] (x) at (-2,0) {$x$};
        \node[quantified] (y) at (0,0) {$y$};
        \node[quantified] (z) at (2,0) {$z$};
        \node (t) at (-2.5,0.5) {$P$};
        
        \draw [->] (x) to [out=90,in=180,loop,looseness=4.8] (x);
        
        \draw[->] (x) -- (y) node[midway,above] {$R$};
        \draw[->] (y) to [bend right=45] node[midway,below] {$R$} (z) ;
        \draw[->] (z) to [bend right=45]  node[midway,above] {$R$} (y);
        
        \end{tikzpicture}
        \caption{CQ $q$}
        %    }
    \end{subfigure}
    %%%%%%%%%%%%%%%%%%%%%%%%%%%%%%%%%%%%%%%%%%%%%%%%%%%%%%%%%%%%%%%%%%%%
    \begin{subfigure}[b]{0.49\textwidth}
        \centering
        %        \resizebox{\linewidth}{!}{
        \begin{tikzpicture}%[scale=1.]
        
        \node (rt) at (2,1) {};
        \node (lb) at (-2,-1) {};
        
        \node[quantified] (x) at (-2,0) {$x$};
        \node[quantified] (y) at (0,0) {$y$};
        \node[quantified] (z) at (2,0) {$z$};

        \draw[->] (x) -- (y) node[midway,above] {$R_{xy}$};
        \draw[->] (y) -- (z) node[midway,above] {$R_{yz}$};

        \end{tikzpicture}
        \caption{CQ $q^{\newschema}$}
        %    }
    \end{subfigure}
    \caption{A CQ $q$ and its self-join free counterpart
        $q^{\newschema}$.}
    \label{fig:transf-q-to-q-sharp}
    %    \begin{subfigure}[b]{0.24\textwidth}
    %        \centering
    %%        \resizebox{\linewidth}{!}{
    %        \begin{tikzpicture}%[scale=1.]
    %        
    %        
    %        \node[quantified] (z1) at (-1,0);
    %        \end{tikzpicture}
    %        \caption{CQ $p_2$}
    %%    }
    %    \end{subfigure}
\end{figure}

    This defines the class of CQs $\class{C}$.

    We now describe the algorithm. Let a CQ
    $q^{\newschema} \in \class{C}$ over schema $\Sbf^\newschema$ and
    an $\Sbf^{\newschema}$\=/database $D^\newschema$ be given as
    input.  To compute $\#q^{\newschema}(D^\newschema)$, we first
    enumerate $\class{Q}$ to find an OMQ $Q=(\Omc,\Sbf,q)$ such that
    $q^\newschema$ can be obtained from $q$ as described above.

    Construct the $\Sbf^\newschema$-database
    $P=D_{q^\newschema} \times D^\newschema$ and then from $P$ the
    $\Sbf^\marked$-database
    $$
    \begin{array}{r@{\;}c@{\;}l}
    D^\marked &=& \{ R(\bar a) \mid
                  \bar a \text{ tuple over some guarded set $S$ in } P
                  \text{ and }
                   R \in \Sbf \text{ of arity } |\bar a|
    \}\ \cup \\[1mm]
    && \{ R_x((x,a)) \mid x \in \mn{var}(q^\marked) \text{ and } (x,a) \in
    \mn{adom}(P) \}
    \end{array}
    $$
    where a tuple is \emph{over set $S$} if it contains only constants
    from $S$, in any order and possibly with repetitions. Intuitively,
    the first line `floods' the database with facts without creating
    fresh guarded sets, by adding all possible facts that use a
    relation symbol from \Sbf and only constants from some guarded set
    in $P$. % We can thus not add further \Sbf-facts
    % without introducing additional guarded sets.
%    Also note that the maximal guarded sets of $D^\marked$ are
 %   exactly those of $P$.
    As a consequence and since \Omc is a set of guarded TGDs,
    $D^\marked= \mn{ch}_\Omc(D^\marked)$.  The relations $R_x$ used in
    the second line are the marking relations from $\Sbf^\marked$.

    Clearly, the databases $P$ and $D^\marked$ can be constructed within
    the time requirements of FPT and we
    can use the oracle to compute $\#Q(D^\marked)$. Let
    $q^{\newschema,\marked}$ be obtained from $q^\newschema$ by adding
    $R_x(x)$ for every $x \in \mn{var}(q^\newschema)$ and let $P^\marked$
    be obtained from $P$ by adding $R_x(x,a)$ for every
    $a \in \mn{adom}(D^\newschema)$.  To end the proof, it suffices to
    show that
    \[
      \#Q(D^\marked) = \#q^\marked(D^\marked)=\#q^{\newschema,\marked}(P^\marked) =
      \#q^\newschema(D^\newschema).
    \]
    The above equalities, as well as the construction of the involved
    databases and queries, are illustrated in
    Figure~\ref{fig:marked-to-queries}. The figure also shows some
    homomorphisms used in the remaining proof.

    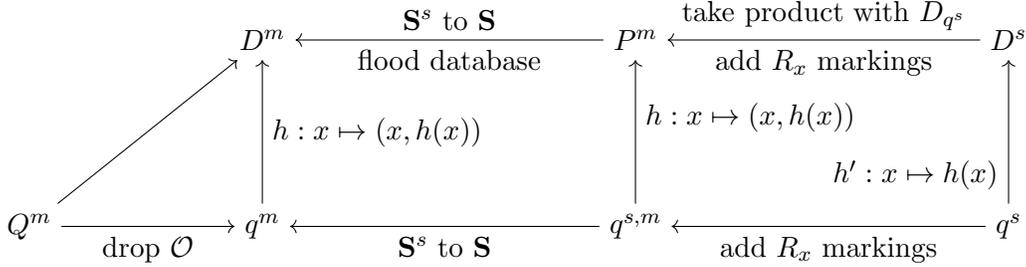
\begin{figure}
    %\begin{wrapfigure}{r}{0.45\textwidth}
    \centering
    \newcommand{\spod}{-6.5}
    \begin{tikzpicture}[scale=0.62]
    
    \node (mD) at (-8,0) {$D^\marked$};
    \node (qD) at (0,0) {$P^\marked$};
    \node (D) at (8,0) {$D^{\newschema}$};
    
    \node (omq) at (-13,-4) {$Q^{\marked}$};
    \node (mQ) at (-8,-4) {$q^{\marked}$};
    \node (q) at (0,-4) {$q^{\newschema,\marked}$};
    \node (sjfQ) at (8, -4) {$q^{\newschema}$};

    \draw[<-] (mD) -- (qD) node[midway,below] {flood database}
    node[midway,above] {$\Sbf^\newschema$ to \Sbf};
    \draw[<-] (qD)-- (D) node[midway,above]{take product with $D_{q^{\newschema}}$} node[midway,below] {add $R_x$ markings};
    \draw[->] (omq) -- (mQ) node[midway,below] {drop \Omc %$D^\marked=
      % \mn{ch}_\Omc(D^\marked)$
    };
    \draw[<-] (mQ) -- (q) node[midway, below]{$\Sbf^\newschema$ to $\Sbf$}%drop non-maximal
    % guarded sets}
    node[midway, above]{};
    \draw[<-] (q)-- (sjfQ) node[midway, below]{add $R_x$ markings};

    \draw[->] (omq) -- (mD);
    \draw[->] (mQ) -- (mD) node[midway,right] {$h:x \mapsto  (x,h(x))$};
    \draw[->] (q) -- (qD) node[right, pos=0.6] {$h:x \mapsto (x,h(x))$};
    \draw[->] (sjfQ) -- (D) node[left, pos=0.2] {$h':x \mapsto h(x)$};

    \end{tikzpicture}
    \caption{Proof strategy of Lemma~\ref{lemma:marked-CQ-to-plainCQ-redux}.}
    \label{fig:marked-to-queries}
    %\end{wrapfigure}
\end{figure}

    % \medskip
    %
    % To prove the crucial equality $\#q^{\sharp}(D) =
    % \#(\Omc,\Sbf,q^\marked)(D^\marked)$, first observe that
    % $D^{\marked}$ is chased.  Indeed, for every rule of form
    % $\psi(\bar{x},\bar{y}) \to R(\bar{x})$, we have a~guard
    % $G(\bar{x},\bar{y})$.  Therefore for every tuple $\bar{a}$ such
    % that there is a tuple $\bar{b}$ such that $\psi(\bar{a},\bar{b})$
    % holds, we have that $G(\bar{a} r\cdot \bar{b}) \in D^{\marked}$
    % and, by Equation~\eqref{eq:marked-to-sjf-interpretation}, we get
    % that $R(\bar{a}) \in D^{\marked}$.

    The first equality is immediate since
    $D^\marked= \mn{ch}_\Omc(D^\marked)$. For the third equality, let
    $\bar x = x_1 \cdots x_n$ be the answer variables in
    $q^\newschema$ and for any
    $\bar a =a_1 \cdots a_n \in \mn{adom}(D^\newschema)^n$, let
    $\bar x \times \bar a$ denote the tuple
    $(x_1,a_1)\cdots(x_n,a_n) \in \mn{adom}(P)^n$. Then
    $q^{\newschema,\marked}(P^\marked)=\{\bar x\times \bar a \mid \bar
    a \in q(D^\newschema)\}$.
%    $q^\marked(P)=\{\bar x \times \bar a \mid \bar a \in
%    q(D^\newschema)\}$
%    where $\bar x$ are the answer variables of $q^\newschema$.
%    This
    In fact, if $h$ is a homomorphism from $q^{\newschema,\marked}$ to
    $P^\marked$ and $x \in \mn{var}(q^{\newschema,\marked})$, then
    $h(x) \in \{ x \} \times \mn{adom}(P^\marked)$ due to the use of
    the marking relation $R_{x}$ in $q^{\newschema,\marked}$ and in
    $P^\marked$. Moreover, every such homomorphism $h$ gives rise to a
    homomorphism $h'$ from $q^\newschema$ to $D^\newschema$ by setting
    $h'(x)=c$ if $h(x)=(x,c)$, for all $x \in \mn{var}(q^\newschema)$.
    Conversely, every homomorphism $h$ from $q^\newschema$ to
    $D^\newschema$ gives rise to a homomorphism $h'$ from
    $q^{\newschema,\marked}$ to $P^\marked$ by setting
    $h'(x)=(x,h(x))$ for all $x \in \mn{var}(q^{\newschema,\marked})$.
%
%     this follows from $P$ being the product of $D_{q^\newschema}$ and
% $D^\newschema$ and from how the relation symbols $R_x$ are used in
% $q^{\newschema,\marked}$ and $P^\marked$.

    % using standard properties of products and the fact that any
    % homomorphism $h$ from $q^{\sharp,\marked}$ to $P$ must satisfy
    % $h(x) \in \{ x \} \times \mn{adom}(D^\sharp)$ due to the atom
    % $R_x(x) \in q^{\sharp,\marked}$, for all
    % $x \in \mn{var}(q^\sharp)$.

It thus remains to deal with the second equality by showing that
$q^\marked(D^\marked)=q^{\newschema,\marked}(P^\marked)$.  It is enough to observe
that any function \fun{h}{\mn{var}(q^\marked)}{\dom({D^{\marked}})} is
a~homomorphism from $q^\marked$ to $D^{\marked}$ if and only if it is
a homomorphism from $q^{\newschema,\marked}$ to~$P^\marked$.

For the ``if'' direction, let $h$ be a~homomorphism from
$q^{\newschema,\marked}$ to $P^\marked$. First let $R(\bar y)$ be an atom
in $q^\marked$ with $R \in \Sbf$. There is a maximal guarded set $S$
of $D_{q^\marked}$ that contains all variables in $\bar y$. Then
$R_S(\overline{S})$ is an atom in $q^\newschema$ and thus
$R_S(h(\overline{S}))\in P$. By construction of $D^\marked$ and since
$\bar y$ is a tuple over $S$, this yields
$R(h(\bar y)) \in D^\marked$, as required. Now let $R_x(x)$ be an atom
in $q^\marked$. Then $R_x(x)$ is also an atom in $q^{\newschema,\marked}$ and
thus $h(x) \in \{ x \} \times \mn{adom}(D^\newschema)$ due to the
definition of $P^\marked$. But then $R_x(h(x)) \in D^\marked$ by definition
of $D^\marked$.

For the ``only if'' direction, let $h$ be a~homomorphism from
$q^\marked$ to $D^{\marked}$. First consider atoms $R_S(\overline{S})$
in $q^{\newschema,\marked}$.  Then $q^\marked$ contains an atom
$R(\bar y)$ where $\bar y$ contains exactly the variables in $S$ and
thus $R(h(\bar y)) \in D^\marked$.  By construction of $D^\marked$,
$h(\bar y)$ is thus a tuple over some guarded set in $P$, that is, $P$
contains an atom $Q(\bar a)$ where $\bar a$ contains all constants
from $h(\bar y)$. In the following, we show that $Q(\bar a)$ must in
fact be $R_S(h(\overline{S}))$, as required.

Let $\bar a = (z_1,c_1),\dots,(z_n,c_n)$ and $\bar z=z_1, \dots, z_n$.
By construction of $P$ as $D_{q^\newschema} \times D^\newschema$,
$Q(\bar a) \in P$ implies that $q^\newschema$ contains an atom
$Q(\bar z)$. It suffices to show that $\bar z$ contains all variables
from $S$: since the construction of $q^\newschema$ uses as $S$ only
\emph{maximal} guarded sets, the only such atom in $q^\newschema$ is
$R_S(\overline{S})$. By construction of $P$, we must thus have
$Q(\bar a)=R_S(h(\overline{S}))$.

Let $V$ be the variables in $\bar z$. Since
$Q(\bar z) \in q^\newschema$, $V$ is a guarded set in $q^\newschema$.
Now note that we must have
$h(y) \in \{ y \} \times \mn{adom}(D^\newschema)$ for every variable
$y$ in $\bar y$ due to the use of the relation symbol $R_y$ in
$q^\marked$ and $D^\marked$. Since $\bar a$ contains all constants
from $h(\bar y)$, every variable from $\bar y$ occurs in
$V$. Moreover, these are exactly the variables in $S$ and thus
$S \subseteq V$.
%
% By definition of $V$, it follows that
% $\bar a$ uses all constants from $h(\bar y)$ in the second components
% of constants/pairs.  Also, the only atom that uses all variables from
% $S$ in $q^\newschema$ is $R_S(\overline{S})$ and consequently
% $Q(\bar a)$ must be $R_S(h(\overline{S}))$. It remains to deal with
% atoms $R_x(x)$ in $q^{\newschema,\marked}$, which is straightforward
% as in the ``if'' direction.
\end{proof}

\section{Approximation and FPTRASes}
\label{sect:approx-fpras}

In many applications of answer counting, it suffices to produce a good
approximation of the exact count. For CQs without ontologies,
significant progress on approximate answer counting has recently been
made by Arenas et al.\ \cite{counting-cq-approx}, see also
\cite{DBLP:journals/corr/abs-2103-12468} for follow-up work.  We
observe some important consequences for approximately counting the
number of answers to ontology-mediated queries.

A \emph{randomized approximation scheme} for a counting problem
$P: \Lambda^* \rightarrow \mathbb{N}$ is a randomized
algorithm that takes as input a word $w \in \Lambda^*$ and an
approximation factor $\epsilon \in (0,1)$ and outputs a value
$v \in \mathbb{N}$ such that
$$
  \mn{Pr}(|P(w)-v| \leq \epsilon \cdot P(w)) \geq \frac{3}{4}.
$$
A \emph{fixed-parameter tractable randomized approximation scheme
  (FPTRAS)} for a parameterized counting problem $(P,\kappa)$ over
alphabet $\Lambda$ is a randomized approximation scheme for the
counting problem $P$ with running time at most
$f(\kappa(w)) \cdot p(|w|,\frac{1}{\epsilon})$ for some computable
function $f$ and polynomial $p$. The results proved in
\cite{counting-cq-approx} imply the following.
%
% For a class $\class{Q}$ of OMQs from $(\class{G},\class{UCQ})$, define
% a class of CQs
% $\class{Q}^\circ = \{ \mn{core}(q )\mid \exists Q \in \class{Q}:
% (\Omc^\exists,\Sbf,q) \in Q^\exists\}$.
An OMQ $Q \in (\class{G},\class{UCQ})$ has \emph{semantic treewidth at
  most} $k \geq 1$ if there is an OMQ $Q' \in (\class{G},\class{UCQ})$
such that $Q \equiv Q'$ and $Q'$ has treewidth at most $k$.
\begin{thm}
  \label{thm:approx}
  Let $\class{Q} \subseteq (\class{G},\class{UCQ})$. If $\class{Q}$
  has bounded semantic treewidth, then there is an FPTRAS for
  $\mn{AnswerCount}(\class{Q})$.
\end{thm}
%
% We sketch the proof of Theorem~\ref{thm:approx}.
\begin{proof}
  Let $\class{Q} \subseteq (\class{G},\class{UCQ})$ and let $k \geq 1$
  be an upper bound on the semantic treewidth of OMQs from
  $\class{Q}$.  The FPTRAS for $\class{Q}$ works as follows.  Assume
  that an OMQ $Q(\bar x) =(\Omc,\Sbf,q) \in \class{Q}$, an
  \Sbf-database $D$, and an $\epsilon \in (0,1)$ are given as
  input. We enumerate $(\class{G},\class{UCQ})$ until we find an OMQ
  $Q'(\bar x)=(\Omc',\Sbf,q')$ such that $Q' \equiv Q$ and $Q'$ is of
  treewidth at most~$k$.  By Theorem~\ref{thm:finmod}, we can compute
  in time $f(||Q'||) \cdot p(||D||)$ a database $D^* \supseteq D$ such
  that $Q'(D)=q'(D^*) \cap \mn{adom}(D)^{|\bar{x}|}$. To get rid of the
  intersection with $\mn{adom}(D)$, let
  $\widehat Q'(\bar x)=(\Omc',\widehat\Sbf,\widehat q')$ be obtained
  from $Q'(\bar x)$ by setting $\widehat \Sbf=\Sbf \cup \{A\}$ where
  $A$ is a fresh unary relation and constructing $\widehat q'$ from
  $q'$ by adding $A(x)$ for every variable $x$ in $\bar x$. Note that
  the treewidth of $\widehat Q'$ is still at most $k$. Let
  $\widehat D^*$ be obtained from $D^*$ by adding $A(c)$ for all
  $c \in \mn{adom}(D)$, that is, $A$ marks the constants that are
  already in $D$, but not those that have been freshly introduced when
  constructing $D^*$. It is clear that
  $Q'(D)=\widehat q'(\widehat D^*)$.  Now, let $\class{UCQ}_k$ be the
  class of UCQs that have treewidth at most $k$.
  By Propoposition~3.5 of \cite{counting-cq-approx}, there is an FPRAS
  for $\mn{AnswerCount}(\class{UCQ}_k)$, where an FPRAS is defined
  like an FPTRAS except that the running time may be at most
  $p(|w|,\frac{1}{\epsilon})$.  We use this FPRAS to compute an
  approximation of $\#\widehat q'(\widehat D^*)$ and return the
  result. Overall, this yields the desired FPTRAS for~$\class{Q}$.
\end{proof}
It is interesting to note the contrast between
Theorem~\ref{thm:approx} and Point~1 of
Theorem~\ref{thm:main-results-OMQs}: the latter refers to the
treewidth and contract treewidth of the class of CQs $\class{Q}^\ast$,
which is defined in a non-trivial way, while Theorem~\ref{thm:approx}
simply speaks about the semantic treewidth of the OMQs in $\class{Q}$
and is thus in line with the characterizations of efficient OMQ
\emph{evaluation} given in \cite{BFLP19,BDFLP-PODS20}. In fact, the
classes of OMQs covered by Theorem~\ref{thm:approx} are precisely
those subclasses of $(\class{G},\class{UCQ})$ for which evaluation is
in FPT \cite{BDFLP-PODS20}, paralleling the situation for CQs without
ontologies. Informally, exact counting and approximate counting
differ in how the CQs inside a UCQ interact (and we cannot avoid
UCQs when we eliminate existential quantifiers from ontologies).
In exact counting, the Chen-Mengel closure captures this interaction,
demonstrating that answer counting for a UCQ may enable answer
counting for CQs whose structural measures are higher than that
of any CQ in the UCQ. In approximate counting, such effects do
not seem to play a role. Also note that Theorem~\ref{thm:approx}
does not rely on the data schema to be full, unlike the upper bounds
in Theorem~\ref{thm:main-results-OMQs}.

It may well be the case that a matching lower bound can be proved for
Theorem~\ref{thm:approx} under the assumptions that
$\wOne \neq \text{FPT}$ and $\text{B}=\text{BPP}$, that is, if
$\class{Q} \subseteq (\class{G},\class{UCQ})$ does not have bounded
semantic treewidth, then there is no FPTRAS for
$\mn{AnswerCount}(\class{Q})$ unless one of the mentioned assumptions
fails. This was proved in \cite{counting-cq-approx} for classes
$\class{Q}$ of CQs (without ontologies) under the additional
assumption that for every $q \in \class{Q}$, there is a self-join free
$q' \in \class{Q}$ that has the same hypergraph as $q$. It is
currently only known that this assumption can be dropped when all OMQs
in $\class{Q}$ are Boolean and when none of the OMQs in $\class{Q}$
contains quantified variables.  We conjecture that it is possible to
lift these restricted cases from pure CQs to OMQs from
$(\class{G},\class{UCQ})$, building on results from
\cite{BDFLP-PODS20}. The general case, however, remains open.

\section{The Meta Problems---Equivalent Queries with Small Measures}
\label{sect:approx}

Theorems~\ref{thm:main-results-OMQs} and~\ref{thm:closedmain} show
that low values for the structural measures of treewidth, contract
treewidth, starsize, and linked matching number are central to
efficient answer counting. This suggests the importance of the meta
problem to decide whether a given query is equivalent to one in
which some selected structural measures are small, and to construct
the latter query if it exists. We present some
%first
 results on this topic both for ontology-mediated
 querying and for querying under constraints. These results and their
 proofs also
%
% consider several kinds
% of approximations of OMQs, as well as of UCQs in the presence of
% integrity constraints. The approximations are from below, that is,
% approximating queries are contained in the approximated query, but not
% necessarily equivalent to it. In addition, they have low values for
% selected structural measures. The
shed
some more light on the interplay between the ontology and the structural
measures.

% and give rise to decidability results for certain meta
% problems.
%
% the complexity of answer counting in the
% presence of guarded TGDs, both in the open and in the closed world
% case. As further Example~\ref{ex:lowering-measures} shows that the ontology/set
% of integrity constraints interacts with these measures. The aim of
% this section is to analyze this further. We define the treewidth of an
% OMQ $(\Omc,\Sbf,q)$ to be that of $q$, and likewise for the other
% structural measures. We are then interested in the following
% questions. How to decide whether a given OMQ $Q'$ is equivalent to an
% OMQ $Q'$ in which a selected structural measure takes a small value,
% and how to compute such a $Q'$? Is it ever necessary to modify the
% ontology when constructing $Q'$ from $Q$? And in the context of a CQS
% $(\Tmc,\Sbf,q)$, how can we decide whether there is an OMQ $q'$ that
% is equivalent to $q$ on databases that satisfy \Tmc and how to compute
% it? Note that it does not make sense to change the set of constraints
% \Tmc used in the CQS as this changes the promise.  Questions of this
% kind have also been studied, for example, in
% \cite{PabloEarly,LICS,PODS,GeorgEtAl}, see also the references
% therein. We remark that the notion of treewidth used in
% \cite{LICS,PODS} is not the same one as in this paper as it disregards
% the answer variables.

\subsection{Querying Under Constraints}
\label{meta:cqs}

We start with querying under constraints, considering all measures in
parallel. In fact, we even consider sets of measures since some of the
statements in Theorems~\ref{thm:main-results-OMQs}
and~\ref{thm:closedmain} refer to multiple measures and it is not a
priori clear whether the fact that each measure from a certain set of
measures can be made small in an equivalent query implies that the
same is true for all measures from the set simultaneously.

Our approach is as follows.  For a~given CQS $(\Tmc,\Sbf,q)$, we
construct a certain CQ $q'$ that approximates $q$ from below under~the
constraints in \Tmc and that has small measures. Similar
approximations have been considered for instance in
\cite{DBLP:journals/siamcomp/BarceloL014}, without constraints. We
then show that if there is any CQ $q''$ that has small measures and is
equivalent to $q$ under the constraints in~\Tmc, then $q'$ is
equivalent to~$q$.  In this way, we are able to simultaneously solve
the decision and computation version of the meta problem at hand. With
`approximation from below', we mean that the answers to $q'$ are
contained in those to $q$ on all \Sbf-databases. This should not be
confused with computing an approximation of the \emph{number} of
answers to a given query as considered in Section~\ref{sect:approx-fpras}.
%
%, which is a very interesting but entirely
%different topic.

A \emph{set of measures} is a subset $M \subseteq \{ \text{TW},
\text{CTW}, \text{SS}, \text{LMN} \}$ with the obvious meaning. For
a set of measures $M$ and $k \geq 1$, we say that a UCQ $q$ is an
\emph{$M_k$-query} if for every CQ in~$q$, every measure from $M$ is
at most $k$. If \Tmc is a finite set of TGDs from
$(\class{G},\class{UCQ})$ over schema \Sbf and $q_1(\bar x), q_2(\bar
x)$ are UCQs over \Sbf, then we say that $q_1$ is \emph{contained in
  $q_2$ under~\Tmc}, written $q_1 \subseteq_\Tmc q_2$, if $q_1(D)
\subseteq q_2(D)$ for every \Sbf-database $D$ that is a model of \Tmc,
and likewise for equivalence and $q_1 \equiv_\Tmc q_2$.
%
% abbreviate treewidth with TW, contract treewidth with
% CTW, starsize with SS, and linked matching number with LMN.
%
%\cMP{Can we use M(q), SS(q), or TW(q) for the respective measure's value?}
%
\begin{defi}
\label{def:closedapprox}
  Let $(\Tmc,\Sbf,q) \in (\class{G},\class{UCQ})$ be a CQS, $M$ a~set
  of measures, and
%  \subseteq \{ \text{TW}, \text{CTW}, \text{SS}, \text{LMN} \}$, and
  $k \geq 1$. An \emph{$M_k$-approximation of $q$ under \Tmc} is a UCQ
  $q'$ such that
  \begin{enumerate}
  \item $q' \subseteq_\Tmc q$,
  \item $q'$ is an $M_k$-query, and
  \item for each UCQ $q''$
    that satisfies Conditions~1 and~2, $q'' \subseteq_{\Tmc} q'$.
  \end{enumerate}
\end{defi}
It might be useful for the reader to reconsider
Example~\ref{ex:lowering-measures}, which for every $n \geq 0$ gives
an OMQ $(\Omc,\Sbf,q_n)$ with full schema such that has high measures,
but is equivalent to an OMQ $(\Omc,\Sbf,p_n)$ with low measures. The
equivalence also holds true if the OMQs are viewed as CQSs, that is,
$q_n \equiv_\Omc p_n$. If we choose for example $M=\{TW,CTW\}$ and
$k=1$, then it can be seen that every $M_k$-approximation of $q_n$
must contain a CQ that is equivalent to $p_n$.

\smallskip

We next identify a simple way to construct $M_k$-approximations.  Let
$(\Tmc,\Sbf,q) \in (\class{G},\class{UCQ})$ be a CQS, $M$ a set of
measures, and $k \geq 1$. Moreover, let $\ell$ be the maximum number
of variables in any CQ in $q$ and fix a set \Vmc of exactly
$\ell \cdot \mn{ar}(\Sbf)$ variables. Assuming that \Tmc is understood
from the context, we define $q^M_k$ to be the UCQ that contains as a
disjunct any CQ $p$ such that $p \subseteq_\Tmc q$, $p$ is an
$M_k$-query, and $p$ uses only variables from \Vmc. As containment
between UCQs under constraints from $\class{G}$ is decidable
\cite{BaBP18},
%
%can be computed in double exponential time \cite{where?}
%{\color{blue}also speak about complexity of checking measures},
given $(\Tmc,\Sbf,q)$ we can effectively compute $q^M_k$.  We show
next that $q^M_k$ is an $M_k$-approximation of $q$ under~\Tmc.
\begin{lem}
%\begin{restatable}{lem}{firstapprox}
\label{prop:firstapprox}
  Let $(\Tmc,\Sbf,q) \in (\class{G},\class{UCQ})$ be a CQS, $M$
  a set of measures,
%  \subseteq \{ \text{TW}, \text{CTW}, \text{SS}, \text{LMN} \}$,
and
  $k \geq 1$. Then $q^M_k$ is an $M_k$-approximation of $q$ under
  \Tmc.
%\end{restatable}
\end{lem}
\begin{proof}
  % We only give the most important parts of the proof and defer some
  % additional details in the appendix.
  By construction, $q_k^M$
  satisfies Points~1 and~2 from Definition~\ref{def:closedapprox}. We
  show that it satisfies also Point~3.
	%
  % We start with some preparations. For an $\Sbf$\-/database $D$, a
  % finite set of TGDs $\Tmc$ from $\class{G}$, and a constant $a \in
  % \dom(\chase_{\Tmc}(D)) \setminus \dom(D)$, a guarded set $X$ over
  % $D$ is a \emph{generator} of $a$ in~$\Tmc$ if $a \in
  % \dom(\chase_{\Tmc}(\chase_{\Tmc}(D)|_{X}))$. Informally, the shape
  % of $\chase_{\Tmc}(D)$ is that of $D$ with a tree-like
  % structure\footnote{More precisely, a structure of treewidth $\ell$,
  %   where $\ell$ is the maximum number of variables in the head of a
  %   TGD from \Tmc.} attached to every guarded set $X$ in $D$, and $X$
  % being a generator of $a$ means that $a$ is located in the tree
  % attached to $X$.  There is at least one generator for every $a \in
  % \dom(\chase_{\Tmc}(D)) \setminus \dom(D)$ and furthermore,
  % generators are \emph{complete} in the sense that if $X$ is a
  % generator for $a$ and $R(\bar a) \in \chase_{\Tmc}(D)$ with $a \in
  % \bar a$, then $R(\bar a) \in
  % \chase_{\Tmc}(\chase_{\Tmc}(D)|_{X})$~\cite{CaGK13}. We also remark
  % that generators need not be unique.
%
Let $q''(\bar x)$ be a UCQ such that $q'' \subseteq_{\Tmc} q$ and
$q''$ is an $M_k$-query. Further, let $p$ be a CQ in~$q''$.
We have to show that $q^M_k$ contains a
CQ $p'$ with $p \subseteq_\Tmc p'$.

We apply Theorem~\ref{thm:finmod} to the ontology $\Omc=\Tmc$, the
database $D=D_p$, and the integer~$n$, defined to be the maximum
number of variables of CQs in~$q$. This yields a database $D^\star_{p}$
which has the properties that $D^\star_{p} \models \Tmc$, $D_{p} \subseteq
D^\star_{p}$, and thus $\bar x \in p(D_{p}^\star)$. From $q'' \subseteq_{\Tmc}
q$, it follows that $\bar x \in q(D_{p}^\star)$, and thus there must be CQ
$q_i$ in $q$ such that $\bar x \in q_i(D_{p}^\star)$. From Point~2 of
Theorem~\ref{thm:finmod} and $|\mn{var}(q_i)| \leq \ell$, it follows
that $\bar x \in Q(D_p)$ for the OMQ $Q=(\Tmc, \Sbf,
q_i)$. Consequently, $q_i$ maps into $\chase_{\Tmc}(D_p)$ via some
homomorphism $h$ that is the identity on $\bar x$. We intend to use
$h$ for identifying the desired CQ $p'$ in $q_k^M$ such that $p \subseteq_{\Tmc}
p'$. We need some preliminaries that we keep on an intuitive level
here and flesh out in the appendix.

Since \Tmc is a set of guarded TGDs, $\chase_{\Tmc}(D_p)$ is of a
certain regular shape. Informally, it looks like $D_p$ with a
tree-like structure attached to every guarded set $X$ in
$D_p$.\footnote{More precisely, a structure of treewidth $\ell$, where
  $\ell$ is the maximum number of variables in the head of a TGD from
  \Tmc.} Note that the constants in $D_p$ are exactly the variables in
$p$. We refer to all other constants in $\mn{ch}_\Tmc(D_p)$ as
\emph{nulls}.  Formally we first identify with every fact
$R(\bar c) \in \mn{ch}_\Tmc(D_p)$ such that $\bar c$ contains at least
one null a unique `source' fact $\mn{src}(R(\bar b)) \in D_p$ that
played the role of the guard when the tree-like structure that
$R(\bar c)$ is in was generated by the chase and then use $\mn{src}$
to identify the tree-like structures in $\chase_{\Tmc}(D_p)$.

Start with setting $\mn{src}(R(\bar c))=R(\bar c)$ for all
$R(\bar c) \in D_p$. Next assume that $R(\bar c) \in \mn{ch}_\Tmc(D_p)$
was introduced by a chase step that applies a TGD $T \in \Tmc$ at a
tuple $(\bar d, \bar d')$, and let $R'$ be the relation symbol in the
guard atom in $\mn{body}(T)$. Then we set
$\mn{src}(R(\bar c))=R'(\bar d, \bar d')$ if
$\bar d \cup \bar d' \subseteq \mn{adom}(D_p)$ and
$\mn{src}(R(\bar c))=\mn{src}(R'(\bar d, \bar d'))$ otherwise. For any
guarded set $X$ of $D_p$, define $\mn{ch}_\Tmc(D_p)|^\downarrow_{X}$ to
contain those facts $R(\bar c) \in \mn{ch}_\Tmc(D_p)$ such that the
constants in $\mn{src}(R(\bar c))$ are exactly those in $X$.

%
%This tree-like structure is then
% the one attached to the guarded set $X$ that contains exactly
% the constants in $\bar b$.
% Recall
% that $\mn{ch}_S(I)$ is constructed by a sequence of chase steps.
% For all facts \mbox{$R(\bar c) \in I$}, $\mn{guard}(R(\bar c)) = R(\bar
% c)$.
% The exact definition of \mn{src} is given in the appendix.  Here,
% we use \mn{src} to define what we mean with the tree-like structure
% below a guarded set $X$ of $D_p$: the subinstance
% $\mn{ch}_\Tmc(D_p)|^\downarrow_{X}$ of $\mn{ch}_\Tmc(D_p)$ contains
% those facts $R(\bar c) \in \mn{ch}_\Tmc(D_p)$ such that
% %
% \begin{itemize}

% \item $\bar c \subseteq X$ or

% \item $\bar c$ contains at least one null and
%   the constants in $\mn{src}(R(\bar c))$ are exactly those in $X$.

% \end{itemize}
% %
In the appendix, we show the following:
\begin{enumerate}

\item[(A)] for every guarded set $X$ in $D_p$, there is a homomorphism from
  $\mn{ch}_\Tmc(D_p)|^\downarrow_{X}$ to\linebreak[4]
  $\mn{ch}_\Tmc(\mn{ch}_\Tmc(D_p)|_{X})$ that is the identity on all
  constants in $X$;

\item[(B)] if $c \in \mn{adom}(\mn{ch}_\Tmc(D_p))$ is a null and
  $R_1(\bar c_1),R_2(\bar c_2) \in \mn{ch}_\Tmc(D_p)$ such that $c$
  occurs in both $c_1$ and $c_2$, then $\mn{src}(R_1(\bar
  c_1))=\mn{src}(R_2(\bar c_2))$.

  % Let $I$ be an instance, \Smc a set of TGDs, $I_0,I_1,\dots$ a chase
  % sequence of $I$ with \Smc, and $S$ a guarded set in $I$. Then for
  % every $i \geq 0$ there is a homomorphism from $I_i|^\downarrow_{S}$
  % to $\mn{ch}_S(I_i|_{S})$ that is the identity on all constants in
  % $S$. Moreover, $h_i(c)=h_{i+1}(c)$ for all $c$ with $h_i(c)$
  % defined.
\end{enumerate}
Informally, Condition~(A) may be viewed as a locality property of the
chase an Condition~(B) says that, as expected, attached tree-like
structures do not share any variables.

As announced, we now construct the CQ $p'$ in $q^M_k$.  All atoms in
$p'$ are facts from $\chase_{\Tmc}(p)|_{\mn{var}(p)}$, viewed as
atoms. To control the number of variables in $p'$, however, we do not
include all such atoms, but only a selection of them.  Consider each
atom $R(\bar y)$ in $q_i$ and distinguish the following cases:
\begin{itemize}
	\item if $h(\bar y)$ contains only variables, then add $R(h(\bar y))$ to $p'$;
	\item if $h(\bar y)$ contains a null, then consider
          the atom $S(\bar z)=\mn{src}(R(h(\bar y))) \in D_p$ and
          let $X$ be the set of variables in $\bar z$; add all facts in
          $\chase_{\Tmc}(p)|_{X}$ as atoms to $p'$.

\end{itemize}
%
 % Note that the atoms added in Point~2 also contain only variables, but
 % no nulls.
The answer variables of $p'$ are exactly those of $p$. All these
variables must be present since $h$ is the identity on $\bar x$. It
follows from the construction of $p'$ that the identity is a
homomorphism from $p'$ to $\mn{ch}_\Tmc(p)$. Thus
$p \subseteq_{\Tmc} p'$ and it remains to show that, up to renaming
the variables so that they are from the set \Vmc fixed for the
construction of $q_k^M$, $p'$ is a CQ in $q_k^M$. This is a
consequence of the following properties:

\smallskip
\noindent
 (1) $p'$ is an $M_k$-query.

 \smallskip
         \noindent
          By definition of $p'$, all guarded sets in $p'$ are also
          guarded sets in $\mn{ch}_\Tmc(D_p)$. Moreover, those guarded
          sets contain no nulls and are thus also guarded sets in $p$.
          Consequently, the Gaifman graph of $p'$ is a subgraph of the
          Gaifman graph of $p$, and all measures are
          monotone regarding subgraphs.
	% \item $p \subseteq_{\Tmc} p'$. Follows from the fact that by
        %   construction of $p'$, the
        %   identity is a homomorphism from $p'$ to $\mn{ch}_\Tmc(p)$.
%
          % For every $\Sbf$-database $D$ such that $D \models \Tmc$ and $\bar x \in p(D)$, $\chase_{\Tmc}(p)$ maps into $D$ via some homomorphism that is the identity on $\bar x$. As $p' \subseteq \chase_{\Tmc}(p)$, it follows that $\bar x \in p'(D)$. Thus, $p \subseteq_{\Tmc} p'$.

\smallskip
\noindent
(2)~$p' \subseteq_{\Tmc} q_i$. % We first observe that w.l.o.g. one can assume that $\chase_{\Tmc}(p') \subseteq \chase_{\Tmc}(p)$.

          \noindent
          It suffices to construct a homomorphism $h'$ from $q_i$ to
          $\chase_{\Tmc}(D_{p'})$. This can be done as follows. For
          each $x \in \mn{var}(q_i)$ with $h(x)$ a variable, put
          $h'(x)=h(x)$. It remains to deal with all $x \in
          \mn{var}(q_i)$ with $h(x)$ a null.

          With any such $x$, we associate a unique atom
          $\Gamma(x)=\mn{src}(R_1(h(\bar x_1))) \in D_p$,
          identifying the tree-like structure in
          $\chase_{\Tmc}(D_{p})$ that $h(x)$ is in. Take any atom
          $R(\bar y) \in q_i$ such that $\bar y$ contains $x$. It
          follows from (B) above that $\mn{src}(R(h(\bar y)))$ is
          the same, no matter which such atom $R(\bar y) \in q_i$ we
          take. We may thus associate with $x$ the unique atom
          $\Gamma(x)=\mn{src}(R(h(\bar y))) \in D_p$.

          Now consider any maximal set $X$ of variables $x
          \in\mn{var}(q_i)$ such that $h(x)$ is a null and
          $x_1,x_2 \in X$ implies $\Gamma(x_1)=\Gamma(x_2)$. Then $h$
          is a homomorphism from $q_1|_X$ to
          $\mn{ch}_\Tmc(D_p)|^\downarrow_X$. By (A) above, there is a
          homomorphism $h_X$ from $\mn{ch}_\Tmc(D_p)|^\downarrow_X$
          to $\mn{ch}_\Tmc(\mn{ch}_\Tmc(D_p)|_{X})$ that is the
          identity on all constants in $X$. Moreover, the construction
          of $p'$ yields $\mn{ch}_\Tmc(D_p)|_{X} \subseteq p'$. We may
          thus view $h_X$ as a homomorphism from
          $\mn{ch}_\Tmc(D_p)|^\downarrow_X$ to $\mn{ch}_\Tmc(p')$.
          Set $h'(x)=h_X \circ h(x)$ for all $x \in X$.

          It can be verified that the constructed $h'$ is indeed a
          homomorphism from $q_i$ to $\chase_{\Tmc}(D_{p'})$.

          \smallskip
	\noindent (3)~$|\dom(p')| \leq \ell \cdot  \mn{ar}(\Sbf) $.

          \noindent
          A straightforward analysis of the construction of $p'$ shows
          that it introduces into $p'$ the following variables,
          implying the statement:
          \begin{itemize}

          \item for every $x \in q_i$ with $h(x)$ a variable, the
            variable $h(x)$,

          \item for every $x \in q_i$ with $h(x)$ a null,
            the variables that occur in $\Gamma(x)$, where $\Gamma$
            is as above.

          \end{itemize}
          In fact, assume that an atom $R(\bar y)$ is treated in the
          construction of $p'$ and let $\bar y = y_1,\dots, y_k$. If
          Case~1 of the construction applies, then the variables
          $h(y_1),\dots,h(y_k)$ are introduced. For Case~2 of the
          construction, we reuse the function $\Gamma$ defined in the
          proof of the previous property. If this case applies,
          then by definition % there is a variable $x$ in
          % $\bar y$ with $h(x)$ a null,
          $\Gamma(x)$ is identical for every variable $x$ in $\bar y$
          with $h(x)$ a null, subsequently just referred to as
          $\Gamma$. This $\Gamma$ contains $h(x)$ for all variables
          $x$ in $\bar y$ with $h(x)$ a variable, and $\Gamma$ is
          precisely the set of variables introduced in this step.
\end{proof}
%
% While $q^M_k$ is a UCQ that approximates a UCQ, we can also
% approximate CQs in terms of CQs, essentially by using a suitable
% disjunct of $q^M_k$. This is an interesting case as it admits an
% approximate version of answer counting in FPT, see
% Corollary~\ref{cor:approxclosed} below.
% %
% \begin{thm}
% \label{thm:CQapprox}
%   Let $(\Tmc,\Sbf,q) \in (\class{G},\class{CQ})$ be a CQS, $M$
% a set of measures, and
% %  \subseteq \{ \text{TW}, \text{CTW}, \text{SS}, \text{LMN} \}$, and
%   $k \geq 1$. Then there is a CQ that is an $M_k$-approximation of
%   $q$ under \Tmc, has at most $\ell \cdot \mn{ar}(\Sbf)$ variables, and
%   can be effectively computed.
% \end{thm}
%
Let $(\Tmc,\Sbf,q) \in (\class{G},\class{UCQ})$ be a CQS.  By
definition of $M_k$-approximations, it is clear that if there exists a
UCQ $q'$ such that $q' \equiv_\Tmc q$ and $q'$ is an $M_k$-query, then
any $M_k$-approximation $q^\star$ of $q$ under \Tmc also satisfies
$q^\star \equiv_\Tmc q$.  The following is thus an immediate
consequence of Lemma~\ref{prop:firstapprox} and the fact that
containment between UCQs under constraints from $\class{G}$ is
decidable.
\begin{thm}
  \label{cor:approxclosed}
  Let $M$ be a set of measures.  Given a CQS $(\Tmc,\Sbf,q) \in
  (\class{G},\class{UCQ})$ and $k \geq
  1$, it is decidable whether
  $q$ is equivalent under \Tmc to a UCQ $q'$ that is an $M_k$-query.
  Moreover, if this is the case, then such a $q'$ can be effectively
  computed.
% Let $M=\{ \text{TW},\text{CTW}\}$ and $k \geq 1$.  Given a CQS
%   $(\Tmc,\Sbf,q) \in (\class{G},\class{CQ})$ and an \Sbf-database $D$,
% $\#q'(D)$ can be computed in FPT where $q'$  is an $M_k$-approximation
%   of $q$ under \Tmc.
\end{thm}
%
% It is not clear whether Corollary~\ref{cor:approxclosed} extends to
% UCQs, essentially because for a UCQ $q$, answer counting in FPT is
% tied to the measures of the CQs in $\mn{cl}_{\mn{CM}}(q)$ rather than
% those in $q$.
%
A particularly relevant case is
$M=\{\text{TW},\text{CTW}\}$, as it is linked to fixed-parameter
tractability. From the above results, we obtain that answer counting
in FPT is possible for CQSs that are \emph{semantically} of bounded treewidth
and contract treewidth, provided that only CQs are admitted as the actual
query. Let us make this more precise. Fix $k \geq
1$ and let $\class{C}_k$ be the class of CQSs $(\Tmc,\Sbf,q) \in
(\class{G},\class{CQ})$ such that $q \equiv_\Tmc q'$ for some UCQ
$q'$ of treewidth and contract treewidth at most
$k$. %  Note that we admit only CQs, but not UCQs as the actual
% query.
Then
$\mn{AnswerCount}(\class{C}_k)$ is in FPT: given a CQS $(\Tmc,\Sbf,q)
\in \class{C}_k$ and an \Sbf-database
$D$, we may compute, as per Theorem~\ref{cor:approxclosed}, a
UCQ~$q'$ that is an $M_k$-query and satisfies $q \equiv_\Tmc
q'$. Since
$q$ is a CQ, it is easy to see that there must be a single disjunct
$q^\star$ of $q$ such that $q \equiv_\Tmc
q^\star$.  We can effectively identify
$q^\star$ and use Point~1 of Theorem~\ref{thm:chenmengeldell} as a
blackbox to count answers to $q^\star$ on
$D$. The same is probably not true when we define
$\class{C}_k$ as a subclass of
$(\class{G},\class{UCQ})$ rather than
$(\class{G},\class{CQ})$. Then, we have to count the answers to
$q'$ on $D$ rather than to a single CQ $q^\star$ in
$q'$, but for UCQs of bounded treewidth and contract treewidth,
Point~1 of Theorem~\ref{thm:chenmengeldell} does not always guarantee
answer counting in FPT because of the use of the Chen-Mengel closure
in that theorem.

% the original query is a UCQ and, related to this, whether
% $M_k$-approximations always admit answer counting in FPT, that is,
% even when the original query is not equivalent under \Tmc to an
% $M_k$-query.

  % As another relevant application of
  % Proposition~\ref{prop:firstapprox} and Theorem~\ref{thm:CQapprox},
  % we note the following.  Let $M$ be a set of measures and $k \geq 1$.
  % Given a CQS $(\Tmc,\Sbf,q) \in (\class{G},\class{UCQ})$, it is
  % decidable whether $q$ is equivalent under \Tmc to a UCQ that is an
  % $M_k$-query, and the same holds true for CQs. This is in fact an
  % easy consequence of the mentioned results, decidability of UCQ
  % containment under constraints from $\class{G}$, and the fact that if
  % $q$ is equivalent under \Tmc to any UCQ $q'$ that is an $M_k$-query,
  % then it is equivalent to its $M_k$-approximation.

\subsection{Ontology-Mediated Queries}

We now turn to ontology-mediated queries, starting with the definition
of their approximations. We say that OMQ $Q_1(\bar x)=(\Omc_1,\Sbf,q_1)$ is
\emph{contained} in OMQ $Q_2(\bar x)=(\Omc_2,\Sbf,q_2)$, written
$Q_1 \subseteq Q_2$, if $Q_1(D) \subseteq Q_2(D)$ for every
\Sbf-database $D$. $Q_1$ and $Q_2$ are \emph{equivalent}, written
$Q_1 \equiv Q_2$, if $Q_1 \subseteq Q_2$ and $Q_2 \subseteq Q_1$. We
say that an OMQ $Q=(\Omc,\Sbf,q)$ is an \emph{$M_k$-query} if $q$ is.
\begin{defi}
  \label{def:M-approximation}
  Let $Q=(\Omc,\Sbf,q) \in (\class{G},\class{UCQ})$ be an OMQ, $M$
  a set of measures, and
%\subseteq \{
 % \text{TW}, \text{CTW}, \text{SS}, \text{LMN} \}$, and
$k \geq 1$.
An
  \emph{$M_k$-approximation of} $Q$ is an OMQ $Q'=(\Omc',\Sbf,q') \in
  (\class{G},\class{UCQ})$ such that
  \begin{enumerate}
  \item $Q' \subseteq Q$,
  \item $Q'$ is an $M_k$-query, and
  \item for each $Q''=(\Omc'',\Sbf,q'') \in (\class{G},\class{UCQ})$
    that satisfies Conditions~1 and~2, $Q'' \subseteq Q'$.
  \end{enumerate}
  We say that $Q'$ is an \emph{$M_k$-approximation of $Q$ while
    preserving the ontology} if it is an $M_k$-approximation and
  $\Omc'=\Omc$.
\end{defi}
We next observe that $M_k$-approximations of OMQs $(\Omc,\Sbf,q) \in
(\class{G},\class{UCQ})$ based on the full schema and while preserving
the ontology are closely related to the approximations studied in the
previous section in the context of CQSs.
\begin{lem}
    \label{lem:relateapprox}
    Let $Q(\bar x) = (\Omc,\Sbf,q) \in (\class{G},\class{UCQ})$ be an OMQ based on
    the full schema, $M$ a set of measures, %  \subseteq \{ \text{TW}, \text{CTW}, \text{SS},
    % \text{LMN} \}$,
    and $k \geq 1$. Then an OMQ $Q'(\bar x)=(\Omc,\Sbf,q') \in
    (\class{G},\class{UCQ})$ is an $M_k$-approximation of $Q$ while
    preserving the ontology iff $q'$ is an $M_k$-approximation of $q$
    under~\Omc.
\end{lem}
\begin{proof}
    % To prove the lemma, recall that if an~$\Sbf$\=/database satisfies $D \models \Omc$, then for every OMQ $Q = (\Omc,\Sbf,q) \in (\class{G},\class{UCQ})$
    % we have that $q(D) = Q(D)$.
    %
  ``if''. Assume that $q'$ is an $M_k$\=/approximation of $q$ under
  \Omc. To show that $Q'(\bar x)=(\Omc,\Sbf,q)$ is an
  $M_k$-approximation of $Q$, we have to show that Points~(1) to~(3) from
  Definition~\ref{def:M-approximation} hold. Point~(2) is obvious.

  \smallskip
  \noindent
  (1)~$Q' \subseteq Q$.

  \smallskip
      \noindent
      Let
    $D$ be an \Sbf-database. Further, let
    $D^\star$ be the database from Theorem~\ref{thm:finmod} invoked with
    \Omc, $D$, and $n=\max\{|\mn{var}(q)|,|\mn{var}(q')|\}$.
    Then
    $$
    Q'(D) = q'(D^\star) \cap \dom(D)^{|\bar x|}\subseteq
    q(D^\star) \cap \dom(D)^{|\bar x|} =
    Q(D).
    $$
    The containment in the center holds since $q'$ is an $M_k$\=/approximation
    of $q$ under \Omc.

  \smallskip
      \noindent
(3)~$P \subseteq Q'$ for all
    $P(\bar x) = (\Omc, \Sbf, p) \in (\class{G},\class{UCQ})$
    such that $P \subseteq Q$ and $p$ is an $M_k$-query.

    \smallskip
    \noindent
    We first observe that $p \subseteq_\Omc q$. Thus let $ D$ be an
    \Sbf-database that satisfies all TGDs from~\Omc. Then $P(D)=p(D)$
    and $Q(D)=q(D)$, thus $P \subseteq Q$ implies
    $p(D) \subseteq q(D)$.

    We now show that $P \subseteq Q'$, as
    required.  Let $D^\star$ be the database from
    Theorem~\ref{thm:finmod} invoked with \Omc, $D$, and
    $n=\max\{|\mn{var}(q')|,|\mn{var}(p)|\}$.  Then
    $$P(D) = p(D^\star) \cap \dom(D)^{|\bar x|} \subseteq q'(D^\star) \cap \dom(D)^{|\bar x|} = Q'(D).
    $$
    The containment in the center holds since $p \subseteq_\Omc q$
    and $q'$ is an $M_k$\=/approximation of $q$ under~\Omc.

    \smallskip ``only if''.  Assume that $Q'(\bar x)=(\Omc,\Sbf,q')$
    is an~$M_k$\=/approximation of $Q$ while preserving the ontology.
    To show that $q'$ is an $M_k$-approximation of $q$ under \Omc, we
    have to show that Points~(1) to~(3) from
    Definition~\ref{def:closedapprox} are satisfied. Again,
    Point~2 is obvious.

    \smallskip
    \noindent
    (1)~$q' \subseteq_\Omc q$.

    \smallskip
      \noindent
      Follows from the fact that $q'(D)=Q'(D) \subseteq Q(D) =q(D)$
      for all \Sbf-databases $D$ that satisfy~\Omc.  The containment
      holds since $Q'$ is an $M_k$-approximation of $Q$.

    \smallskip
    \noindent
    (2)~$p \subseteq_\Omc q'$ for all UCQs $p$ such that
    $p \subseteq_\Omc q$ and $p$ is an $M_k$-query.

\smallskip    \noindent
    We first observe that $P \subseteq Q'$ where $P=(\Omc,\Sbf,p)$.
    In fact, let $D$ be an \Sbf-database. Now take the database
    $D^\star$ from Theorem~\ref{thm:finmod} invoked with \Omc, $D$,
    and $n=\max\{|\mn{var}(q')|,|\mn{var}(p)|\}$.  Then
    $$P(D) = p(D^\star) \cap \dom(D)^{|\bar x|} \subseteq q'(D^\star) \cap \dom(D)^{|\bar x|} = Q'(D).
    $$
    Now, $p \subseteq_\Omc q'$ is a consequene of $P \subseteq Q'$
    and the fact that $P(D)=p(D)$ and $Q(D)=q(D)$ for all $\Sbf$-databases
    $D$ that satisfy \Omc.
\end{proof}
Lemma~\ref{lem:relateapprox} allows us to compute approximations of
OMQs using the construction given in Section~\ref{meta:cqs}.  As in
the CQS case, it is easy to see that a given OMQ $Q=(\Omc,\Sbf,q)$ is
equivalent to an OMQ $Q'=(\Omc,\Sbf,q')$ that is an $M_k$-query
if and only if the $M_k$-approximation of $Q$ is an $M_k$-query.
We thus obtain the following.
\begin{thm}
%\begin{restatable}{thm}{thmomqfpt}
  \label{thm:omqDecFirst}
  Let $M$ be a set of measures.  Given an OMQ $Q=(\Omc,\Sbf,q) \in
  (\class{G},\class{UCQ})$ based on the full schema and $k \geq
  1$, it is decidable
  whether $Q$ is equivalent to an OMQ $Q'=(\Omc,\Sbf,q') \in
  (\class{G},\class{UCQ})$ that is an
  $M_k$-query.  Moreover, if this is the case, then such a
  $Q'$ can be effectively computed.
%\end{restatable}
\end{thm}
%
% \begin{proof}
%     Let $Q(\bar x)=(\Omc,\Sbf,q) \in (\class{G},\class{CQ})$ based on the full
%     schema and an \Sbf-database $D$ be given. We first construct an
%     $M_k$-approximation $Q^M_k$ of $Q$ as follows: first construct the
%     $M_k$-approximation $q^M_k$ of $q$ under \Omc as per
%     Lemma~\ref{prop:firstapprox}. By Lemma~\ref{lem:relateapprox},
%     the OMQ $Q^M_k=(\Omc,\Sbf,q^M_k)$ is an $M_k$-approximation of $Q$.
%     Next construct from \Omc and \Dmc the database $D^\star$ as per
%     Theorem~\ref{thm:finmod}, setting $n$ to be the number of variables
%     in $q^M_k$. We seek to compute $\#(q^M_k(D^\star)
%     \cap \dom(D)^{\bar x})$.
%
%     To do this, we apply the algorithm asserted by
%     Lemma~\ref{lemma:partial-domain-answers} with $D^\star$ in place of $D$
%     and with $F \eqdef \dom(D)$. The algorithm needs an oracle for
%     $\mn{AnswerCount}(\{q^M_k\}, \mn{clones}(D^\star))$. Since $q^M_k$
%     is an $M_k$-query, $\mn{AnswerCount}(\{q^M_k\})$ (on the class of
%     all databases) is in FPT. When we use an fpt algorithm as the
%     oracle, overall we stay within the time requirements for fixed
%     parameter tractability.
% \end{proof}
%
While Theorem~\ref{thm:omqDecFirst} requires the schema to be full and
the ontology to be preserved, we now turn to approximations of OMQs
that need neither preserve the ontology nor assume the full schema. We
focus on contract treewidth and starsize and leave treewidth and
dominating starsize as open problems. To simplify notation, instead of
$\{ \text{CTW} \}_k$-approximations we speak of
CTW$_k$-approximations, and likewise for SS$_k$-approximations.

% If $M_k$-approximations are effectively computable, then we can decide
% whether a given OMQ is equivalent to one for which all measures in $M$
% are bounded by $k$, see also the comment after
% Definition~\ref{def:closedapprox}.

% We next consider CTW$_k$-approximations and SS$_k$-approximations,
% starting with the former. In both cases, it turns out that modifying
% the ontology is not necessary.

A \emph{collapsing} of a CQ $q(\bar x)$ is a CQ $p(\bar x)$ that can
be obtained from $q$ by identifying variables and adding equality
atoms (on answer variables). When an answer variable $x$ is identified
with a non-answer variable $y$, the resulting variable is $x$; the
identification of two answer variables is not allowed.  The
\emph{CTW$_k$-approx\-imation} of an OMQ
$Q=(\Omc,\Sbf,q) \in (\class{G},\class{UCQ})$, for $k \geq 1$, is the
OMQ $Q_{k}^{\text{CTW}}=(\Omc,\Sbf,q^\text{CTW}_k)$ where
$q^\text{CTW}_k$ is the UCQ that contains as CQs all collapsings of
$q$ that have contract treewidth at most $k$. The
\emph{SS$_k$-approx\-imation} of $Q$ is defined accordingly, and denoted
with $Q^{\text{SS}}_k$.

% {\color{blue}speak about
  % complexity of computing approximation?}
%
\begin{thm}
\label{thm:CTWSSjointapprox}
Let $(\Omc,\Sbf,q) \in (\class{G},\class{UCQ})$ be an OMQ and $k \geq
1$. Then $Q^{\text{CTW}}_k$ is a CTW$_k$-approximation of $Q$.
Moreover, if $k \geq \mn{ar}(\Sbf)$, then $Q^{\text{SS}}_k$ is an
SS$_k$-approximation of~$Q$.
%    and can be computed in double exponential time.
\end{thm}
%
% The \emph{SS$_k$-approx\-imation} of an OMQ $Q=(\Omc,\Sbf,q) \in
% (\class{G},\class{UCQ})$, for $k \geq 1$, is defined in the same way
% as CTW$_k$-approximations, but includes only collapsings that have
% starsize at most~$k$. We denote it with $Q^{\text{SS}}_k$.
% %{\color{blue}speak about complexity of computing approximation?}
%
% \begin{restatable}{thm}{SSapprox}
% \label{thm:SSapprox}
%   Let $(\Omc,\Sbf,q) \in (\class{G},\class{UCQ})$ be an OMQ and $k
%   \geq \mn{ar}(\Sbf)$. Then $Q^{\text{SS}}_k$ is an
%   SS$_k$-approximation of~$Q$.
% %    and can be computed in double exponential time.
% \end{restatable}
%
The proof of Theorem~\ref{thm:CTWSSjointapprox} is non-trivial and
relies on careful manipulations of databases that are tailored towards
the structural measure under consideration. Details are given
below. The theorem gives rise to decidability results that, in contrast to
Theorem~\ref{thm:omqDecFirst}, neither require the ontology to be
preserved nor the schema to be full.
\begin{cor}
  \label{cor:newcor}
  Given an OMQ $Q=(\Omc,\Sbf,q) \in (\class{G},\class{UCQ})$ and
  $k \geq 1$, it is decidable whether $Q$ is equivalent to an OMQ
  $Q' \in (\class{G},\class{UCQ})$ of contract treewidth at most $k$.
  Moreover,
  if this is the case, then such a $Q'$ can be effectively computed.
  The same is true for starsize in place of contract treewidth.
\end{cor}
Note that, although we are concerned here with approximations that are
not required to preserve the ontology,
Theorem~\ref{thm:CTWSSjointapprox} implies that for
CTW$_k$-approximations and SS$_k$-approximations, it is never
necessary to use an ontology different from the one in the original
OMQ. Before proving Theorem~\ref{thm:CTWSSjointapprox}, we observe
that treewidth behaves differently in this respect, and thus a
counterpart of Theorem~\ref{thm:CTWSSjointapprox} for treewitdth
cannot be expected. This is even true when the schema is full.
\begin{exa}
    \label{ex:altering-tw-full-schema}
    For $n \geq 3$, let $Q_n()=(\emptyset,\Sbf_n,q_n \vee p_n)$ where
    %$\Omc=\emptyset$,
    $\Sbf_n = \{ W, R_1, \dots, R_n\}$ with $W$ of arity $n$ and each
    $R_i$ binary and where
    $$
    q_n = \exists x_1 \cdots \exists x_n \, W(x_1,\dots,x_n)
    \ \text{ and } \
    p_n = \exists x_1 \cdots \exists x_n \exists y \,
    R_1(x_1,y),\dots,R_n(x_n,y).
    $$
    Then $Q'_n()=(\Omc_n,\Sbf_n,p_n)$ with
    $\Omc_n=\{W(\bar{x}) \rightarrow p_n(\bar{x})\}$ is a
    TW$_1$-approximation of $Q_n$. In fact, it is equivalent to
    $Q_n$. However, $Q_n$ has no TW$_k$-approximation $Q^\star$ based
    on the same (empty) ontology for any $k < n {\color{\highlightColourTwo} -1}$ since
    $Q_n \not\subseteq Q^\star$ for any
    $Q^\star=(\emptyset,\Sbf_n,q^\star)$ such that
    %$Q^\star \subseteq Q$ and
    $q^\star$ is of treewidth $k<n{\color{\highlightColourTwo} -1}$. In fact, any
    $Q^\star$  of treewidth $k<n{\color{\highlightColourTwo} -1}$ does not return any answers on
    the database $\{ W(a_1,\dots,a_n) \}$.
\end{exa}
One might criticize that in Example~\ref{ex:altering-tw-full-schema},
the arity of relation symbols grows unboundedly. The next example
shows that this is not necessary. It does, however, use a data schema
that is not full.
\begin{exa}
    \label{ex:altering-tw-not-full-schema}
    Let  $\Sbf = \{ W, R\}$ with $W$ of arity $3$ and $R$ of arity 2.
    For $n \geq 0$, let $Q_n()=(\emptyset,\Sbf,q_n)$ where
    \[
    q_n = \exists z_1 \exists z_2 \exists z_3 \exists x_1 \cdots
    \exists x_{n}  \, \bigwedge_{1 \leq i,j < n; i \neq
        j}R(x_i,x_j) \land
    \bigwedge_{1 \leq i,j < n; j \in \{1,2,3\}}R(x_i,z_j)
    \land W(z_1,z_2,z_3).
    \]
    Then, $G_{q_n}$ is the $(n+3)$\=/clique and thus the treewidth
    of $q_n$ is $n+2$. Since $q_n$ is a core, there is no OMQ
    based on the empty ontology that is equivalent to $Q_n()$ and
    in which the actual query has treewidth less than $n+2$.

    For $n \geq 0$, let $P_n()=(\Omc,\Sbf, p_n)$ where
    \[
    p_n = \exists y \exists z_1 \exists z_2 \exists z_3 \exists
    x_1 \cdots \exists x_{n}  \, \bigwedge_{1 \leq i,j < n; i \neq
        j}R(x_i,x_j) \land
    \bigwedge_{1 \leq i,j < n; j \in \{1,2,3\}}R(x_i,z_j) \land
    \bigwedge_{1 \leq i \leq 3} S(z_i,y)
    \]
    and
    $$\Omc=W(z_1,z_2,z_3) \rightarrow \exists y \bigwedge_{1 \leq
        i \leq 3} S(z_i,y).
    $$
    Then $p_n$ has treewidth $n{+}1$ and $P_n$ is equivalent to
    $Q_n$. Consequently, $P_n$ is a TW$_{n+1}$-approximation of $Q_n$.
    For the case $n=2$, the involved CQs are displayed in
    Figure~\ref{fig:lastexample}.
\end{exa}
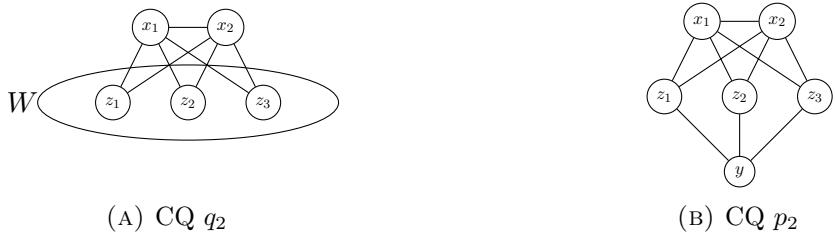
\begin{figure}[t]
    \tikzstyle{quantified}=[circle, draw, scale=.6]
    \tikzstyle{answer} = [circle, fill, scale=.6]
    \centering
    \begin{subfigure}[b]{0.49\textwidth}
        \centering
        %        \resizebox{\linewidth}{!}{
        \begin{tikzpicture}%[scale=1.]
            
            \node[quantified] (z1) at (-1,0) {$z_1$};
            \node[quantified] (z2) at (0,0) {$z_2$};
            \node[quantified] (z3) at (1,0) {$z_3$};
            
            \node[quantified] (x1) at (-.5,1) {$x_1$};
            \node[quantified] (x2) at (.5,1) {$x_2$};
            
            \node (y) at (0,-1)  {};
            
            \draw
            
            (x1) -- (x2)
            
            (x1) -- (z1)
            (x1) -- (z2)
            (x1) -- (z3)
            
            (x2) -- (z1)
            (x2) -- (z2)
            (x2) -- (z3);
            
            \draw (0,0) ellipse (2cm and .5cm);
            \node (x) at (-2.2,0) {$W$};
            
        \end{tikzpicture}
        \caption{CQ $q_2$}
        %    }
    \end{subfigure}
    %%%%%%%%%%%%%%%%%%%%%%%%%%%%%%%%%%%%%%%%%%%%%%%%%%%%%%%%%%%%%%%%%%%%
    \begin{subfigure}[b]{0.49\textwidth}
        \centering
        %        \resizebox{\linewidth}{!}{
        \begin{tikzpicture}%[scale=1.]
            
            \node[quantified] (z1) at (-1,0) {$z_1$};
            \node[quantified] (z2) at (0,0) {$z_2$};
            \node[quantified] (z3) at (1,0) {$z_3$};
            
            \node[quantified] (x1) at (-.5,1) {$x_1$};
            \node[quantified] (x2) at (.5,1) {$x_2$};
            
            \node[quantified] (y) at (0,-1)  {$y$};
            
            \draw
            
            (x1) -- (x2)
            
            (x1) -- (z1)
            (x1) -- (z2)
            (x1) -- (z3)
            
            (x2) -- (z1)
            (x2) -- (z2)
            (x2) -- (z3)
            
            (z1) -- (y)
            (z2) -- (y)
            (z3) -- (y);

        \end{tikzpicture}
        \caption{CQ $p_2$}
        %    }
    \end{subfigure}
    \caption{Queries $q_2$ and $p_2$ from
      Example~\ref{ex:altering-tw-not-full-schema}.}
    \label{fig:lastexample}
    %    \begin{subfigure}[b]{0.24\textwidth}
    %        \centering
    %%        \resizebox{\linewidth}{!}{
    %        \begin{tikzpicture}%[scale=1.]
    %        
    %        
    %        \node[quantified] (z1) at (-1,0);
    %        \end{tikzpicture}
    %        \caption{CQ $p_2$}
    %%    }
    %    \end{subfigure}
\end{figure}

We now turn to the proof of Theorem~\ref{thm:CTWSSjointapprox}.  Here,
we present only the statement about starsize made in
Theorem~\ref{thm:CTWSSjointapprox}, restated as
Lemma~\ref{thm:SSapprox} below. The statement about contract treewidth
is proved in the appendix. The proof follows a similar strategy as for
starsize, but is a bit more involved.

A \emph{pointed \Sbf-database} is a
pair $(D,\bar c)$ with $D$ an \Sbf-database and $\bar c$ a tuple of
constants from $\mn{adom}(D)$.  The \emph{contract treewidth} and
\emph{starsize} of $(D,\bar c)$ are that of $D$ viewed as a
conjunctive query with constants from $\bar c$ playing the role of
answer variables.

\begin{lem} %{theorem}{SSapprox}
    \label{thm:SSapprox}
    Let $(\Omc,\Sbf,q) \in (\class{G},\class{UCQ})$ be an OMQ and $k
    \geq \mn{ar}(\Sbf)$. Then $Q^{\text{SS}}_k$ is an
    SS$_k$-approximation of~$Q$.
    %    and can be computed in double exponential time.
\end{lem}
%\SSapprox*
%
% The proof has a similar structure to the one above. In the proof below, we may confuse a tuples with a set of it elements.
% It should follow from the context, which we mean.
%
\begin{proof}
  Let $Q(\bar{x})=(\Omc,\Sbf,q) \in (\class{G},\class{UCQ})$.  By
  construction of $Q_{k}^{SS} = (\Omc, \Sbf, q^{\text{SS}}_k)$, it is
  clear that Points~1 and~2 of the definition of SS$_k$-approximations
  are satisfied.  It remains to establish Point~3.

  Let $P(\bar x) = (\Omc',\Sbf,p) \in (\class{G},\class{UCQ})$ such
  that $P \subseteq Q$ with $p$ of starsize at most $k$.  We have to
  show that $P \subseteq Q_{k}^{\text{SS}}$, i.e., $\bar c \in P(D)$
  implies $\bar{c} \in Q_{k}^{\text{SS}}(D)$~for all
  $\Sbf$\=/databases $D$.  Thus let $D$ be an~\Sbf-database and let
  $\bar c \in P(D)$. Since $P \subseteq Q$, we have
  $\bar{c} \in Q(D)$.  We construct a pointed $\Sbf$\=/database
  $(D',\bar{c})$ such that
    \begin{enumerate}
        \item \label{item:ss-approximation-inclusion-property} $\bar{c} \in P(D')$,
        \item \label{item:ss-approximation-small-ss} the starsize of
          $(D',\bar{c})$ at most $k$, and
        \item \label{item:ss-approximation-h-D'-to-D} there is a
          homomorphism %\fun{h_D}{D'}{D}
          from $D'$ to $D$ that is the~identity on
        $\bar{c}$.

    \end{enumerate}
    In the following, we consider sets $S$ of constants that occur in
    $\bar c$, with $|S| \leq k$. Let \Smc denote the set of all such
    sets $S$. For every $S \in \Smc$, let $D_S$ denote the database
    obtained from $D$ by renaming every constant $c \notin S$ to
    $c^{S}$.
%
    % that
    % borrows subdomain $\bar b$ of $D$'', i.e.,
    % $\dom(D) \cap \dom(D_{\bar{b}}) = \bar{b}$. For distinct
    % $\bar{b},\bar{b}'$, we choose $D_{\bar{b}}$ and $D_{\bar{b}'}$
    % such that
    % $\dom(D_{\bar{b}}) \cap \dom(D_{\bar{b}'}) = \bar b \cap {\bar
    %   b}'$.
    % More precisely, $\dom(D_{\bar{b}}) = \{c_{\bar{b}}\}_{c \in \dom(D)}$, where
    % the constants $c_{\bar{b}}$ for $c \notin \bar{b}$ are fresh and $c_{\bar{b}} = c$ for $c \in \bar{b}$.
    % The content of $D_{\bar{b}}$ we define as follows.
    % \[
    % D_{\bar{b}} = \{ R(c_{\bar{b}}^1, c_{\bar{b}}^2, \dots, c_{\bar{b}}^n) \mid R(c^1,c^2,\dots, c^n) \in D\}
    % \]
    We then define
    \[
    %    \begin{equation}
    \label{eq:ss-approximation-new-database}
    D' = \bigcup_{S \in \Smc} D_S.
    %    D' = \bigcup_{\bar{b} \subseteq \bar{a}, |\bar{b}| \leq k} D_{\bar{b}} \cup D_{|\bar{a}},
    %    \end{equation}
    \]
    % where $D_{\bar{b}}$ are chosen so that for every $\bar{b},\bar{c}$ we have that $\dom(D_{\bar{b}}) \cap \dom(D_{\bar{c}}) = \bar{b} \cap \bar{c}$
    % and $D_{|\bar{a}}$ is the restriction of $D$ to constants $\bar{a}$, i.e.~$D_{|\bar{a}} = \{ R(b_1, b_2, \dots, b_n) \in D \mid \forall i.\ b_i \in \bar{a} \}$.
    \noindent
    By definition, $(D', \bar{c})$ has no $\bar{c}$\=/component with
    more than $k$ constants from $\bar{c}$, and thus
    Point~\ref{item:ss-approximation-small-ss} is~satisfied.
    Point~\ref{item:ss-approximation-h-D'-to-D} is clear by
    construction of $(D', \bar{c})$.  We need to show that
    Point~\ref{item:ss-approximation-inclusion-property} also holds.

    Since $\bar{c} \in P(D)$, there is a homomorphism $h$ from some CQ
    $p'(\bar x)$ in $p$ to $\chase_{\Omc'}(D)$. We construct a
    homomorphism $h'$ from $p'$ to $\chase_{\Omc'}(D)$, which shows
    $\bar{c} \in P(D')$ as desired.

    For every $S \in \Smc$, there is a homomorphism (even isomorphism)
    $h_S$ from $D$ to $D_S$ that is the identity on $S$. This
    homomorphism can be extended to a homomorphism from
    $\mn{ch}_{\Omc'}(D)$ to $\mn{ch}_{\Omc'}(D_S)$ by following the
    chase steps used to construct $\mn{ch}_{\Omc'}(D)$. Moreover,
    $\mn{ch}_{\Omc'}(D_S) \subseteq \mn{ch}_{\Omc'}(D')$, and thus we
    can view $h_S$ as a homomorphism from $\mn{ch}_\Omc(D)$ to
    $\mn{ch}_\Omc(D')$ that is the identity on $S$.

    Now for the construction of $h'$. For all answer variables $x$ in
    $p'$, we set $h'(x)=h(x)$. Note that this yields
    $h'(\bar x)=\bar c$. For every quantified variable $y$, let $S$ be
    the set of answer variables that are part of the unique
    $\bar x$-component that contains $y$. Then set
    $h'(y)=h_S \circ h(y)$. This is well-defined since $p$ has
    starsize at most $k$, and thus $|S| \leq k$ implying $S \in \Smc$.

    We argue that $h'$ is indeed a homomorphism. For every atom
    $R(\bar{z}) \in p'$, we have
    $R(h(\bar{z})) \in \chase_{\Omc'}(D)$.
    %Since the starsize of $p_i$ is at most $k$,
    %either all variables in $\bar{z}$ are answer variables
    %or all variables in $\bar{z}$ belong to a single
    %$\bar{x}$\=/component that contains at most $k$ answer variables.
    First assume that the variables in $\bar{z}$ are all answer
    variables. Let $S$ be the set of all constants in $h(\bar z)$. We
    have $|S| \leq k$ since $k \geq \maxarity$.  Since
    $h'(\bar z)=h(\bar z)$ and $h_S$ is the identity on $S$,
    $R(h(\bar{z})) \in \chase_{\Omc'}(D)$ implies
    $R(h'(\bar{z})) \in \chase_{\Omc'}(D')$, as required.  Now assume
    that $\bar z$ contains at least one quantified variable. Then all
    variables in $\bar{z}$ belong to the same $\bar{x}$\=/component of
    $p'$.  Let $S$ be the set of constants $h(x)$ such that $x$ is an
    answer variable in this $\bar x$-component. Then
    $h'(\bar z)=h_S \circ h(\bar z)$ and we are done. We have thus
    established Point~\ref{item:ss-approximation-inclusion-property}
    above.

    \medskip

    From $P \subseteq Q$ and $\bar{c} \in P(D')$, we obtain
    $\bar{c} \in Q(D')$.  Thus, for some CQ $q'$ in $q$, there is a
    homomorphism $g$ from $q'$ to $\mn{ch}_{\Omc}(D')$ such that
    $g(\bar{x}) = \bar{c}$. Let $\widehat q$ denote the collapsing of $q'$
    that is obtained by identifying $y_1$ and $y_2$ whenever
    $g(y_1)=g(y_2)$ with at least one of $y_1,y_2$ a quantified
    variable and adding $x_1=x_2$ whenever $g(x_1)=g(x_2)$ and
    $x_1,x_2$ are both answer variables.
%
    % Let $\widehat q$ be obtained from $q'$ by
    % any identifying variables $x_1, x_2$ with $g(x_1)=g(x_2)$ and
    % at least one of $x_1,x_2$ quantified.
    %
    % i.e.~$q''$
    %     is~a~collapsing of $q'$ obtained by~adding atom $x_1 = x_2$ for all pairs $x_1,x_2$ of variables of $q'$ such that $g(x_1) = g(x_2)$.
    %
    Then $g$ is also a homomorphism from $\widehat q$ to
    $\mn{ch}_{\Omc}(D')$. By Point~3, there is a homomorphism $h_D$
    from $D'$ to $D$, which can be extended to a homomorphism from
    $\mn{ch}_\Omc(D')$ to $\mn{ch}_\Omc(D)$. The composition
    $h_D \circ g$ is a homomorphism from $\widehat q$ to
    $\mn{ch}_\Omc(D)$, and thus $\bar c \in Q^{SS}_k(D)$. To finish
    the proof, it thus remains to show that $\widehat q$ is a CQ in
    $q^{SS}_k$.

    \medskip

    Assume to the contrary of what is to be shown that the starsize of
    $\widehat q(\bar{x})$ is at least
    $\ell = \max \{k, \mn{ar}(\Sbf)\} + 1$. Then, there is an
    $\bar{x}$\=/component $S$ of $\widehat q$ with at least $\ell$
    distinct answer variables, say $x_1, x_2, \dots, x_\ell$ such that
    $\widehat q$ does not contain {\color{\highlightColourTwo} atoms} $x_i=x_j$ for
    $1 \leq i < j \leq \ell$. Let $y$ be a quantified variable in
    $S$. By definition of $\bar x$-components, $G_{\widehat q}$
    contains (simple) paths $P_i$ between $y$ and $x_i$, for
    $1 \leq i \leq \ell$. % Clearly, we can choose $y$ so that the paths
    % $P_1,\dots,P_\ell$ do not share any inner nodes.
    Together with
    the homomorphism $g$, each path $P_i$ gives rise to a path $P'_i$
    in $G_{\chase_{\Omc}(D')}$ between $a = h(y)$ and $c_i=h(x_i)$,
    for $1 \leq i \leq \ell$.  By definition of~$\widehat q$,
    $g(z_1)=g(z_2)$ implies that $z_1=z_2$ or $z_1,z_2$ are both
    answer variables and $\widehat q$ contains an equality atom
    $z_1=z_2$. %We thus know the following:
    {\color{\highlightColourTwo} It follows:}
    \begin{enumerate}

        \item[(a)] the constants
        $c_1, \dots, c_\ell$ and $a$ are all different;
        %$a_1, \dots, a_l, \dots, a_n$ and $c$ are all different;

         \item[(b)]
         $a$ is different
         from all constants in $\bar{c}$;

        \item[(c)] path $P'_i$ contains no constants from $\bar{c}$
          as inner nodes.

        % \item[(d)] the paths $P'_1,\dots,P'_\ell$ do not share any
        %   inner nodes.

    \end{enumerate}
    { \color{\highlightColourTwo}
    First assume that $a \in \dom(D')$.
    %Assume the contrary.
    An easy
    analysis of the chase shows that, due to the existence of the path
    $P'_i$ and since all TGDs in $\Omc$ are guarded, for every $1 \leq i \leq \ell$ there is
    a~path $P''_i$ in $G_{D'}$ between $c_i$ and $a$ such that $P''_i$ uses no
    constants introduced by the chase. In fact, we can obtain
    $P''_i$ from $P'_i$ by dropping all constants that have been
    introduced by the chase.  It then follows from (a) to (c) that the
    starsize of $(D',\bar{c})$ is at least $\ell$, a contradiction.}

    Now assume that $a \notin \dom(D')$.
    % Then $a$ is in the tree-like
    % structure that the chase has generated below some fact
    % $R(b_1,\dots,b_n)$.\footnote{This can be made precise in the same
    %   way as in the proof of Lemma~\ref{prop:firstapprox}. We prefer
    %   to remain on the intuitive level here to not distract from the
    %   main proof.}
    Let $b_i$ be the last constant on the path $P'_i$ that is in
    $\dom(D')$ when traveling the path from $c_i$ to $a$. Thus, the
    subpath of $P'_i$ that connects (the last occurrence of) $b_i$
    with $a$ uses only constants introduced by the chase as inner
    nodes.  Another easy analysis of the chase reveals that since all
    paths $P'_1,\dots,P'_\ell$ end at the same constant $a$, there
    must be a fact in $D'$ that contains all of $b_1,\dots,b_\ell$.
    Note that $ \{c_1,\dots,c_\ell\} \subseteq \{b_1,\dots,b_\ell\}$
    is impossible since $\mn{ar}(\Sbf) <\ell$.  It thus follows from
    (c) that some $b_i$ is not in $\bar c$. Consequently, there is a
    path $P''_i$ in $G_{D'}$ that connects $c_i$ and $b_i$ and uses no
    constants from $\bar{c}$ as inner nodes, for $1 \leq i \leq \ell$.
    We may again obtain $P''_i$ by dropping constants introduced by
    the chase.  This implies that the starsize of $(D',\bar{c})$ is at
    least $\ell$, a contradiction.
\end{proof}

%%% Local Variables:
%%% mode: latex
%%% TeX-master: "paper_main"
%%% End:

\section{Conclusions}
\label{sec:conclusions}

We have provided a complexity classification for counting the number
of answers to UCQs in the presence of TGDs that applies both to
ontology-mediated querying and to querying under constraints.
%
% ontology-mediated queries from $(\class{G},\class{UCQ})$
% that are based on the full schema. % All upper bounds except the \PTime
% % one for the description logics $\mathcal{ELH}^\mn{dr}$ also apply to
% % the more general case where the data schema is not required to be full.
% Our classification essentially follows the classification of UCQs
% provided in \cite{DBLP:conf/icalp/DellRW19} and uses the same %semantic
% measures. % We
% remark, however, that % `semantic' has a different meaning in our case,
% % which makes the technical development more difficult. In fact,
% in our case
% it is
% possible to exploit and even modify the ontology when transitioning to
% an equivalent OMQ with smaller measures.
%
The classification also applies to ontology-mediated querying with the
OMQ language $(\mathcal{ELIH},\text{UCQ})$ where $\mathcal{ELIH}$ is a
well-known description logic \cite{DBLP:books/daglib/0041477}. In
fact, this is immediate if the ontologies in OMQs are in a certain
well-known normal form that avoids nesting of concepts
\cite{DBLP:books/daglib/0041477}.  In the general case, it suffices to
observe that all our proofs extended from guarded TGDs to
frontier-guarded TGDs \cite{BLMS11} with bodies of bounded treewidth,
a strict generalization of $\mathcal{ELIH}$. In contrast, a complexity
classification for OMQs based on frontier-guarded TGDs with
unrestricted bodies is an interesting problem for future work.

There are several other interesting questions that remain open. In
querying under constraints that are guarded TGDs, does answer counting
in FPT coincide with answer counting in \PTime? Do our results extend
to ontology-mediated querying when the data schema is not required to
be full? What happens when we drop the restriction that relation
symbols are of bounded arity? What about OMQs and CQSs based on other
decidable classes of TGDs?  And how can we decide the meta problems
for the important structural measure of treewidth when the ontology
needs not be preserved, with full data schema or even with
unrestricted data schema?

%%% Local Variables:
%%% mode: latex
%%% TeX-master: "paper_main"
%%% End:

\bibliographystyle{alphaurl}
\bibliography{biblio}

\newcommand{\etalchar}[1]{$^{#1}$}
\begin{thebibliography}{BtCLW14}

\bibitem[ACJR21]{counting-cq-approx}
Marcelo Arenas, Luis~Alberto Croquevielle, Rajesh Jayaram, and Cristian
  Riveros.
\newblock When is approximate counting for conjunctive queries tractable?
\newblock In {\em Proc.\ of {STOC}}, pages 1015–--1027, 2021.
\newblock \href {https://doi.org/10.1145/3406325.3451014}
  {\path{doi:10.1145/3406325.3451014}}.

\bibitem[AHV95]{AbHV95}
Serge Abiteboul, Richard Hull, and Victor Vianu.
\newblock {\em Foundations of Databases}.
\newblock Addison-Wesley, 1995.
\newblock URL: \url{http://webdam.inria.fr/Alice/}.

\bibitem[BBP18]{BaBP18}
Pablo Barcel{\'{o}}, Gerald Berger, and Andreas Pieris.
\newblock Containment for rule-based ontology-mediated queries.
\newblock In {\em Proc.\ of {PODS}}, pages 267--279, 2018.
\newblock \href {https://doi.org/10.1145/3196959.3196963}
  {\path{doi:10.1145/3196959.3196963}}.

\bibitem[BDF{\etalchar{+}}20]{BDFLP-PODS20}
Pablo Barcel\'{o}, Victor Dalmau, Cristina Feier, Carsten Lutz, and Andreas
  Pieris.
\newblock The limits of efficiency for open- and closed-world query evaluation
  under guarded {{TGD}s}.
\newblock In {\em Proc.\ of {PODS}}, pages 259--270, 2020.
\newblock \href {https://doi.org/10.1145/3375395.3387653}
  {\path{doi:10.1145/3375395.3387653}}.

\bibitem[BFGP20]{BFGP20}
Pablo Barcel\'{o}, Diego Figueira, Georg Gottlob, and Andreas Pieris.
\newblock Semantic optimization of conjunctive queries.
\newblock {\em J. ACM}, 67(6), 2020.
\newblock \href {https://doi.org/10.1145/3424908} {\path{doi:10.1145/3424908}}.

\bibitem[BFLP19]{BFLP19}
Pablo Barcel{\'{o}}, Cristina Feier, Carsten Lutz, and Andreas Pieris.
\newblock When is ontology-mediated querying efficient?
\newblock In {\em Proc.\ of LICS}, pages 1--13, 2019.
\newblock \href {https://doi.org/10.1109/LICS.2019.8785823}
  {\path{doi:10.1109/LICS.2019.8785823}}.

\bibitem[BGO10]{DBLP:conf/lics/BaranyGO10}
Vince B{\'{a}}r{\'{a}}ny, Georg Gottlob, and Martin Otto.
\newblock Querying the guarded fragment.
\newblock In {\em Proc.\ of {LICS}}, pages 1--10, 2010.

\bibitem[BGP16]{BaGP16}
Pablo Barcel{\'{o}}, Georg Gottlob, and Andreas Pieris.
\newblock Semantic acyclicity under constraints.
\newblock In {\em Proc.\ of {PODS}}, pages 343--354, 2016.
\newblock \href {https://doi.org/10.1145/2902251.2902302}
  {\path{doi:10.1145/2902251.2902302}}.

\bibitem[BHLS17]{DBLP:books/daglib/0041477}
Franz Baader, Ian Horrocks, Carsten Lutz, and Ulrike Sattler.
\newblock {\em An Introduction to Description Logic}.
\newblock Cambridge University Press, 2017.
\newblock \href {https://doi.org/10.1017/9781139025355}
  {\path{doi:10.1017/9781139025355}}.

\bibitem[BLMS11]{BLMS11}
Jean-Fran\c{c}ois Baget, Michel Lecl{\`e}re, Marie-Laure Mugnier, and Eric
  Salvat.
\newblock On rules with existential variables: {W}alking the decidability line.
\newblock {\em Artif. Intell.}, 175(9-10):1620--1654, 2011.
\newblock \href {https://doi.org/10.1016/j.artint.2011.03.002}
  {\path{doi:10.1016/j.artint.2011.03.002}}.

\bibitem[BLR14]{DBLP:journals/siamcomp/BarceloL014}
Pablo Barcel{\'{o}}, Leonid Libkin, and Miguel Romero.
\newblock Efficient approximations of conjunctive queries.
\newblock {\em {SIAM} J. Comput.}, 43(3):1085--1130, 2014.

\bibitem[BMT20]{Meghyn-IJCAI20}
Meghyn Bienvenu, Quentin Mani\`ere, and Micha\"el Thomazo.
\newblock Answering counting queries over {DL}-{L}ite ontologies.
\newblock In {\em Proc.\ of IJCAI}, pages 1608--1614, 2020.
\newblock \href {https://doi.org/10.24963/ijcai.2020/223}
  {\path{doi:10.24963/ijcai.2020/223}}.

\bibitem[BMT21a]{DBLP:conf/ijcai/BienvenuMT21}
Meghyn Bienvenu, Quentin Mani{\`{e}}re, and Micha{\"{e}}l Thomazo.
\newblock Cardinality queries over {DL}-{L}ite ontologies.
\newblock In {\em Proc.\ of {IJCAI}}, pages 1801--1807, 2021.

\bibitem[BMT21b]{DBLP:conf/dlog/BienvenuMT21}
Meghyn Bienvenu, Quentin Mani{\`{e}}re, and Micha{\"{e}}l Thomazo.
\newblock Counting queries over {{ELHI}$_\bot$} ontologies.
\newblock In {\em Proc.\ of {DL}}, 2021.

\bibitem[BMT22]{DBLP:conf/kr/BienvenuMT22}
Meghyn Bienvenu, Quentin Mani{\`{e}}re, and Micha{\"{e}}l Thomazo.
\newblock Counting queries over {{ELHI}$_\bot$} ontologies.
\newblock In {\em Proc.\ of {KR}}, 2022.

\bibitem[BO15]{DBLP:conf/rweb/BienvenuO15}
Meghyn Bienvenu and Magdalena Ortiz.
\newblock Ontology-mediated query answering with data-tractable description
  logics.
\newblock In {\em Proc.\ of {Reasoning Web}}, pages 218--307, 2015.
\newblock \href {https://doi.org/10.1007/978-3-319-21768-0_9}
  {\path{doi:10.1007/978-3-319-21768-0_9}}.

\bibitem[BtCLW14]{DBLP:journals/tods/BienvenuCLW14}
Meghyn Bienvenu, Balder ten Cate, Carsten Lutz, and Frank Wolter.
\newblock Ontology-based data access: {A} study through disjunctive datalog,
  {CSP}, and {MMSNP}.
\newblock {\em {ACM} Trans. Database Syst.}, 39(4):33:1--33:44, 2014.
\newblock \href {https://doi.org/10.1145/2661643} {\path{doi:10.1145/2661643}}.

\bibitem[CCLR20]{Diego-IJCAI20}
Diego Calvanese, Julien Corman, Davide Lanti, and Simon Razniewski.
\newblock Counting query answers over {DL}-{L}ite knowledge base.
\newblock In {\em Proc.\ of IJCAI}, pages 1658--1666, 2020.
\newblock \href {https://doi.org/10.24963/ijcai.2020/230}
  {\path{doi:10.24963/ijcai.2020/230}}.

\bibitem[CGK13]{CaGK13}
Andrea Cal\`{\i}, Georg Gottlob, and Michael Kifer.
\newblock Taming the infinite chase: Query answering under expressive
  relational constraints.
\newblock {\em J. Artif. Intell. Res.}, 48:115--174, 2013.
\newblock \href {https://doi.org/10.1613/jair.3873}
  {\path{doi:10.1613/jair.3873}}.

\bibitem[CGL98]{DBLP:conf/pods/CalvaneseGL98}
Diego Calvanese, Giuseppe~De Giacomo, and Maurizio Lenzerini.
\newblock On the decidability of query containment under constraints.
\newblock In {\em Proc.\ of {PODS}}, pages 149--158, 1998.
\newblock \href {https://doi.org/10.1145/275487.275504}
  {\path{doi:10.1145/275487.275504}}.

\bibitem[CGP12]{CaGP12}
Andrea Cal\`{\i}, Georg Gottlob, and Andreas Pieris.
\newblock Towards more expressive ontology languages: {T}he query answering
  problem.
\newblock {\em Artif. Intell.}, 193:87--128, 2012.
\newblock \href {https://doi.org/10.1016/j.artint.2012.08.002}
  {\path{doi:10.1016/j.artint.2012.08.002}}.

\bibitem[CM15]{countingTrichotomy}
Hubie Chen and Stefan Mengel.
\newblock A trichotomy in the complexity of counting answers to conjunctive
  queries.
\newblock In {\em Proc. of {ICDT}}, pages 110--126, 2015.
\newblock \href {https://doi.org/10.4230/LIPIcs.ICDT.2015.110}
  {\path{doi:10.4230/LIPIcs.ICDT.2015.110}}.

\bibitem[CM16]{countingPositiveQueries}
Hubie Chen and Stefan Mengel.
\newblock Counting answers to existential positive queries: {A} complexity
  classification.
\newblock In {\em Proc. of {PODS}}, pages 315--326, 2016.
\newblock \href {https://doi.org/10.1145/2902251.2902279}
  {\path{doi:10.1145/2902251.2902279}}.

\bibitem[DJ04]{countingHomomorphisms}
V{\'{\i}}ctor Dalmau and Peter Jonsson.
\newblock The complexity of counting homomorphisms seen from the other side.
\newblock {\em J. Theor. Comput. Sci.}, 329(1-3):315--323, 2004.
\newblock \href {https://doi.org/10.1016/j.tcs.2004.08.008}
  {\path{doi:10.1016/j.tcs.2004.08.008}}.

\bibitem[DM14]{DBLP:journals/jcss/DurandM14}
Arnaud Durand and Stefan Mengel.
\newblock The complexity of weighted counting for acyclic conjunctive queries.
\newblock {\em J. Comput. Syst. Sci.}, 80(1):277--296, 2014.
\newblock \href {https://doi.org/10.1016/j.jcss.2013.08.001}
  {\path{doi:10.1016/j.jcss.2013.08.001}}.

\bibitem[DM15]{countingQueriesSructural}
Arnaud Durand and Stefan Mengel.
\newblock Structural tractability of counting of solutions to conjunctive
  queries.
\newblock {\em J. Theory Comput. Syst.}, 57(4):1202--1249, 2015.
\newblock \href {https://doi.org/10.1007/s00224-014-9543-y}
  {\path{doi:10.1007/s00224-014-9543-y}}.

\bibitem[DRW19]{DBLP:conf/icalp/DellRW19}
Holger Dell, Marc Roth, and Philip Wellnitz.
\newblock Counting answers to existential questions.
\newblock In {\em Proc.\ of {ICALP}}, pages 113:1--113:15, 2019.
\newblock \href {https://doi.org/10.4230/LIPIcs.ICALP.2019.113}
  {\path{doi:10.4230/LIPIcs.ICALP.2019.113}}.

\bibitem[Fag80]{faginSTOC80}
Ronald Fagin.
\newblock Horn clauses and database dependencies (extended abstract).
\newblock In {\em Proc. of\ {STOC}}, pages 123--134, 1980.
\newblock \href {https://doi.org/10.1145/800141.804660}
  {\path{doi:10.1145/800141.804660}}.

\bibitem[Fei22]{feier-icdt-2022}
Cristina Feier.
\newblock Characterising fixed parameter tractability for query evaluation over
  guarded {TGD}s.
\newblock In {\em Proc.\ of {ICDT}}, pages 12:1--12:20, 2022.
\newblock \href {https://doi.org/10.4230/LIPIcs.ICDT.2022.12}
  {\path{doi:10.4230/LIPIcs.ICDT.2022.12}}.

\bibitem[FG04]{DBLP:journals/siamcomp/FlumG04}
J{\"{o}}rg Flum and Martin Grohe.
\newblock The parameterized complexity of counting problems.
\newblock {\em {SIAM} J. Comput.}, 33(4):892--922, 2004.
\newblock \href {https://doi.org/10.1137/S0097539703427203}
  {\path{doi:10.1137/S0097539703427203}}.

\bibitem[FGRZ21]{DBLP:journals/corr/abs-2103-12468}
Jacob Focke, Leslie~Ann Goldberg, Marc Roth, and Stanislav Zivn{\'{y}}.
\newblock Approximately counting answers to conjunctive queries with
  disequalities and negations.
\newblock {\em CoRR}, abs/2103.12468, 2021.

\bibitem[Fig16]{DiegoFunctional16}
Diego Figueira.
\newblock Semantically acyclic conjunctive queries under functional
  dependencies.
\newblock In {\em Proc. of LICS}, page 847–856, 2016.
\newblock \href {https://doi.org/10.1145/2933575.2933580}
  {\path{doi:10.1145/2933575.2933580}}.

\bibitem[FKMP05]{FKMP05}
Ronald Fagin, Phokion~G. Kolaitis, Ren{\'{e}}e~J. Miller, and Lucian Popa.
\newblock Data exchange: semantics and query answering.
\newblock {\em J. Theor. Comput. Sci.}, 336(1):89--124, 2005.
\newblock \href {https://doi.org/10.1016/j.tcs.2004.10.033}
  {\path{doi:10.1016/j.tcs.2004.10.033}}.

\bibitem[FLP21]{DBLP:conf/icdt/FeierLP21}
Cristina Feier, Carsten Lutz, and Marcin Przybylko.
\newblock Answer counting under guarded {TGD}s.
\newblock In {\em Proc.\ of {ICDT}}, pages 11:1--11:22, 2021.

\bibitem[Gro07]{Grohe07}
Martin Grohe.
\newblock The complexity of homomorphism and constraint satisfaction problems
  seen from the other side.
\newblock {\em J. {ACM}}, 54(1):1:1--1:24, 2007.
\newblock \href {https://doi.org/10.1145/1206035.1206036}
  {\path{doi:10.1145/1206035.1206036}}.

\bibitem[JK84]{JoKl84}
David~S. Johnson and Anthony~C. Klug.
\newblock Testing containment of conjunctive queries under functional and
  inclusion dependencies.
\newblock {\em J. Comput. Syst. Sci.}, 28(1):167--189, 1984.
\newblock \href {https://doi.org/10.1016/0022-0000(84)90081-3}
  {\path{doi:10.1016/0022-0000(84)90081-3}}.

\bibitem[KK18]{DBLP:journals/ojsw/KostovK18}
Bogdan Kostov and Petr Kremen.
\newblock Count distinct semantic queries over multiple linked datasets.
\newblock {\em Open J. Semantic Web}, 5(1):1--11, 2018.

\bibitem[KR15]{DBLP:journals/ws/KostylevR15}
Egor~V. Kostylev and Juan~L. Reutter.
\newblock Complexity of answering counting aggregate queries over {DL}-{L}ite.
\newblock {\em J. Web Semant.}, 33:94--111, 2015.
\newblock \href {https://doi.org/10.1016/j.websem.2015.05.003}
  {\path{doi:10.1016/j.websem.2015.05.003}}.

\bibitem[LMTV19]{LMTV19}
Nicola Leone, Marco Manna, Giorgio Terracina, and Pierfrancesco Veltri.
\newblock Fast query answering over existential rules.
\newblock {\em {ACM} Trans. Comput. Log.}, 20(2):12:1--12:48, 2019.
\newblock \href {https://doi.org/10.1145/3308448} {\path{doi:10.1145/3308448}}.

\bibitem[MMS79]{MaMS79}
David Maier, Alberto~O. Mendelzon, and Yehoshua Sagiv.
\newblock Testing implications of data dependencies.
\newblock {\em ACM Trans. Database Syst.}, 4(4):455--469, 1979.
\newblock \href {https://doi.org/10.1145/320107.320115}
  {\path{doi:10.1145/320107.320115}}.

\bibitem[PLC{\etalchar{+}}08]{DBLP:journals/jods/PoggiLCGLR08}
Antonella Poggi, Domenico Lembo, Diego Calvanese, Giuseppe~De Giacomo, Maurizio
  Lenzerini, and Riccardo Rosati.
\newblock Linking data to ontologies.
\newblock {\em J. Data Semantics}, 4900:133--173, 2008.
\newblock \href {https://doi.org/10.1007/978-3-540-77688-8_5}
  {\path{doi:10.1007/978-3-540-77688-8_5}}.

\bibitem[PS13]{DBLP:journals/jcss/PichlerS13}
Reinhard Pichler and Sebastian Skritek.
\newblock Tractable counting of the answers to conjunctive queries.
\newblock {\em J. Comput. Syst. Sci.}, 79(6):984--1001, 2013.
\newblock \href {https://doi.org/10.1016/j.jcss.2013.01.012}
  {\path{doi:10.1016/j.jcss.2013.01.012}}.

\end{thebibliography}

\newpage

\appendix

\section{Preliminary Notes}

The main purpose of
Sections~\ref{app:firstblack},~\ref{app:secondblack},
and~\ref{app:thirdblack} of the appendix is to provide proof sketches
for the results that we take over from Chen and Mengel, that is,
Lemma~\ref{lem:countingEquivClassD},
Theorem~\ref{thm:chenMengelStronger}, and
Lemma~\ref{lemma:CQ-to-marked-CQ-byChenMengel}, respectively.  These
results are implicit in~\cite{countingTrichotomy,countingPositiveQueries}.
They are stated
there explicitly only for the class of all databases, while we need
them for classes of databases that satisfy certain properties, made
precise in the mentioned lemmas and theorem. In~Sections~\ref{app:firstblack},~\ref{app:secondblack},
and~\ref{app:thirdblack}, we summarize the proofs given in
\cite{countingTrichotomy,countingPositiveQueries} so that the reader
can convince themselves that all results indeed hold in the form stated
in the current paper. We~use our own terminology and language in the
proof sketches, so the presentation is somewhat different from the one
given in Lemma~\ref{lemma:CQ-to-marked-CQ-byChenMengel} where, for
example, relational structures are used in place of conjunctive queries.

\section{Additional Details for Section~\ref{sect:countequiv}}
\label{app:firstblack}

In this section we describe  constructions from~\cite{countingPositiveQueries}
that relate to the notions of~counting equivalence and semi\=/counting equivalence.
%In particular, we provide the arguments allowing us to
%decide both counting equivalence and semi\=/counting equivalence.
In particular, we provide the proof of Lemma~\ref{lem:countingEquivClassD}.

%\begin{lem*}[Lemma~\ref{lem:countingEquivClassD}]
%  \lemcountingEquivClassD
%\end{lem*}
%
% {\color{red}refer back to description of decision procedure in main
%    part, describe the countermodel constructions in detail!}

\bigskip
\noindent
{\bf Lemma~\ref{lem:countingEquivClassD}.}%
\textit{\lemcountingEquivClassD}

%\noindent
We start with two simple, yet crucial, observations about counting: one regarding products
and one regarding cloning. Those observation will often be used implicitly in the following three sections of the appendix.
\begin{lem}[product rule]
	\label{lem:counting-and-products}
	Let $q(\bar{x})$ be a CQ over schema $\Sbf$ and $D, D'$ be $\Sbf$\=/databases.
%	Let $p_1(\bar{y}_1), \dots, p_n(\bar{y}_n)$ be maximal connected components in $q$.%
%	\footnote{Of course, here $\bar{y}_i$ are disjoint and $\bigcup_{i} \bar{y}_i = \bar{x}$.
%		E.g.~$q(x,y) = A(x) \land B(y)$ has two maximal connected components $p_1(x) = A(x)$ and $p_2(y) = B(x)$.}
%	Then, we have the following.
%	\begin{enumerate}
%		\item $|q(D)| = \prod_{i=1}^{n} |p_i(D)|$
%		\item $|q(D \times D')| = |q(D)| \cdot |q(D')|$
%	\end{enumerate}
  Then,
  $\#q(D \times D') = \#q(D) \cdot \#q(D')$.
\end{lem}
\noindent
The proof is folklore.

We need one more definition before we formulate the statement regarding cloning.
Let $D$ be an $\Sbf$\=/database and $q(\bar{x})$ be a CQ  over schema $\Sbf$.
For a number $i \geq 0$ and a set $T{\subseteq} \dom(D)$,  by $\text{hom}_{i,T}(q, D, \bar{x})$ we denote the set of all functions $\fun{h}{\bar{x}}{\dom({D})}$
that extend to a homomorphism from $q$ to $D$ such that $h$ maps exactly $i$ variables from $\bar{x}$ to~$T$.

\begin{lem}
    \label{lem:counting-and-cloning}
    Let $D$ be an $\Sbf$\=/database,
    let $i \geq 0, j>0$ be natural numbers, and $T \subseteq \dom(D)$ be a subset of the active domain of $D$.
    Let $q(\bar{x})$ be an equality\=/free CQ over schema $\Sbf$. % with no equality atoms and no repeated variables in $\bar{x}$.

    If $D_{j}$ is a database obtained from $D$ by cloning every element from $T$ exactly $j{-}1$ times
    and $T_{j} \subseteq \dom(D_{j})$ is the set of all those clones,
    then $|\text{hom}_{i,T_j}(q, D_j, \bar{x})| = j^i|\text{hom}_{i,T}(q, D, \bar{x})|$.
\end{lem}

 \noindent
As before, the proof is straightforward.
Nevertheless, observe that in the above statement it is crucial that the answer variables are independent
and do not repeat in the tuple $\bar{x}$.
%This is the last place in Sections~\ref{app:firstblack}, \ref{app:secondblack}, \ref{app:thirdblack} where we
%explicitly state the assumption that every CQ $q(\bar{x})$ has no equality atoms and no repeated variables in $\bar{x}$.

For the notion of counting equivalence, we inspect the strongly related
notion~of being renaming equivalent.
Let $q_1(\bar{x}_1), q_2(\bar{x}_2)$ be CQs over some schema $\Sbf$.
We say that $q_1$ and $q_2$ are \emph{renaming equivalent} if there
are two surjections $\fun{h_1}{\bar{x}_1}{\bar{x}_2}$ and
$\fun{h_2}{\bar{x}_2}{\bar{x}_1}$ that can be extended
to homomorphisms $\fun{h_1}{q_1}{D_{q_2}}$ and
$\fun{h_2}{q_2}{D_{q_1}}$.
%We may refer to
%those homomorphism as \emph{renaming homomorphism}.

\begin{lem}[counting equivalence]
	\label{lem:counting-equality-cq}
	Let $q_1(\bar{x}_1),q_2({\bar{x}}_2)$ be equality\=/free CQs over schema $\Sbf$.
	Let $\class{D}$ be a class of databases such that
	\begin{itemize}
		\item $D_{q_1}, D_{q_2} \in \class{D}$ and
		\item $\class{D}$ is closed under cloning.
	\end{itemize}
	Then $q_1$ and $q_2$ are renaming equivalent or there is a $D \in \class{D}$
	such that \mbox{$\#q_1(D) \neq \#q_2(D)$}.
\end{lem}

\begin{proof}
	%    We will show that if $q_1$, $q_2$ are not renaming equivalent
	%    then we can find the required database $D$.

    We start by observing that if $|\bar{x}_1| \neq |\bar{x}_2|$
    then there is a database $D \in \class{D}$ such that $\#q_1(D) \neq q_2(D)$.%
%    \footnote{Recall that we assume that we have no equality atoms nor repeated answer variables.}

    Assume that $|\bar{x}_1| \neq |\bar{x}_2|$.
    If $\#q_1(D_{q_1}) \neq \#q_2(D_{q_1})$ then we take $D = D_{q_1}$ and we are done.
    Hence, assume otherwise, i.e.~$\#q_1(D_{q_1}) = \#q_2(D_{q_1})$. Let $D$ be $D_{q_1}$ with every element cloned once,
    in particular $|\dom(D)| = 2|\dom(D_{q_1})|$.
    Clearly, $D \in \class{D}$. Therefore, if we show that $\#q_1(D) \neq \#q_2(D)$ then we will prove the observation.

    Since  $q_1$ is equality free we have $\#q_1(D_{q_1}) > 0$.
    Moreover, the following holds:
    \[
        \#q_1(D) = 2^{|\bar{x}_1|} \#q_1(D_{q_1}) = 2^{|\bar{x}_1|} \#q_2(D_{q_1}) \neq  2^{|\bar{x}_2|} \#q_2(D_{q_1}) = \#q_2(D).
    \]
    The first equality is a consequence of Lemma~\ref{lem:counting-and-cloning}, so is the last one.
    The middle equality follows from the fact that $|x_1| \neq |x_2|$
    and the assumption that $\#q_1(D_{q_1}) = \#q_2(D_{q_1})$.
    Indeed, since $\#q_1(D_{q_1}) > 0$, we have $\#q_2(D_{q_1})  > 0$.
    Thus, $2^{|\bar{x}_1|} \#q_2(D_{q_1}) \neq  2^{|\bar{x}_2|} \#q_2(D_{q_1})$ is equivalent to $|x_1| \neq |x_2|$.
    Hence, we infer that $\#q_1(D) \neq \#q_2(D)$, which ends the proof of the observation.

	Now, we show that if $\#q_1(D) =  \#q_2(D)$ for all databases $D \in \class{D}$
	then $q_1$ and $q_2$ are renaming equivalent.
    Since $\#q_1(D) =  \#q_2(D)$ for all databases $D \in \class{D}$, the
above observation yields $|\bar{x}_1| = |\bar{x}_2|$. Hence, possibly after some renaming,
    we can assume that $\bar{x}_1 = \bar{x}_2$, drop the subscript, and simply write $\bar{x}$.

	Let $q(\bar{x})$ be a CQ over schema $\Sbf$, let $D$ be an $\Sbf$\=/database such that $\bar{x} \subseteq \dom({D})$.
	By $\mn{hom}(q,D, \bar{x})$ we denote all mappings from $\bar{x}$ to $\dom({D})$ that can
	be extended to homomorphisms from $q$ to $D$.
	Similarly, by $\mn{surj}(q,D, \bar{x})$ we denote all surjections from $\bar{x}$ to $\bar{x}$ that lie in $\mn{hom}(q,D, \bar{x})$.
	Notice that if $|\mn{surj}(q_1,D_{q_2}, \bar{x})| > 0$ and $|\mn{surj}(q_2,D_{q_1}, \bar{x})| > 0$ then, by definition,
	$q_1$ and $q_2$ are renaming equivalent.

	For $T \subseteq \bar{x}$ let $\mn{hom}_{T}(q,D, \bar{x})$
	denote the set of mappings $h \in \mn{hom}(q,D, \bar{x})$ such that
	$h(\bar{x}) \subseteq T$. By an inclusion\=/exclusion argument we get
	\[
	|\mn{surj}(q,D, \bar{x})| = \sum_{T \subseteq \bar{x}} (-1)^{|\bar{x}| - |T|} |\mn{hom}_{T}(q,D, \bar{x})|.
	\]
	We now show how to compute $\mn{hom}_{T}(q,D, \bar{x})$ for all $T \subseteq \bar{x}$.
	For $i \geq 0$, let $\mn{hom}_{i,T}(q,D, \bar{x})$ be~the set of mappings
	$h \in \mn{hom}(q,D, \bar{x})$ such that $h$ maps exactly $i$ variables from $\bar{x}$ into $T$.
	In~particular, $\mn{hom}_{T}(q,D, \bar{x}) = \mn{hom}_{|\bar{x}|,T}(q,D, \bar{x})$.
	For $j \geq 1$ and $T \subseteq \dom(D)$, let $D_{j,T}$ be a database obtained from $D$ by cloning all elements from $T$ exactly
	${j{-}1}$ times, i.e.~for every $a \in T$ the database $D_{j,T}$ has exactly $j$ clones of $a$. In particular, $D_{1,T} = D$.

	%Let $T_j$ be as in Lemma~\ref{lem:counting-and-cloning}.
	By Lemma~\ref{lem:counting-and-cloning}, for every $T \subseteq \bar{x}$ and every $j>0$, we have

	\[
	\#q(D_{j,T}) = |\mn{hom}(q,D_{j,T}, \bar{x})| = \sum_{i=0}^{|\bar{x}|} i^j |\mn{hom}_{i,T}(q,D, \bar{x})|.
	\]

	Since the above equation holds for every $j \geq 1$, by taking first $|\bar{x}|{+}1$ equations we construct a system of linear equations
	where $|\mn{hom}_{i,T}(q,D, \bar{x})|$ are the unknowns, the coefficients $i^j$ form a Vandermonde matrix, and
	$\#q(D_{j,T})$ are the constant terms. Notice that the matrix does not depend on $q$ nor $D$.
	Since the matrix has full rank, the values $|\mn{hom}_{i,T}(q,D, \bar{x})|$ are uniquely determined by, and can be effectively computed from,
	the constant terms $\#q(D_{j,T})$. In particular, the value $|\mn{hom}_{T}(q,D, \bar{x})| = |\mn{hom}_{|\bar{x}|,T}(q,D, \bar{x})|$ is uniquely determined by the constant terms, and,
	in consequence, so is the value $|\mn{surj}(q, D, \bar{x})|$.

	If we apply the above system of equations to CQ $q_1$ and database $D_{q_2}$ we can conclude that
	$|\mn{surj}(q_1, D_{q_2}, \bar{x})|$ is uniquely determined by the values $\#q_1(D)$ for databases $D$ from a certain set $S\subseteq  \class{D}$.
	Similarly, $|\mn{surj}(q_2, D_{q_2}, \bar{x})|$ is uniquely determined the same equations with the values $\#q_1(D)$ replaced by $\#q_2(D)$.
    % for the databases from the same set $S \subseteq \class{D}$.

	Now, since the equations' coefficients do not depend on the query nor on the database and since we have $\#q_1(D) = \#q_2(D)$
	for all $D \in \class{D}$, we can infer that $|\mn{surj}(q_1,D_{q_2}, \bar{x})| = |\mn{surj}(q_2,D_{q_2}, \bar{x})|$.
	Hence, $|\mn{surj}(q_1,D_{q_2}, \bar{x})| > 0$
	%    \[
	%        |\mn{surj}(q_1,D_{q_2}, \bar{x})| = |\mn{surj}(q_2,D_{q_2}, \bar{x})| \geq 1,
	%    \]
	as the identity function clearly belongs to $\mn{surj}(q_2,D_{q_2}, \bar{x})$.

	A similar reasoning shows that $|\mn{surj}(q_2,D_{q_1}, \bar{x})|>0$ and ends the proof.
\end{proof}

\bigskip

The notion of renaming equivalence is also connected to semi\=/counting equivalence.
%Let $q(\bar{x})$ be a CQ over schema $\Sbf$.
%
% NOW DEFINED IN MAIN PART
%
% Let $\hat{q}(\bar{x})$ be the CQ obtained from $q$
% by dropping all maximal connected components not containing variables from $\bar{x}$.

%\cMP{If necessary, in the lemma we can replace $\hat{q}$ with induced subdatabases}
\begin{lem}[semi-counting equivalence]
	\label{lem:semi-counting-equality-cq}
	Let $q_1(\bar{x}_1),q_2({\bar{x}_2})$ be equality\=/free CQs over schema $\Sbf$.
	Let $\class{D}$ be a class of databases such that
	\begin{itemize}
		\item $D^{\top}_{\Sbf}, D_{\hat{q}_1}, D_{\hat{q}_2} \in \class{D}$,
		\item and $\class{D}$ is closed under cloning and disjoint union.
	\end{itemize}
	Either $\hat{q}_1$, $\hat{q}_2$ are renaming equivalent or there is, and can be computed, a database $D \in \class{D}$
	such that %$|q_1(D)|>0$,  $|q_2(D)| > 0$, and $|q_1(D)| \neq |q_2(D)|$.
	$|q_1(D)| \neq |q_2(D)|$ and for every CQ $q$ over schema $\Sbf$ we have that $|q(D)| > 0$.
\end{lem}

\begin{proof}We will show that if $\hat{q}_1$, $\hat{q}_2$ are not renaming equivalent then there is a database $D \in \class{D}$
	such that %$|q_1(D)|>0$,  $|q_2(D)| > 0$, and $|q_1(D)| \neq |q_2(D)|$.
	$\#q_1(D) {\neq} \#q_2(D)$ and such that for every CQ $q$ over schema $\Sbf$ we have that $\#q(D) > 0$.

	Since $\hat{q}_1$, $\hat{q}_2$ are not renaming equivalent then, by the previous lemma, we can find
	a database $D' \in \class{D}$ such that $\#\hat{q}_1(D') \neq \#\hat{q}_2(D')$.

	Consider the function $f_1 : k \mapsto \#q_1(D' + kD_{\Sbf}^{\top})$ defined for $k \geq 1$,
    where $D' + kD_{\Sbf}^{\top}$ is the disjoint union of $D'$ and $k$ copies of $D_{\Sbf}^{\top}$.
	After some elementary transformations on $\#q_1(D' + kD_{\Sbf}^{\top})$ we can infer
	that $f_1$ is a polynomial in $k$ whose constant term, i.e. term of degree $0$, is $\#\hat{q}_1(D')$.
	Similarly, the function $f_2 : k \mapsto \#q_2(D' + kD_{\Sbf}^{\top})$ is a polynomial whose constant term is $\#\hat{q}_2(D')$.
	For the details please refer to~the proof of Theorem~5.9 in~\cite{countingPositiveQueries}.

	If the counts of $q_1$ and $q_2$ would agree on all databases from $\class{D}$ then for all $k{\geq}1$ we would have that $f_1(k) {=} f_2(k)$.
	Moreover, since $f_1,f_2$ are polynomials, this would imply that $f_1$ and $f_2$ are equal,
    i.e.~they have the same degree and their corresponding coefficients coincide.
	In particular, this would imply that $\#\hat{q}_1(D') {=} \#\hat{q}_2(D')$.
	But this is impossible, as $D'$ was chosen so that $\#\hat{q}_1(D') \neq \#\hat{q}_2(D')$.
	Therefore, there is $k' \geq 1$ such that $f_1(k') \neq f_2(k')$.

	What remains is to check that $D = D' + k'D_{\Sbf}^{\top}$ is as required.
	First, by definition $D \in \class{D}$. Moreover, since $k'>0$, for every CQ $q$ over schema $\Sbf$ we have that
	$\#q(D) > 0$. %$|q_1(D)|>0$ and  $|q_2(D)| > 0$.
	Finally, we observe that $\#q_1(D)= f_1(k') \neq f_2(k') = \#q_2(D)$.
\end{proof}

From Lemma~\ref{lem:counting-equality-cq} and Lemma~\ref{lem:semi-counting-equality-cq}
we obtain the following.

\begin{lem}
	\label{lem:counting-and-renaming-equivalence}
	Let $q_1,q_2$ be CQs over some schema $\Sbf$.
	The following holds:
	\begin{itemize}
		\item $q_1$ and $q_2$ are counting equivalent
		if and only if they are renaming equivalent;
		\item $q_1$ and $q_2$ are semi\=/counting equivalent
		if and only if $\hat{q}_1$ and $\hat{q}_2$  are renaming equivalent.
	\end{itemize}
\end{lem}

\begin{proof}
  For the first bullet we argue as follows.  If $q_1$ and $q_2$ are
  not renaming equivalent then they are not counting equivalent.
  Indeed, by Lemma~\ref{lem:counting-equality-cq} there is a database
  $D$ such that $\#q_1(D) \neq \#q_2(D)$.  On the other hand, if $q_1$
  and $q_2$ are renaming equivalent then the surjections promised by
  the definition of renaming equivalence provide for every
  $\Sbf$\=/database~$D$ surjections $\fun{g_1^D}{q_1(D)}{q_2(D)}$ and
  $\fun{g_2^D}{q_2(D)}{q_1(D)}$. {\color{\highlightColourTwo} Clearly, these are then even bijections.}  This shows that
  $\#q_1(D) = \#q_2(D)$ for every $\Sbf$\=/database $D$ and, thus,
  shows that $q_1$ and $q_2$ are counting equivalent.

	For the second bullet we observe that if $\hat{q}_1$ and $\hat{q}_2$ are not renaming equivalent then by Lemma~\ref{lem:semi-counting-equality-cq}
	there is a database $D$ such that $\#q_1(D) > 0$, $\#q_2(D) > 0$, and $\#q_1(D) \neq \#q_2(D)$. Thus, $q_1$ and $q_2$ are not semi\=/counting equivalent.
	On the other hand, if $\hat{q}_1$ and $\hat{q}_2$ are renaming equivalent then they are counting equivalent and for every database
	such that $\#q_1(D) > 0$ and $\#q_2(D) > 0$ we have that $\#{q}_1(D) = \#\hat{q}_1(D) = \#\hat{q}_1(D)  = \#q_2(D)$.
	Thus, $q_1$ and $q_2$ are semi\=/counting equivalent.
	The middle equality follows from counting equivalence, the other follow from the fact that maximal Boolean connected components
	either force the answer set to be empty or do not change the size of the answer set.
\end{proof}

As a consequence, we get the following.

\begin{lem}[equivalence relations]
	\label{lem:counting-equivalence-is-equivalence-relation}
	Counting equivalence and semi\=/counting equivalence are equivalence relations.
\end{lem}

\noindent
We can finally prove the first missing lemma.%Lemma~\ref{lem:countingEquivClassD}.
\begin{proof}[Proof of Lemma~\ref{lem:countingEquivClassD}]

For Point 1, if $q_1$ and $q_2$ are counting equivalent, then they are counting equivalent over class $\class{D}$.
On the other hand, if they are not counting equivalent then, by Lemma~\ref{lem:counting-and-renaming-equivalence}, they are not renaming equivalent. Thus,
by Lemma~\ref{lem:counting-equality-cq} there is a database $D \in \class{D}$
such that $\#q_1(D) \neq \#q_2(D)$. This implies that $q_1$ and $q_2$ are not counting equivalent over $\class{D}$.
Therefore, $q_1$ and $q_2$ are counting equivalent if and only if they are counting equivalent over $\class{D}$.

For Point 2, if $q_1$ and $q_2$ are semi\=/counting equivalent, then they are clearly semi\=/counting equivalent over class $\class{D}$.
On the other hand, if they are not semi\=/counting equivalent, then by Lemma~\ref{lem:counting-and-renaming-equivalence} $\hat{q}_1$ and $\hat{q}_2$
are not renaming equivalent. Thus, by Lemma~\ref{lem:semi-counting-equality-cq} there is a database $D \in \class{D}$
such that $\#q_1(D)>0$, $\#q_2(D)>0$, and $\#q_1(D) \neq \#q_2(D)$. Hence, $q_1$ and $q_2$ are not semi\=/counting equivalent over $\class{D}$.
Therefore, $q_1$ and $q_2$ are semi\=/counting equivalent if and only if they are semi\=/counting equivalent over $\class{D}$.
\end{proof}

\section{Additional Details for Section~\ref{sect:fromucqtocq}}
\label{app:secondblack}

%\section{Proofs regarding counting equivalence and the blackbox reduction from UCQs to CQs as in \cite{countingPositiveQueries}}

This section is dedicated to the results from \cite{countingPositiveQueries}
that culminate in the algorithm promised in the statement below.

%\begin{thm}[Theorem~\ref{thm:chenMengelStronger}]
%    \thmchenMengelStronger
%\end{thm}

\bigskip
\noindent
{\bf Theorem~\ref{thm:chenMengelStronger} \cite{countingPositiveQueries}.}%
\textit{\thmchenMengelStronger}

As mentioned before, the statements are provided in our notation and are enough for our purposes.
For the original statements please refer to~\cite{countingPositiveQueries}.

We start by proving a stronger version of Lemma~\ref{lem:semi-counting-equality-cq}.
We show that given a set of pairwise not semi\=/counting equivalent CQs,
we can always find a database that distinguishes those queries.

\begin{lem}[inequivalence witness; Lemma~5.12 from \cite{countingPositiveQueries}]
	\label{lem:semi-counting-equality-witness-database-cq}
	Let $q_1(\bar{x}_1),q_2(\bar{x}_2), \dots,$ $q_n(\bar{x}_n)$, $n>0$, be equality\=/free CQs over schema $\Sbf$ such that
    for all $1 \leq i \leq n$ we have that $|\bar{x}_i|>0$.
	Let $\class{D}$ be a class of databases, such that
	\begin{itemize}
		\item $D^{\top}_{\Sbf} \in \class{D}$,
		\item $D_{q_i}, D_{\hat{q}_i} \in \class{D}$, for $1 \leq i \leq n$,
%		\item $D_{\hat{q}_i} \in \class{D}$, for $1 \leq i \leq n$,
		\item and $\class{D}$ is closed under direct product, disjoint union, and cloning.
	\end{itemize}
	Then there is, and can be computed, a database $D \in \class{D}$ such that
	\begin{itemize}
		\item for all $1 \leq i \leq n$ $\#q_i(D) > 0$,
		\item and for all $1 \leq i,j \leq n$ if $q_i$, $q_j$ are not semi\=/counting equivalent then $\#q_i(D) \neq \#q_j(D)$.
	\end{itemize}
\end{lem}
\begin{proof}
	We construct the database inductively.
	By requirement, the constructed database~$D$ will satisfy that $\#q(D)> 0$ for every CQ $q(\bar{x})$ over schema $\Sbf$.
	Therefore, for every pair of semi\=/counting equivalent CQs $q,q'$ we will necessarily have that $\#q(D) = \#q'(D)$.
	This allows us to assume without loss of generality that the CQs $q_i$ used in the construction are pairwise not semi\=/counting equivalent.

	For the base case, i.e.~$n {=} 1$, take database $D_{\Sbf}^{\top}$. Now, let us assume that $n > 0$ and we have
	already created the database $D_n$ for the queries $q_1, \dots, q_n$. For the inductive step, we show how
	to construct database $D_{n+1}$ for the queries $q_1, \dots, q_{n}, q_{n+1}$.

	Without loss of generality, we can assume that $0 < \#q_1(D_n) < \#q_2(D_n) < \dots < \#q_n(D_n)$.
	Now, if $\#q_{n+1}(D_n) \neq \#q_{i}(D_n)$ for all $1 \leq i \leq n$ then we are done and we take $D_{n+1} = D_n$.
	Otherwise, let us assume that there is $1 {\leq} i {\leq} n$ such that $\#q_{n+1}(D_n) {=} \#q_{i}(D_n)$.

	By Lemma~\ref{lem:semi-counting-equality-cq} there is a database $D' \in \class{D}$ such that $\#q_{n+1}(D') \neq \#q_{i}(D')$
	and for every CQ $q$ over schema $\Sbf$ we have that $\#q(D') > 0$. We can assume that $\#q_{n+1}(D') > \#q_{i}(D')$. If the equality is
	reversed we simply swap $q_{n+1}$ with $q_i$ before proceeding.

	Now, we can show that there is $l>0$ such that for the database $D = D' \times (D_{n})^{l}$, where $(D_{n})^{l}$ is the direct product of $l$ copies of $D_n$,
	we have
	\[0 < \#q_1(D) < \dots < \#q_{i}(D) < \#q_{n+1}(D) < \#q_{i+1}(D) < \dots < \#q_n(D).\]
	By the product rule, see Lemma~\ref{lem:counting-and-products}, the inequality $\#q_{i}(D) < \#q_{n+1}(D)$ holds trivially for all $l>0$.
    For the remaining inequalities, observe that for any two CQs $p,p'$ over schema $\Sbf$
    and ant two $\Sbf$\=/databases $B,C$ such that $0 < \#p(B) < \#p'(B)$ and $\#p(C), \#p'(C) >0$
    we have that $1 < \frac{\#p'(B)}{\#p(B)}$. Thus, there is $l>0$
    such that $\frac{\#p(C)}{\#p'(C)} < (\frac{\#p'(B)}{\#p(B)})^l$.
    For such $l$ we have that $\#p(C)\#p(B)^l < \#p'(C)\#p'(B)^l$ and, finally, $\#p(C \times B^l) < \#p'(C\times B^l)$.

    Clearly, $\#q(D)> 0$ for every CQ $q(\bar{x})$ over schema $\Sbf$.
	Moreover, since $D$ is a product of databases from $\class{D}$ we have that $D \in \class{D}$.
	Thus, taking $D_{n+1} = D$ ends the inductive step and the whole construction.
\end{proof}

The witness produced in the above statement will not distinguish two CQs in $\mn{cl}^{\class{D}}_{\mn{CM}}(q)$ that are semi\=/counting equivalent,
but not counting equivalent. The below lemma shows that this is not necessarily a problem.

\begin{lem}[Lemma~5.18 in \cite{countingPositiveQueries}: extracting counts from semi-counting equivalence classes]
	\label{lem:recovering-cq-from-semi-counting-equivalent-cq}
	Let $p_1, \dots, p_n$ be a set of semi\=/counting equivalent equality\=/free CQs that are pairwise not counting equivalent and let $c_1, \dots, c_n$ be a set of non-zero integers.
	Let $\class{D}$ be a class of databases such that\begin{itemize}
		%        \item $D^{\top}_{\Sbf} \in \class{D}$,
		\item $D_{p_i} \in \class{D}$, for $1 \leq i \leq n$,
		\item and $\class{D}$ is closed under direct product. %disjoint union, direct product, and cloning.
	\end{itemize}

	There is an fpt algorithm that performs the following: given a database $D \in \class{D}$
	and a CQ $q \in \{p_1, \dots, p_n\}$ the algorithm computes $\#q(D)$; the algorithm may make
	calls to an oracle $\Amc$ that provides $\sum_{i} c_i \cdot \#p_i(D')$ upon being
	given a database $D' \in \class{D}$. %Here, the $p_i$s with the $c_i$s constitute the
%	parameter.

	%    There is an fpt algorithm using and oracle $\Amc$ that performs the following: given a database $D \in \class{D}$,
	%    the algorithm computes $|p_i(D)|$ for every $1 \leq i \leq n$; it may make
	%    calls to an oracle that provides $\sum_{i} c_i |p_i(D')|$ upon being
	%    given a structure $D' \in \class{D}$. Here, the $p_i$ with the $c_i$ constitute the
	%    parameter.
\end{lem}

\begin{proof}
	Let $S = \{p_1, \dots, p_n\}$. We start with the observation that
	for all non-empty subsets $S' \subseteq S$ there are, and can be computed, a CQ $q_{S'} \in S'$ and a database $D_{S'}$
	such that

	\begin{itemize}
		\item[] $(\diamond)$ $\#q_{S'}(D') > 0$ and for all $q \in S' \setminus \{q_{S'}\}$ we have that $\#q(D_{S'}) = 0$.
	\end{itemize}

	Indeed, let $q_1(\bar{x}_1),q_2(\bar{x}_2) \in S'$ be two different CQs. Since $q_1,q_2$ are semi\=/counting equivalent,
	the CQs $\hat{q}_1,\hat{q}_2$ are renaming equivalent by Lemma~\ref{lem:counting-and-renaming-equivalence}.
	Hence, there are two surjections $\fun{h_1}{\bar{x}_1}{\bar{x}_2}$ and
	$\fun{h_2}{\bar{x}_2}{\bar{x}_1}$ that can be extended
	to homomorphisms $\fun{h_1}{\hat{q}_1}{D_{\hat{q}_2}}$ and
	$\fun{h_2}{\hat{q}_2}{D_{\hat{q}_1}}$.
%	Indeed, are homomorphisms $\fun{h_1}{\hat{q}}{D_{\hat{q}'}}$ and $\fun{h_2}{\hat{q}'}{D_{\hat{q}}}$.
	Therefore, if there would be homomorphisms $\fun{g_1}{{q_1}}{D_{{q_2}}}$ and $\fun{g_2}{{q_2}}{D_{{q_1}}}$ then
	we could extend the surjections $h_1$ and $h_2$ to homomorphisms $\fun{h_1}{{q_1}}{D_{{q_2}}}$ and $\fun{h_2}{{q_2}}{D_{{q_1}}}$, respectively.
	This would imply that $q_1$ and $q_2$ are renaming equivalent and, by Lemma~\ref{lem:counting-and-renaming-equivalence}, counting equivalent.
    Since they are not counting equivalent,	one of the homomorphisms $g_1$ or $g_2$ does not exist.

	Let $q_{S'}$ be a minimal element in $S'$ with respect to the partial order defined as $q \leq q'$ if there is a homomorphism from $q'$ to $D_q$.
	Let $D_{S'} = D_{q_{S'}}$. It is easy to check that the pair $(q_{S'}, D_{{S'}})$ satisfies the requirements in~$(\diamond)$.
	Let $\textit{get-min}(S')$ be the algorithm that given a~set $S' \subseteq \{p_1, \dots, p_n\}$
	returns the pair $(q_{S'}, D_{{S'}})$.

	Finally we can describe the desired algorithm.
	For $T \subseteq \{p_1, \dots, p_n\}$,  let $\Amc_T$ be an oracle that takes a database $D \in \class{D}$
	and returns the value $\sum_{p_i \in T} c_i \cdot \#p_i(D)$.
	Then, the algorithm promised by the lemma, let us call it $\textit{compute-count}(T, q, \Amc_T, D)$,
	takes a set of CQs $T$, a CQ $q \in T$,	an oracle $\Amc_T(\cdot)$, a database $D \in \class{D}$,
	and outputs the value $\#q(D)$. The algorithm works as follows.

	First, the algorithm finds the pair $(p_i, D_i) = \textit{get-min}(T)$.
	If $p_i = q$ then it returns
    %$\frac{\Amc_T(D \times D_i)}{\#q(D_i)}$. % = \sum_{p_j \in T} c_j \cdot \frac{\#p_j(D \times D_i)}{\#q(D_i)} = \#p_i(D)$.
    $\frac{\Amc_T(D \times D_i)}{{\color{\highlightColourTwo}c_i{\cdot}\#p_i}(D_i)}$.
	Otherwise, it returns the result of the recursive call $\textit{compute-count}(T', q, \Amc_{T'}, D)$
	where $T' = T  \setminus \{p_i\}$ and $\Amc_{T'}$ is an fpt algorithm that given database $D' \in \class{D}$
	returns $\Amc_{T'}(D') = \Amc_T(D')  - \frac{\Amc_T(D \times D_i)}{\#q(D_i)} = \sum_{p_j \in T'} c_j \cdot \#p_j(D')$.
	\noindent
	The algorithm clearly belongs to \fpt. To infer that it returns the desired value,
    we observe that
    \[
    \begin{array}{r c l}
        \frac{\Amc_T(D \times D_i)}{\#p_i(D_i)} &=& \frac{1}{\#p_i(D_i)}\cdot \sum_{p_j \in T} c_j \cdot \#p_j(D \times D_i) \\
        & = &  \frac{1}{\#p_i(D_i)}\cdot \sum_{p_j \in T} c_j \cdot \#p_j(D) \cdot \#p_j(D_i) \\
        &=& \frac{1}{\#p_i(D_i)} \cdot {\color{\highlightColourTwo}c_i {\cdot}}\#p_i(D)\#p_i(D_i) = {\color{\highlightColourTwo}c_i{\cdot}}\#p_i(D).
    \end{array}
    \]
    The first equality follows from definition of $\Amc_T$, the second is the product rule, the third
    from the fact that by construction of $D_i$, for all $p_j \in T$ we have that $\#p_j(D_i) = 0$
    if and only if $p_j \neq p_i$. The last equality is trivial.
\end{proof}
\begin{lem}[reduction from CQs to UCQs, the all free case in~\cite{countingPositiveQueries}]
	\label{lem:recovering-cqs-from-cm-closure-cq}
	Let $q(\bar{x})$ be an UCQ over schema $\Sbf$ and $|\bar{x}| > 0$;
	let $p_1(\bar{x}_1), \dots, p_n(\bar{x}_n)$ be a set of equality\=/free CQs such that $|\bar{x}_i| > 0$ for $0 < i \leq n$;
    and let $c_1, \dots, c_n$ be a sequence of non\=/zero integers.
	Let $\class{D}$ be a class of databases such that\begin{itemize}
		\item $D^{\top}_{\Sbf} \in \class{D}$,
		\item $D_{p_i},D_{\hat{p}_i} \in \class{D}$, for $1 \leq i \leq n$,
%        \item $D_{\hat{p}_i} \in \class{D}$, for $1 \leq i \leq n$,
		\item and $\class{D}$ is closed under disjoint union, direct product, and cloning.
	\end{itemize}
	If $p_1, \dots, p_n$ are pairwise not counting equivalent and for every database $D \in \class{D}$
	we have that $(\dagger)\ \#q(D) = \sum_{i} c_i \cdot \#p_i(D)$, then there is an algorithm
	that
	\begin{itemize}
		\item takes as an input a database $D \in \class{D}$ and a CQ $p \in \{p_1, \dots, p_n\}$;
		\item has an access to an~oracle for $\mn{AnswerCount}(\{q\}, \class{D})$;
		\item works in time $f(K){\cdot} \textit{poly}(||D||)$, where $f$ is some computable function and $K$ is the combined size of $p_1, \dots, p_n, c_1, \dots, c_n$, and $q$;
		\item and outputs $\#p(D)$.
	\end{itemize}
\end{lem}

\begin{proof}
	Let $D_{\circ}$ be an $\Sbf$\=/database as in Lemma~\ref{lem:semi-counting-equality-witness-database-cq}.
	Let $\varphi_1, \dots, \varphi_k$ be the set of equivalence classes of the semi\=/counting equivalence.
	Then, for every $1 \leq i \leq k$ for  $q,q' \in \varphi_i$ we have that $\#q(D_{\circ}) = \#q'(D_{\circ})$.
	Let $d_i$ denote the value $\#q(D_{\circ})$ for some $q \in \varphi_i$.
	For an $\Sbf$\=/database $D \in \class{D}$,  let $\varphi_j(D) = \sum_{p_i \in \varphi_j} c_i \cdot \#p_i(D)$.

	Then, for $l \geq 0$ we have the following.

	\[\#q(D \times D_{\circ}^l) = \sum_{1 \leq i \leq n} c_i \cdot \#p_i(D \times D_{\circ}^l) = \sum_{1 \leq i \leq n} c_i \cdot \#p_i(D) (\#p_i \big(D_{\circ})\big)^l = \sum_{1 \leq j \leq k} \varphi_j(D)\cdot d_j^l\]

    First equality holds because of $(\dagger)$, the second equality holds by product rule, see Lemma~\ref{lem:counting-and-products},
    and the final equality holds by definitions of $\varphi_i(D)$ and $d_i$.

	Invoking the above equation for integers $l=1, 2, \dots, k$
    we obtain a system of linear equations where $\varphi_j(D)$ are the unknowns, $(d_j)^l$ form the matrix of coefficients,
	and $\#q(D \times D_{\circ}^l)$ are the constant terms. Since the matrix is a~Vandermonde matrix, this system of equations has an unique solution and gives an fpt algorithm $\Amc(\varphi_j, D)$ with access to an oracle for $\mn{AnswerCount}(\{q\}, \class{D})$ that given one of the equivalence classes $\varphi_j$ and a database $D \in \class{D}$ computes the values $\varphi_j(D) = \sum_{p_i \in \varphi_j} c_i \cdot \# p_i(D)$.

	Let $\varphi$ be the equivalence class containing CQ $p$.
	The fpt algorithm promised in the lemma works by invoking the algorithm from Lemma~\ref{lem:recovering-cq-from-semi-counting-equivalent-cq} with the set of queries $\varphi$ and the oracle being the algorithm $\Amc(\varphi, \cdot)$.
\end{proof}

We can now prove the blackbox Theorem~\ref{thm:chenMengelStronger}.

\begin{proof}[Proof of Theorem~\ref{thm:chenMengelStronger}]

    Let $\{p_1, \dots, p_n\} \subseteq \mn{cl}^{\class{D}}_{\mn{CM}}(q)$ be a maximal set of
    pairwise not semi\=/\linebreak[4]counting equivalent CQs such that $p_1 = p$. Let $c_1, \dots, c_n$
    be the sequence of non\=/zero integers such that for every $\Sbf$\=/database $D' \in \class{D}$ we have that
    $\#q(D') = \sum_{i=1}^{n} c_i \#p_i(D')$. By the definition of the Chen-Mengel closure,
    we know that such set of CQs exists and that this sequence of integers %exists and
    is well defined. Now, for every $p_i$ we enumerate $\class{D}$ to find the promised databases $D_{p_i''}$
    and thus the equality\=/free CQs $p_i''$. We can do this, as we have access to a procedure
    that decides counting equality over $\class{D}$.
    Since for all $1 \leq i \leq n$ we have that $p_i$ and $p_i''$ are counting equivalent over $\class{D}$,
    the equality $\#q(D') = \sum_{i=1}^{n} c_i \#p''_i(D')$ hold for every database $D' \in \class{D}$.

    Before we proceed, we observe that by the definition of UCQs, either UCQ $q$ is Boolean and so is every CQ in $\mn{cl}^{\class{D}}_{\mn{CM}}(q)$
    or $q$ has at least one answer variable and so does every CQ in $\mn{cl}^{\class{D}}_{\mn{CM}}(q)$.
    Hence, either all CQs $p_1'', \dots, p_n''$ are Boolean or none is.

    If $q$ has a non\=/empty set of answer variables then we simply apply Lemma~\ref{lem:recovering-cqs-from-cm-closure-cq}.
    Otherwise, $q$ is Boolean and so is every query in $\{p''_1, \dots, p''_n\}$. Hence all those CQs
    are semi\=/counting equivalent and we can apply Lemma~\ref{lem:recovering-cq-from-semi-counting-equivalent-cq} directly.
\end{proof}

\section{Additional Details for Section~\ref{subsect:removeonto}}
\label{app:thirdblack}
%\section{Proofs from trichotomy}

\newcommand{\setV}{V}

For the sake of completeness, this section provides
the construction from \cite{countingTrichotomy} that allows us to remove markings from the query.
In particular, we provide the algorithm promised by the below statement.

%\begin{lem*}[Lemma~\ref{lemma:CQ-to-marked-CQ}]
%    \lemCQtoMarked
%\end{lem*}

\bigskip
\noindent
{\bf Theorem~\ref{lemma:CQ-to-marked-CQ-byChenMengel} \cite{countingTrichotomy}.}%
\textit{\lemCQtoMarked}

%\cMP{why the above shows as C.1?}
\begin{proof}
	We need to exhibit an fpt algorithm that given a CQ $q(\bar{x})$
	and an $\Sbf^{\marked}$\=/database $D$ computes $\#q^{\marked}(D)$.
	The algorithm may ask the oracle for the values $\#q(D')$
	%where $q$ is the CQ $q^{\marked}$ with the markings removed and
	for $D' \in \class{D}$.
	Let $\setV$ % = \bar{x}$.
    be the set of variables in $\bar x$.

	%First, observe that $q$ can be computed in time that depends only on $q^{\marked}$.
    Let $D^{\times}$ be the product of $D_{q^{\marked}}$ and
    $D$, thus in particular $\dom(D^{\times}) =
    \dom(D_{q^{\marked}}) \times \dom(D)$.  Let $D^{\circ} \subseteq
    D^{\times}$ be the subdatabase induced by the set $\{(x,a) \in
    \dom(D^{\times}) \mid x \in \mn{var}(q) \text{ and } R_x(a) \in D\}$.  Clearly, $D^{\circ} \in
    \class{D}$.

	%    For a CQ $p(\bar{x})$ and an $\Sbf^{\marked}$\=/database $B$, l
	Let $\Hmc$ be the set of functions $\fun{g}{\setV}{\dom(D^{\circ})}$
	such that $g$ can be extended to a~homomorphism $\fun{g}{q}{D^{\circ}}$
	such that for all $x \in \setV$, $g(x) = (y,b)$ implies $x=y$.
	Moreover, let $\Hmc'$ be the set of functions $\fun{g}{\setV}{\dom(D^{\circ})}$
	such that $g$ can be extended to a homomorphism $\fun{g}{q}{D^{\circ}}$
	and satisfy %$\{ y \in \setV \mid \exists x \in \setV: g(x) = (y,b) \} = \setV$.
	$\{ x \in \setV \mid \exists b:\ (x,b) \in g(\setV)\} =
        \setV$, that is, all variables from $V$ occurs in the first
        component
        of some element hit by $g$.
	Finally, let $\Imc$ be the set of mappings  $\fun{g}{\setV}{\setV}$
	that can be extended to an automorphism $\fun{g}{q}{q}$.
	\begin{itemize}
		\item[] \textit{Claim 1.} There is a bijection between $q^{\marked}(D)$ and $\Hmc$.
		\item[] \textit{Claim 2.} $|\Hmc'| = |\Imc| \cdot |\Hmc|$.
	\end{itemize}
	Claim~1 follows from the fact that we simulate markings by
        enforcing that for all functions $q \in \Hmc$ and all answer
        variables $x$ we have $g(x) = (x,b)$ for some
        $b \in \dom(D^{\circ})$.  Claim~2 is shown by arguing that
        every function in $\Hmc'$ can be obtained as the composition of a
        function from $\Hmc$ with a permutation of answer variables
        from $\Imc$.  The precise proofs of the above claims can be
        found in~\cite{countingTrichotomy}.

	For $T \subseteq \setV$, let $\Hmc_{T}$ be the set of functions $\fun{g}{\setV}{\dom(D^{\circ})}$
	that can be extended to a homomorphism $\fun{g}{q}{D^{\circ}}$
	and satisfy $\{ y \in \setV \mid \exists x \in \setV: g(x) =
        (y,b) \} \subseteq T$,
        that is, $g$ maps all variables from $V$ to elements that have
        a variable from $T$ in its first component.
	By an inclusion\=/exclusion argument, we get the equation
	\[
	|\Hmc'| = \sum_{T \subseteq \setV} (-1)^{|\setV \setminus T|} |\Hmc_T|.
	\]
	The above equation used together with Claims~1 and~2 gives the following.

	\begin{itemize}
		\item[] \textit{Claim 3.} $\#q^{\marked}(D) = \frac{1}{|\Imc|} \cdot \sum_{T \subseteq \setV} (-1)^{|\setV \setminus T|} |\Hmc_T|$.
	\end{itemize}

	Since $|\Imc|$ can be computed brute-force, we focus on how to compute the values $|\Hmc_T|$ for all $T \subseteq \setV$.
	Fix $T \subseteq \setV$. For $i > 0$, let $\Hmc_{i,T}$ be the
        set of functions $\fun{g}{\setV}{\dom(D^{\circ})}$ such that
        $g$ can be extended to a homomorphism $\fun{g}{q}{D^{\circ}}$
        such that for exactly $i$ variables $x \in \setV$, the first
        component of $g(x)$ is in $T$, that is, $g(x) = (y,b)$ implies
        $y \in T$. Note that $\Hmc_{|\setV|,T} = \Hmc_{T}$.

	Let $D^{\circ}_{j,T}$ be the $\Sbf$\=/database obtained from the database $D^{\circ}$ by cloning
	every element from the set~$\{(x,a) \in \dom(D^{\circ}) \mid x \in T\}$ exactly $j{-}1$ times.
	Note that $D^{\circ}_{1,T} = D^{\circ}$ and $D^{\circ}_{j,T} \in \class{D}$ for all $j > 0$. Moreover, the following holds.
	\begin{itemize}
		\item[] \textit{Claim 4.} For $j > 0$, $\#q(D^{\circ}_{j,T}) = \sum_{i = 0} ^{|\setV|} j^i |\Hmc_{i,T}|$.
	\end{itemize}

	Invoking the equation from Claim~4 for $j= 1, 2, \dots, |V|+1$,
    we obtain a system of $|V|+1$ linear equations where $|\Hmc_{i,T}|$ are the unknowns,
	$j^i$ form the matrix of coefficients, and $\#q(D_{j,T})$ are the constant terms. Since the matrix is a Vandermonde matrix, this system of equations has a unique solution. Moreover, since the class $\class{D}$ is closed under cloning, the values $\#q(D^{\circ}_{j,T})$ can be effectively computed
	using an oracle for $\mn{AnswerCount}(\{q\}, \class{D})$.
	% Hence, for $q^\marked$ and $D^{\circ}$ there is, and can be computed, an algorithm $\textit{solve}_{q^{\marked}, D^{\circ}} (T)$ that given
	% %the query $q^{\marked}$, a set $T \subseteq \setV$, and the database $D^{\circ}$ computes the value $|\Hmc_{T}|$.
	% a set $T \subseteq \setV$ computes the value $|\Hmc_{T}|$.
	Hence there is an algorithm $\textit{solve}(q^{\marked}, D^{\circ},T)$ that given
	%the query $q^{\marked}$, a set $T \subseteq \setV$, and the database $D^{\circ}$ computes the value $|\Hmc_{T}|$.
	$q^{\marked}$, $D^\circ$, and a set $T \subseteq \setV$
        computes the value $|\Hmc_{T}|$.

	\begin{itemize}
		\item[] \textit{Claim 5.}  Algorithm $\textit{solve}(q^{\marked}, D^{\circ},T)$
        %$\textit{solve}(q^{\marked}, T, D^{\circ})$
         is an fpt algorithm with access to an oracle for
         $\mn{AnswerCount}(\{q\}, \class{D})$, the parameter being the
         size of the query $q^{\marked}$ (which dominates the size of $T$).
	\end{itemize}

	All we need to show is that \textit{solve} works in time
        $f(||q^m||) \cdot \textit{poly}(||D^{\circ}||)$
	for some computable function $f$.
	Since we can use the standard algorithm to solve the system of linear equations in time polynomial in the size of the system,
	it is enough to show that the system can be constructed in the desired time.

    The system has $|V|+1$ equation. Since $0 \leq i \leq |V|$ and $1 \leq j \leq |V|+1$,
    the coefficients $j^i$ do not depend on the database and can be constructed in the desired time.

    For the constant terms
	$\#q(D^{\circ}_{j,T})$ we observe that we can compute
        $D^{\circ}_{j,T}$ in time  bounded by $f'(||q^m||) \cdot \textit{poly}(||D^{\circ}||)$
	for some computable function $f'$. Indeed, since $D^{\circ}$ is a subdatabase of the product database $D_{q^{\marked}} {\times} D$,
	it uses no more than $||q^{m}||$ relational symbols. Moreover, every fact in $D^{\circ}$ has arity not greater than $||q^{\marked}||$
	and, thus, cannot give rise to more than $j^{||q^{\marked}||}$ facts. Hence, $||D^{\circ}_{j,T}|| \leq  j^{|q^{\marked}|} ||D^{\circ}||$
	and $D^{\circ}_{j,T}$ can be computed in the desired time by simply enumerating all facts in database $D^{\circ}$.
	In consequence, there is an fpt algorithm with an oracle that constructs the desired system of equations. This ends the
	proof of Claim~5.

        \smallskip
	To conclude, the algorithm that given CQ $q^{\marked}$ and database $D$ computes $\#q^\marked(D)$ works as follows.
	First, we construct the unmarked query $q$ and the
        database~$D^{\circ}$,
        and compute the size of the set $\Imc$.
	Then, for all $T \subseteq \setV$ we invoke $\textit{solve}(q^m, T, D^{\circ})$ to compute
	the values $|\Hmc_{T}|$. Finally, we use the equation from Claim~3 to compute the value $\#q^{\marked}(D)$.

	For the time complexity of the algorithm, we observe that $q$, $|\Imc|$, and $D^{\circ}$ can be easily computed by an fpt algorithm.
	Since the size of the equation in Claim~3 depends only on the size of the query $q$, we have no more than
    a constant number of values $|\Hmc_{T}|$ to compute. By Claim~5 each value $|\Hmc_{T}|$
    can be computed by an fpt algorithm with access to an oracle computing $\#q(D')$ for $D' \in \class{D}$.
    Hence, the overall running time is bounded by $f(||q||) \cdot p(||D||)$ for some computable function $f$ and a polynomial $p$.
%
    % Thus, there is an fpt algorithm that given a CQ $q(\bar{x})$
    % and an $\Sbf^{\marked}$\=/database $D$ computes $\#q^{\marked}(D)$.
    % The algorithm may ask the oracle for the values $\#q(D')$
    %where $q$ is the CQ $q^{\marked}$ with the markings removed and
   % for $D' \in \class{D}$.
\end{proof}

\section{Additional Details for Section~\ref{sect:approx}}

For the proof of Lemma~\ref{prop:firstapprox}, it remains to establish
Points~(A) and~(B) used in the proof given in the main part of the
paper. For the reader's convenience, we repeat the central definitions.

Let $I$ be an instance and $\Tmc$ a set of TGDs. For each fact in
$\mn{ch}_\Tmc(I)$, we want to identify a source fact in $I$.  Start
with setting $\mn{src}(R(\bar c))=R(\bar c)$ for all
$R(\bar c) \in I$. Next assume that $R(\bar c) \in \mn{ch}_\Tmc(I)$
was introduced by a chase step that applies a TGD $T \in \Tmc$ at a
tuple $(\bar d, \bar d')$, and let $R'$ be the relation symbol in the
guard atom in $\mn{body}(T)$. Then we set
$\mn{src}(R(\bar c))=R'(\bar d, \bar d')$ if
$\bar d \cup \bar d' \subseteq \mn{adom}(I)$ and
$\mn{src}(R(\bar c))=\mn{src}(R'(\bar d, \bar d'))$ otherwise. For any
guarded set $X$ of $I$, define $\mn{ch}_\Tmc(I)|^\downarrow_{X}$ to
contain those facts $R(\bar c) \in \mn{ch}_\Tmc(I)$ such that the
constants in $\mn{src}(R(\bar c))$ are exactly those in $X$.

We shall actually consider such subinterpretations not only of the
final result $\mn{ch}_\Tmc(I)$ of the chase, but also of the instances
constructed as part of a chase sequence $I_0,I_1,\dots$ for $I$
with~$\Tmc$. In fact, we can define $I_i|^\downarrow_X$ in exact
analogy with $\mn{ch}_\Tmc(I)|^\downarrow_{X}$, for all $i \geq 0$.
\begin{lem}
  \label{lem:treesarenice}
  Let %$I$ be an instance, \Tmc a set of TGDs,
  $I_0,I_1,\dots$ be a chase
  sequence of $I$ with \Tmc and $i \geq 0$.
  Then
  \begin{enumerate}

  \item[(A)] for all guarded sets $X$ in~$I$, there is a homomorphism
    from $I_i|^\downarrow_{X}$ to $\mn{ch}_\Tmc(I_i|_{X})$ that is the
    identity on all constants in $X$;

  \item[(B)] if $c \in \mn{adom}(I_i)$ is a null and
    $R_1(\bar c_1),R_2(\bar c_2) \in I_i$ such that $c$
    occurs in both $c_1$ and~$c_2$, then
    $\mn{src}(R_1(\bar c_1))=\mn{src}(R_2(\bar c_2))$.

  % Let $I$ be an instance, \Smc a set of TGDs, $I_0,I_1,\dots$ a chase
  % sequence of $I$ with \Smc, and $S$ a guarded set in $I$. Then for
  % every $i \geq 0$ there is a homomorphism from $I_i|^\downarrow_{S}$
  % to $\mn{ch}_S(I_i|_{S})$ that is the identity on all constants in
  % $S$. Moreover, $h_i(c)=h_{i+1}(c)$ for all $c$ with $h_i(c)$
  % defined.
  \end{enumerate}
\end{lem}
\begin{proof} %[Proof (sketch).]
  The proof of both (A) and (B) is by induction on $i$. We only
  present
  the more interesting proof of (A). The induction
  start holds as
  $I_0|^\downarrow_X=I_0|_X \subseteq \mn{ch}_\Tmc(I|_X)$. %  and the
  % induction start of (B) holds since $I_0$ contains no null.

  For the induction step, assume that $I_{i+1}$ was obtained from
  $I_i$ by applying a TGD
  $T= \phi(\bar x,\bar y) \rightarrow \exists \bar z \, \psi(\bar
  x,\bar z)$ at a tuple $(\bar c, \bar c')$. Let $R$ be the relation
  symbol used in a guard atom of $\phi$, and let $Z$ be the constants
  in $\mn{src}(R(\bar c, \bar c'))$.

  Now consider any guarded set $X$ in~$I$.  If $X \neq Z$, then by
  definition of $\cdot|^\downarrow_X$ in terms of $\mn{src}$, we must
  have $I_{i+1}|^\downarrow_{X} =I_{i}|^\downarrow_{X}$ and it
  suffices to use the induction hypothesis. Thus assume that $X=Z$.
%
  % First assume that $X \neq Z$. Then by definition of
  % $\cdot|^\downarrow_X$ in terms of $\mn{src}$, all facts in
  % $I_{i+1}|^\downarrow_{X}\setminus I_{i}|^\downarrow_{X}$ contain
  % only constants in $\bar c$, but no nulls.\footnote{For example,
  %   assume that $R(a,b),A(a) \in I \subseteq I_i$ and that the TGD
  %   $A(x) \rightarrow B(x)$ is applied at $a$. For $X=\{a,b\}$, we
  %   have $X \neq Z=\{a\}$, and
  %   $I_{i+1}|^\downarrow_{X}\setminus I_{i}|^\downarrow_{X} = \{ B(a)
  %   \}$.}  By induction hypothesis, there is a homomorphism $h_i$
  % from $I_i|^\downarrow_{X}$ to $\mn{ch}_\Tmc(I_i|_{X})$ that is the
  % identity on all constants in $X$. Clearly, $h_i$ is also a
  % homomorphism from $I_{i+1}|^\downarrow_{X}$ to
  % $\mn{ch}_\Tmc(I_{i+1}|_{X})$.
%
  %Now assume that $X=\bar c \cup \bar c'$.
  Then clearly
  $I_{i+1}|^\downarrow_{X}\setminus I_{i}|^\downarrow_{X} = \psi(\bar c,
  \bar c'')$ where $\bar c''$ consists of nulls that do not occur
  in $I_i$.  By induction hypothesis, there is a homomorphism $h_i$
  from $I_i|^\downarrow_{X}$ to $\mn{ch}_\Tmc(I_i|_{X})$ that is the
  identity on all constants in $X$. Applicability of $T$ at
  $(\bar c, \bar c')$ implies $\phi(\bar c,\bar c') \subseteq I_i$ and
  thus
  $\phi(h_i(\bar c),h_i(\bar c')) \subseteq \mn{ch}_\Tmc(I_i|_{X})$.
  It follows $T$ has been applied at $(h(\bar c),h(\bar c'))$ in (any
  fair chase sequence that produces) $\mn{ch}_\Tmc(I_i|_{X})$.  As a
  consequence, there are constants $\bar d$ such that $\psi(h(\bar
  c),\bar d) \subseteq \mn{ch}_\Tmc(I_i|_{X})$. We extend $h_i$ to $h_{i+1}$
  so that $h_{i+1}(\bar c'')=\bar d$. Clearly, $h_{i+1}$ is a homomorphism from
  $I_{i+1}|^\downarrow_{X}$ to $\mn{ch}_\Tmc(I_{i+1}|_{X})$.
\end{proof}

We next prove the part of Theorem~\ref{thm:CTWSSjointapprox} that is
concerned with contract treewidth. % For simplicity, we assume that
% (U)CQs in OMQs do not contain equality atoms.
The proof uses minors.
We recall that an undirected graph $G$ is a \emph{minor} of an
undirected graph $H$ if $G$ can be obtain from $H$ by contracting
edges and then taking a subgraph.
\begin{lem}
    \label{thm:CTWapprox}
    Let $(\Omc,\Sbf,q) \in (\class{G},\class{UCQ})$ be an OMQ and $k \geq
    1$. Then $Q^{\text{CTW}}_k$ is a CTW$_k$-approximation of $Q$.
\end{lem}
%\CTWapprox*
%
\begin{proof}
    % {\color{blue}Have to assume original query has no equalities or
    %   equalities contribute to Gaifman graph}
  Let $Q(\bar x)=(\Omc,\Sbf,q) \in (\class{G},\class{UCQ})$.  By
  construction of $Q_{k}^{CTW} = (\Omc, \Sbf, q^{\text{CTW}}_k)$, it
  is clear that Points~1 and~2 of the definition of
  CTW$_k$-approximations are satisfied. It thus remains to establish
  Point~3.

    Thus let $P(\bar x) = (\Omc',\Sbf,p) \in (\class{G},\class{UCQ})$
    such that $P \subseteq Q$ with $p$ of contract treewidth at most
    $k$.  We have to show that $P \subseteq Q_{k}^{\text{CTW}}$. Let
    $D$ be an \Sbf-database and let $\bar c \in P(D)$.  Thus there is
    a homomorphism $h$ from some CQ $p'(\bar x)$ in $p$ to
    $\mn{ch}_{\Omc'}(D)$ such that $h(\bar x )=\bar c$.

    The general strategy of the proof is the same as in the proof of
    Lemma~\ref{thm:SSapprox}, but there is an additional complication.
    Ideally, we would like to use the homomorphism $h$ and the fact
    that $p'$ has contract treewidth at most~$k$ to construct from
    $D$ a database $D''$ such that the following three conditions are
    satisfied:
    \begin{enumerate}
    \item[(a)] $\bar c \in P(D'')$;
    \item[(b)] the contract treewidth of $(D'',\bar c)$ is at most $k$;
    \item[(c)] there is a homorphism from $D''$ to $D$ that is the
      identity on $\bar c$.
    \end{enumerate}
    Let us briefly argue how this helps to prove
    Lemma~\ref{thm:CTWapprox}.  Point~(a) and $P \subseteq Q$ imply
    that there is a homomorphism~$g$ from some CQ $q'(\bar x)$ in $q$
    to $\mn{ch}_\Omc(D'')$ with $g(\bar x)=\bar c$.  We obtain a CQ
    $\widehat q(\bar x)$ from $q'$ by identifying variables, achieving
    that for all distinct variables $y_1,y_2$ in~$\widehat q$,
    $g(y_1)=g(y_2)$ implies that
    $y_1=y_2 \in \widehat q$. Point~(b) then
    implies that $\widehat q(\bar x)$ is a CQ in~$q^{\text{CTW}}_k$
    and we may use Point~(c) to show that $\bar c \in Q_{k}^{CTW}(D)$,
    as required.

    The additional complication is as follows. To ensure Point~(b), it
    is necessary that the homomorphism $h$ from $p'$ to $D$ that we
    use in the construction of $D''$ is `as injective as
    possible'. Intuitively, this is because an injective homomorphism
    provides a much closer link between $p'$ and $D$, and such a close
    link is needed to transfer the bound on contract treewidth from
    $p'$ to $D''$. In particular, non-injectivities on answer
    variables must be avoided as much as possible. To address this
    issue, we first construct from $(D,\bar c)$ another pointed
    database $(D',\bar c')$ and a homomorphism $h'$ from $D'$ to $D$
    with $h(\bar c')=\bar c$, breaking as many non-injectivities on
    answer variables as possible. We then construct $D''$
    starting from $D'$ and $h'$ rather than from $D$ and $h$.

    We now describe the construction of $(D',\bar c')$ and $h'$ in
    detail.
    Let `$\sim$' denote the smallest
    equivalence relation on the variables in $\bar x$ such that $x_1 \sim x_2$ whenever
    $p'$ has an $\bar x$-component $S$ such that $x_1,x_2 \in
    S$. % \footnote{We remark that $x_1,x_2$ being in the same $\bar
       %  x$-component of $p'$ implies $x_1 \sim x_2$ implies that $x_1,x_2$
       %  are in the same maximal connected component of $p'$, but the
       %  converse implications fail.}
    We construct  $(D',\bar c')$
%      We first construct an \Sbf-database
 %     $D'$ and a tuple $\bar c' \in \mn{adom}(D')$
      such that:
    \begin{enumerate}

        %  \item $h'(y) \notin \bar c$ for all quantified variables $y$ of $p'$;

        \item there is a homomorphism $h'$ from $p'$ to
        $\mn{ch}_{\Omc'}(D')$ such that $h'(\bar x) = \bar c'$ and
        $h'(x_1)=h'(x_2)$, with $x_1,x_2 \in \bar{x}$, implies
        $x_1 \sim x_2$;
        % one of the following holds:
        % %
        % \begin{enumerate}

        % \item $x_1=x_2 \in p'$;

        % \item $G_{p'}$ contains a path $y_1,\dots,y_n$ with $y_1=x_1$,
        %   $y_n=x_2$, and $h(y_i) \notin \mn{adom}(D')$;

        % \end{enumerate}

        \item there is a homomorphism from $D'$ to $D$ that maps $\bar
        c'$ to $\bar c$.

    \end{enumerate}
    Informally, the condition on $h'$ in Point~1 says that $h'$ avoids
    non-injectivities on answer variables as much as possible.
    %
    % To improve readability of the construction, we assume that there are
    % no equality atoms in $p'$, and thus the second condition on $h'$ in
    % Point~1 simply says that all constants in $\bar a'$ are
    % distinct. Our proof, however, easily extends also to the case with
    % equality atoms.
    %
    % Let
    %   $p_1,\dots,p_m$ be the maximal connected components of $p'$.
    Let $\Cmc_1,\dots,\Cmc_m$ be the equivalence classes of `$\sim$'.
    %   For $c \in \mn{adom}(D)$, set
    % % $\delta(c) = \{ c \} \cup \{ c^i \mid
    % %   \ \exists x \in \Cmc_i: h(x)=c\}$.
    % $\delta(c) = \{ c^1 \} \cup \{ c^i \mid
    %   \ \exists x \in \Cmc_i: h(x)=c\}$.
    Define
    $$
    D'=\{ R(c^{i_1}_1,\dots,c^{i_n}_n) \mid
    R(c_1,\dots,c_n) \in D \text{ and } 1 \leq i_1,\dots,i_n \leq m \}.
    $$
    %
    % If $\bar x = x_1,\dots,x_n \in \Cmc_{i_1} \times \cdots \times
    % \Cmc_{i_n}$ and $\bar a = a_1 \cdots a_n$, set $\bar a' = a_1^{i_1}
    % \cdots a_n^{i_n}$.
    It is easy to see that Point~2 is indeed satisfied as long as we
    construct $\bar c'$ from $\bar c$ by replacing each component $c$
    with some $c^i$, $i \leq i \leq m$. To define the
    homomorphism $h'$ required by Point~1, we need two preliminaries.

    First, by definition of `$\sim$' we find, for each quantified
    $y \in \mn{var}(p')$, at most one $i$ such that there is an answer
    variable $x \in \Cmc_i$ that is reachable from $y$ in $G_{p'}$
    without passing an answer variable. Set $\rho(y)=i$ and
    $\rho(y)=1$ if there is no such $i$. Moreover, for each variable
    $x$ in $\bar x$, set $\rho(x)=i$ if $x \in \Cmc_i$. It is easy to
    see that when quantified variables $y_1,y_2$ co-occur in an atom in
    $p'$, then $\rho(y_1)=\rho(y_2)$.

    Second, by construction of $D'$ we find for each fact
    $\alpha=R(c^{i_1}_1,\dots,c^{i_n}_n) \in D'$ a homomorphism
    $g_\alpha$ from $D$ to $D'$ such that $g_\alpha(c_j)= c^{i_j}_j$
    for $1 \leq j \leq n$. We can extend $g_\alpha$ to a homomorphism
    from $\mn{ch}_{\Omc'}(D)$ to $\mn{ch}_{\Omc'}(D')$. We now define
    $h'$ as follows:
    \begin{itemize}

        \item \label{item:ctw-approx-1} if $x \in \mn{var}(p_i)$ and $h(x) \in \mn{adom}(D)$, then $h'(x)=h(x)^{\rho(x)}$;

        \item \label{item:ctw-approx-2} if $x \in \mn{var}(p_i)$ and
          $h(x)\notin \mn{adom}(D)$ is in the tree-like structure that
          the chase has generated below fact
          $R(c_1,\dots,c_n)$,\footnote{This can be made precise in the
            same way as in the proof of
            Lemma~\ref{prop:firstapprox}. We prefer to remain on the
            intuitive level here to not distract from the main proof.}
          then $h'(y)=g_\alpha(h(y))$,
          $\alpha=R(c^{\rho(x)}_1,\dots,c^{\rho(x)}_n)$.

    \end{itemize}
    Moreover, set $\bar c' = h'(\bar c)$. We argue that $h'$ is indeed
    a homomorphism from $p'$ to $\mn{ch}_{\Omc'}(D')$. Let $R(\bar y)$
    be an atom in $p'$. % with $\bar y = y_1,\dots,y_n$.
    First assume that $h(y)\in \mn{adom}(D)$ for all variables $y$ in
    $\bar y$. Let $\alpha = R(h'(\bar y))$. Then $g_\alpha$ is a
    homomorphism from $\mn{ch}_{\Omc'}(D)$ to $\mn{ch}_{\Omc'}(D')$
    with $g_\alpha(h(\bar y))=h'(\bar y)$. Together with
    $R(h(y)) \in \mn{ch}_\Omc(D)$, this yields
    $R(h'(y)) \in \mn{ch}_\Omc(D')$ as required.

    Now
    assume that $\bar y$ contains at least one variable $y$ with
    $h(y) \notin \mn{adom}(D)$. Then $y$ is a quantified variable.  By
    definition of $\rho$, $\rho(y)$ must be identical for all
    quantified variables $y$ in $\bar y$, and it must also be
    identical to $\rho(x)$ for all answer variables in $\bar y$.  This
    means that $h'(y)$ is defined based on the same homomorphism
    $g_\alpha$ for all variables $y$ in $\bar y$, and in particular
    $g_\alpha(h(\bar y))=h'(\bar y)$. From
    $R(h(y)) \in \mn{ch}_\Omc(D)$, we again obtain
    $R(h'(y)) \in \mn{ch}_\Omc(D')$ as required.

    \medskip
    \noindent
    We next construct an \Sbf-database $D''$ that satisfies Points~(a)
    to~(c) above, in a slightly modified form:
    %
    % and a  homomorphism $h''$ from $p$ to
    % $\mn{ch}_\Omc(D'')$ such that $h'(\bar x) = \bar a$ and
    %
    \begin{enumerate}
        \setcounter{enumi}{2}
        \item \label{item:ctw-approx-3} there is a homomorphism $h''$ from $p'$ to
        $\mn{ch}_{\Omc'}(D'')$ that maps $\bar x$ to $\bar c'$;

        \item \label{item:ctw-approx-4} the contract treewidth of $(D'',\bar c')$ is at most $k$;

        \item \label{item:ctw-approx-5} there is a homomorphism from $D''$ to $D'$ that is the
        identity on $\bar c'$.

    \end{enumerate}
    Set
    $$\Gamma=\{ h'(\bar y \cap \bar x) \mid R(\bar y) \in p' \} \cup \{
    h'(S \cap \bar x) \mid S \text{ $\bar x$-component of } p' \}$$
    and note the tight connection to the definition of contracts, in
    which edges step from atoms (first set in the definition of
    $\Gamma$) and from $\bar x$-components (second set).  For every
    $S \in \Gamma$, let $D_{S}$ be the database obtained from $D'$ by
    renaming every constant $c \notin S$ to $c^{S}$. Now define
    $$D'' = \bigcup_{S \in \Gamma} D_S.$$

    It is easy to see that Point~\ref{item:ctw-approx-5} is
    satisfied. We now argue that Points~3 and~4 also hold. For
    Point~3, we have to construct a homomorphism $h''$ from $p'$ to
    $\mn{ch}_{\Omc'}(D'')$ with $h''(\bar x)=\bar c'$. Start with
    setting $h''(x)=h'(x)$ for all $x\in \bar x$, and thus
    $h''(\bar x)=\bar c'$ as required.  It remains to define $h''$ for
    the quantified variables in $p'$. We do this per
    $\bar x$-component. Thus let $S$ be an $\bar x$-component of $p'$.
    Then $h(S \cap \bar x) \in \Gamma$.  It is not hard to prove that
    $\mn{ch}_{\Omc'}(D'') = \bigcup_{S \in \Gamma}
    \mn{ch}_{\Omc'}(D_S)$. Clearly, there is a homomorphism $h_S$ from
    $D'$ to $D_S$ with $h_S(c)=c$ for all $c \in S$, and $h_S$ can be
    extended to a homomorphism from $\mn{ch}_{\Omc'}(D')$ to
    $\mn{ch}_{\Omc'}(D_S)$.  Define $h''(y)=h_S \circ h'(y)$ for all
    $y \in S$. It can be verified that $h''$ is indeed a homomorphism.
    Details are omitted.

    To prove Point~4, we show that
    $\mn{contract}(G_{D''})$%\footnote{What we mean with
        %$\mn{contract}(G_{D''})$ is $\mn{contract}(G_{q_D})$ where $q_D$
        %is $D$ viewed as a CQ with answer variables~$\bar c$.}
    is a minor of $\mn{contract}(G_{p'})$.
    Here, by $\mn{contract}(G_{D''})$ we mean $\mn{contract}(G_{q_D})$ where $q_D$
    is $D$ viewed as a CQ with answer variables~$\bar c$.
    To this end, we first note that
    \begin{enumerate}
        \item[($*$)] if $\{a,b\}$ is an edge in $\mn{contract}(G_{D''})$,
        then there are $x_a,x_b$ such that $h'(x_a)=a$, $h'(x_b)=b$, and
        $\{x_a,x_b\}$ is an edge in $\mn{contract}(G_{p'})$.

    \end{enumerate}
    Thus let $\{a,b\}$ be an edge in $\mn{contract}(G_{D''})$. By
    construction of $D''$, this implies that there is an $S \in \Gamma$
    with $\{a,b\} \subseteq S$.  Consequently, $p'$ contains an atom
    $R(\bar y)$ such that there are $x_a,x_b \in \bar y \cap \bar x$
    with $h'(x_a)=a$ and $h'(x_b)=b$, or $p'$ has an $\bar x$-component
    $S$ such that there are $x_a,x_b \in S \cap \bar x$ with $h'(x_a)=a$
    and $h'(x_b)=b$. In both cases, $\{x_a,x_b\}$ is an edge in
    $\mn{contract}(G_{p'})$. At this point, we are done if $h'$ is
    injective on $\bar x$ because then $\mn{contract}(G_{D''})$ is
    a subgraph of $\mn{contract}(G_{p'})$. But this need not be the case. However,
    Point~1 above implies that if $h'(x)=h'(x')$, $x,x' \in \bar x$,
    then we find a sequence of variables $x_1,\dots,x_n$ from $\bar x$
    such that $x_1=x$, $x_n=x'$, and $x_i$ is connected to $x_{i+1}$ in
    $G_{p'}$ via a path whose non-end nodes are all from outside $\bar
    x$. Thus, $\mn{contract}(G_{p'})$ contains edges
    $\{x_1,x_2\},\dots,\{x_{n-1},x_n\}$. If we contract all these edges,
    we obtain a minor $G$ of $\mn{contract}(G_{p'})$ that $h'$ maps
    injectively into $\mn{contract}(G_{D''})$ and $(*)$ still
    holds. Thus $\mn{contract}(G_{D''})$ is a subgraph of $G$. Moreover,
    the treewidth of $G$ is not larger than that of
    $\mn{contract}(G_{p'})$ as the latter contains $G$ as a minor.

    \medskip

    Now back to the main proof.  From $\bar c' \in P(D'')$, we obtain
    $\bar c' \in Q(D'')$. Consequently, there is a homomorphism $g$
    from some CQ $q'$ in $q$ to $\mn{ch}_{\Omc}(D'')$ such that
    $g(\bar x)=\bar c'$. Let $\widehat q$ denote the collapsing of
    $q'$ that is obtained by identifying $y_1$ and $y_2$ whenever
    $g(y_1)=g(y_2)$ with at least one of $y_1,y_2$ a quantified
    variable and adding $x_1=x_2$ whenever $g(x_1)=g(x_2)$ and
    $x_1,x_2$ are both answer variables. Note that $g$ is an injective
    homomorphism from $\widehat q$ to $\mn{ch}_\Omc(D'')$, that is, if
    $g(y_1)=g(y_2)$ then $y_1$ and $y_2$ are answer variables and
    $y_1=y_2 \in \widehat q$.  In what follows, we use this fact to
    show that the contract treewidth of $\widehat q$ is at most
    $k$. This finishes the proof as it means that $\widehat q$ is a CQ
    in $q^{\text{CTW}}_k$, and thus $g$ witnesses that
    $\bar c' \in Q^{\text{CTW}}_k(D'')$. Points~2 and~5 above yield a
    homomorphism $g'$ from $D''$ to $D$ such that
    $g'(\bar c') = \bar c$. We can extend $g'$ to a homomorphism from
    $\mn{ch}_\Omc(D'')$ to $\mn{ch}_\Omc(D)$. It is thus easy to see
    that $\bar c \in Q^{\text{CTW}}_k(D)$.

    \medskip Since $g$ is injective and by
    definition of $G_{\widehat q}$, it suffices to show that if $\{ x_1,x_2\}$
    is an edge in $\mn{contract}(G_{\widehat q})$, then
    $\{ g(x_1),g(x_2)\}$ is an edge in
    $\mn{contract}(G_{D''})$.
    Thus let
    $\{ x_1,x_2\}$ be an edge in $\mn{contract}(G_{\widehat q})$.

    First assume that $\{ x_1,x_2\}$ is an edge in the restriction of
    $G_{\widehat q}$ to nodes~$\bar x$. Then $\widehat q$ contains
    an atom $R(\bar z)$ such that $x_1,x_2 \in \bar z$. This implies
    that $g(x_1),g(x_2)$ occur in the fact
    $R(g(\bar z)) \in \mn{ch}_\Omc(D'')$. Since $\Omc \in \class{G}$, $D''$ must contain a fact in which both of
    $g(x_1),g(x_2)$ occur. Thus, $\{ g(x_1),g(x_2)\}$ is an edge in
    $\mn{contract}(G_{D''})$.

    Now assume that $x_1,x_2$ co-occur in some $\bar x$-component of
    $G_{\widehat q}$. Then $G_{\widehat q}$ contains a path
    $z_1,\dots,z_n$ such that $z_1=x_1$, $z_n =x_n$, and
    $z_2,\dots,z_{n-1}$ are quantified variables.  Let
    $z_{i_1},\dots,z_{i_k}$ denote the subsequence of $z_1,\dots,z_n$
    obtained by dropping all $z_i$ such that $g(z_i)$ is a constant that
    was introduced by the chase. Since $\Omc \in \class{G}$,
    $\{ z_{i_j}, z_{i_{j+1}} \}$ is an edge in $G_{D''}$ for
    $1 \leq j < k$.  We know that $g(z_{i_j})\notin\bar c'$ for
    $2 \leq j< k$ since if $g(z_{i_j}) =a \in \bar c'$, then $z_{i_j}$
    was identified with some $x \in \bar{x}$ such that $g(x)=a$ during
    the construction of $\widehat q$, in contrary to the fact that
    $z_{i_j}$ is a quantified variable.  It follows that
    $\{ g(x_1),g(x_2)\}$ is an edge in $\mn{contract}(G_{D''})$.
\end{proof}

{\bf Acknowledgements.} This research was funded by the
DFG project QTEC. We
thank the anonymous reviewers for useful comments.

\end{document}